\documentclass[11pt,a4paper,USenglish]{article}
\pdfoutput=1
\usepackage[in]{fullpage}

\usepackage{microtype}
\usepackage{boxedminipage}
\usepackage{mathrsfs}

\usepackage{soulutf8}
\usepackage{aurical}
\usepackage[T1]{fontenc}
\usepackage{xspace}
\usepackage{esvect}
\usepackage{amssymb,mathtools,amsmath}
\usepackage{mathrsfs}
\usepackage[c2]{optidef}
\usepackage{datetime}
\usepackage{stmaryrd}
\usepackage{tcolorbox}
\usepackage{wasysym, fancybox}
\usepackage{dsfont}
\usepackage{eqparbox}
\usepackage{comment}
\usepackage{enumitem,amsthm}
\usepackage{hyperref}

\usepackage[normalem]{ulem} 

\usepackage{tikz-cd}

\graphicspath{{./Figures/}}

\newtheorem{theorem}{Theorem}
\newtheorem{lemma}[theorem]{Lemma}
\newtheorem{corollary}[theorem]{Corollary}
\newtheorem{proposition}[theorem]{Proposition}

\theoremstyle{definition}
\newtheorem{definition}[theorem]{Definition}

\theoremstyle{remark}
\newtheorem{remark}[theorem]{Remark}

\definecolor{lightblue}{rgb}{.80,.85,1}
\definecolor{khaki}{rgb}{.94,.90,.55}
\definecolor{salmon}{rgb}{1.0,0.62,0.47}
\definecolor{grey}{rgb}{.5,.5,.5}
\definecolor{darkorchid}{rgb}{0.60,0.20,0.80}
\definecolor{venitian}{rgb}{0.8,0.12,0.2}

\newcommand{\defunder}[1]{\underset{\text{def.}}{#1} \:}

\newcommand{\refitem}[1]{\textcolor{darkgray}{\sffamily\bfseries\upshape\mathversion{bold}\ref{#1}}}

\newcommand{\refhypo}[1]{\textcolor{darkgray}{\sffamily\bfseries\upshape\mathversion{bold}\ref{#1}}}

\DeclarePairedDelimiter{\thicknormaux}{\bracevert}{\bracevert}

\NewDocumentCommand{\thicknorm}{som}{%
  \IfBooleanTF{#1}
    {\thicknormaux*{#3}}
    {\IfNoValueTF{#2}
       {\thicknormaux*{\mkern-5mu\vphantom{dq}#3\mkern-5mu}}
       {\thicknormaux[#2]{#3}}%
    }%
    }

\DeclarePairedDelimiter\abs{\lvert}{\rvert}

\makeatletter
\let\oldabs\abs
\def\abs{\@ifstar{\oldabs}{\oldabs*}}

\makeatletter
\def\widebreve{\mathpalette\wide@breve}
\def\wide@breve#1#2{\sbox\z@{$#1#2$}%
     \mathop{\vbox{\m@th\ialign{##\crcr
\kern0.08em\brevefill#1{0.8\wd\z@}\crcr\noalign{\nointerlineskip}%
                    $\hss#1#2\hss$\crcr}}}\limits}
\def\brevefill#1#2{$\m@th\sbox\tw@{$#1($}%
  \hss\resizebox{#2}{\wd\tw@}{\rotatebox[origin=c]{90}{\upshape(}}\hss$}
\makeatletter

\DeclareFontFamily{U}{matha}{\hyphenchar\font45}
\DeclareFontShape{U}{matha}{m}{n}{
      <5> <6> <7> <8> <9> <10> gen * matha
      <10.95> matha10 <12> <14.4> <17.28> <20.74> <24.88> matha12
      }{}
\DeclareSymbolFont{matha}{U}{matha}{m}{n}
\DeclareFontSubstitution{U}{matha}{m}{n}

\DeclareFontFamily{U}{mathx}{\hyphenchar\font45}
\DeclareFontShape{U}{mathx}{m}{n}{
      <5> <6> <7> <8> <9> <10>
      <10.95> <12> <14.4> <17.28> <20.74> <24.88>
      mathx10
      }{}
\DeclareSymbolFont{mathx}{U}{mathx}{m}{n}
\DeclareFontSubstitution{U}{mathx}{m}{n}

\DeclareMathDelimiter{\vvvert}{0}{matha}{"7E}{mathx}{"17}

\newcounter{defproblem}
{
\addtocounter{equation}{-1}
\refstepcounter{defproblem}
\renewcommand\theequation{P$_\thedefproblem$[\ensuremath{#1}]}
\begin{equation}
}%
{
\end{equation}
}

\newcommand{\Rspace}{\ensuremath{\mathbb{R}}\xspace}

\newcommand{\M}{\ensuremath{\mathcal{M}}\xspace}

\newcommand{\D}{\ensuremath{\mathcal{D}}\xspace}

\newcommand{\Hcal}{\ensuremath{\mathcal{H}}\xspace}
\newcommand{\Hscr}{\ensuremath{\mathscr{H}}\xspace}

\newcommand{\Dim}{N}
\newcommand{\Rrel}{\ensuremath{\mathscr{R}}\xspace}
\newcommand{\Srel}{\ensuremath{\mathscr{S}}\xspace}
\newcommand{\Qrel}{\ensuremath{\mathscr{Q}}\xspace}

\newcommand{\styleitem}[1]{\textcolor{darkgray}{\sffamily\bfseries\upshape\mathversion{bold}#1}\xspace}
\newcommand{\ita}{\styleitem{(a)}}
\newcommand{\itb}[1]{\styleitem{(b$_{#1}$)}}

\newcommand{\Conv}[1]{\operatorname{conv}#1}
\newcommand{\Aff}[1]{\operatorname{aff}#1}

\newcommand{\Reach}[1]{\operatorname{reach}#1}
\newcommand{\reach}{\ensuremath{\mathcal{R}}\xspace}
\newcommand{\Offset}[2]{#1^{\oplus #2}}
\newcommand{\Diam}[1]{\operatorname{Diam}(#1)}

\newcommand{\Sep}[1]{\operatorname{separation}(#1)}

\newcommand{\height}[1]{\operatorname{height}(#1)}
\newcommand{\Height}[2]{\operatorname{height}(#1,#2)}
\newcommand{\Protection}[2]{\operatorname{protection}(#1,#2)}
\newcommand{\MA}[1]{\operatorname{axis}(#1)}

\newcommand{\VectorSpace}[1]{\operatorname{Vec}(#1)}

\newcommand{\Tangent}[2]{\mathbf{T}_{#1}#2}
\newcommand{\Normal}[2]{\mathbf{N}_{#1}#2}

\def\restrict#1{\raise-.5ex\hbox{\ensuremath|}_{#1}}

\newcommand\restr[2]{{
  \left.\kern-\nulldelimiterspace 
  #1 
  \vphantom{\big|} 
  \right|_{#2} 
}}

\newcommand \Above[2]{\genfrac{}{}{0pt}{0}{#1}{#2}}

\newcommand{\Domain}[1]{\operatorname{Domain}(#1)}
\newcommand{\Range}[1]{\operatorname{Range}(#1)}

\newcommand{\prTangent}[2]{\pi_{\Tangent {#1} \M}(#2)}

\newcommand{\Vertexset}[1]  {\operatorname{Vert}{#1}}

\newcommand{\Shadow}[1]{\boldsymbol{\lvert} #1 \boldsymbol{\rvert}}
\newcommand{\relint}[1]{\operatorname{relint}(#1)}
\newcommand{\US}[1]{\boldsymbol{\lvert} #1 \boldsymbol{\rvert}}

\newcommand{\Del}[1]{\operatorname{Del}(#1)}

\newcommand{\FlatDel}[2]{\operatorname{FlatDel}_{\M}(#1,#2)}

\newcommand{\Prestar}[2]{\operatorname{Prestar}_{P,\M}(#1,#2)}
\newcommand{\Delstar}[2]{\operatorname{Star}_{P,\M}(#1,#2)}

\title{Flat Delaunay Complexes for Homeomorphic Manifold Reconstruction}

\author{Dominique Attali\footnote{Univ. Grenoble Alpes, CNRS,
  GIPSA-lab, Grenoble, France. \texttt{Dominique.Attali@grenoble-inp.fr}}
\and
André Lieutier\footnote{Dassault systèmes, Aix-en-Provence,
France. \texttt{andre.lieutier@3ds.com}}
}

\begin{document}

\maketitle

\begin{abstract}
  Given a smooth submanifold of the Euclidean space, a finite point cloud
  and a scale parameter, we introduce a construction which we call the
  {\em flat Delaunay complex} (FDC). This is a variant of the
  tangential Delaunay complex (TDC) introduced by Boissonnat et
  al. \cite{boissonnat2014manifold,boissonnat2018geometric}. Building
  on their work, we provide a short and direct proof that when the
  point cloud samples sufficiently nicely the submanifold and is
  sufficiently safe (a notion which we define in the paper), our
  construction is homeomorphic to the submanifold. Because the proof
  works even when data points are noisy, this allows us to propose a
  perturbation scheme that takes as input a point cloud sufficiently
  nice and returns a point cloud which in addition is sufficiently
  safe. Equally importantly, our construction provides the framework
  underlying a variational formulation of the reconstruction problem
  which we present in a companion paper \cite{socg22}.
\end{abstract}

\section{Introduction}


In this paper, we consider a variant of the {\em tangential Delaunay
  complex} for triangulating smooth $d$-dimensional
  submanifolds of $\Rspace^\Dim$ that we call the {\em flat Delaunay
  complex}.

\subparagraph{Manifold reconstruction and learning.}
In many practical situations, the shape of interest is only known
through a finite set of data points. Given these data points as input,
it is then natural to try to construct a triangulation of the shape,
that is, a set of simplices whose union is homeomorphic to the
shape. This problem has given rise to many research works in the
computational geometry community, motivated by applications to 3D
model reconstruction and manifold learning; see for instance
\cite{edelsbrunner1994triangulating,amenta02:_simpl_algor,dey2007curve,boissonnat2014manifold,boissonnat2006provably,khoury2021restricted}
to mention a few of them.

In manifold learning, data sets typically live in high dimensional
spaces but are assumed to be distributed near unknown relatively low
dimensional smooth manifolds. In this context, reconstruction
algorithms have to deal efficiently with manifolds having an arbitrary
codimension and, most importantly, should have a complexity which is
only polynomial in the ambient dimension.
The  tangential Delaunay complex of Boissonnat et al. \cite{boissonnat2014manifold} and
\cite[section 8.2]{boissonnat2018geometric} enjoys this polynomial complexity with respect to 
the ambient dimension.

\subparagraph{Tangential Delaunay complex (TDC).}
Consider a set of data points $P$ that sample a smooth $d$-submanifold
\M of $\Rspace^\Dim$. The idea of the TDC is that, given as input $P$
together with the tangent spaces {$\Tangent p \M$} for each $p \in P$,
it is possible to triangulate $\M$ locally around a point $p \in P$ by
considering the Delaunay complex of $P$ restricted to $\Tangent p \M$
and collecting Delaunay simplices incident to $p$; see
\cite{boissonnat2014manifold,boissonnat2018geometric}. In those papers, the resulting collection of simplices
is called the {\em star} of $p$ and its
computation is made efficient by observing that restricting the
Delaunay complex $P$ to the tangent space $\Tangent p \M$ boils down
to projecting points of $P$ onto $\Tangent p \M$ and computing a
$d$-dimensional weighted Delaunay complex of the projected points, the
weight of the projection of $q \in P$ being the squared distance
between $q$ and $\Tangent p \M$. The tangential Delaunay complex (TDC)
is defined as the union of the stars of all points in $P$.



The stars in the TDC are said to be consistent if any simplex in the TDC
belongs to the star of each of its vertices.  The authors prove in
particular that (1) when the data set is sufficiently dense with
respect to the reach of \M, a weight assignment -- through Moser
Tardos Algorithm \cite{moser2010constructive} -- makes the stars
consistent and (2) that when the stars are consistent, the TDC is a
triangulation of the manifold, more precisely, the TDC is embedded and
the projection onto $\M$ restricted to the TDC is an homeomorphism.

\subparagraph{Our contributions.} We propose a construction called the
{\em flat Delaunay complex} (FDC) that exhibits the same behavior as
the TDC described above.  First, it has a geometric characterization
of simplices analog to that of the TDC, the only difference being that
around each point $p$, we replace the computation of the weighted
Delaunay complex by that of an unweighted one and, as a counterpart,
restrict computations inside a ball of radius $\rho$ around
$p$. While, from an application perspective, our FDC would lead to
similar practical algorithms than the TDC, we claim that it brings
significant theoretical contributions.


First, while the criterion of star consistency in TDC is simple and
elegant, the proof of homeomorphism for TDC, once this consistency is
assumed, prove to be rather involved, requiring in particular the use
of a lemma by Whitney about the projection of oriented PL
pseudo-manifolds; see \cite[Lemma 5.14]{boissonnat2014manifold},
\cite{boissonnat2020local} and \cite[Lemma 15a, Appendix
  II]{whitney2005geometric}.  Our construction defines instead what we
call {\em prestars} everywhere in space, not merely at the points of the
data set and, for each $d$-simplex $\sigma$ in the FDC, requires these
prestars to agree at every pair of points in $\Conv\sigma$ and not
merely at the vertices of $\sigma$. This allows us to give a
more direct and, in our opinion, more insightful proof for the
homeomorphism.

Second, as in the proof of correctness for TDC, a crucial ingredient
consists in quantifying some metric distorsion between projections on
various affine $d$-spaces.  By considering metric distorsions in the
context of relations instead of maps (as in Gromov-Hausdorff distance
definition \cite[Section 5.30]{bridson2013metric}), we are able to
generalize stability results to the case of noisy data points. By
assuming $P \subseteq \Offset \M \delta$ instead of $P \subseteq \M$,
this gives us the flexibility to perturb the data points and ensure
correctness of the FDC after some particular perturbation.

Third, the framework of the FDC is particularly convenient for supporting
the proof of correctness of a linear variational formulation, which
we present in a companion paper \cite{socg22}.


%
%
%
%
%

\section{Preliminaries}
\label{section:preliminaries}

In this section, we review the necessary background and explain some of our terms.

\subsection{Subsets and submanifolds}
%
Given a subset $A \subseteq \Rspace^\Dim$, the affine space spanned by
$A$ is denoted by $\Aff A$ and the convex hull of $A$ by $\Conv{A}$.
The {\em medial axis} of $A$, denoted as $\MA{A}$, is the set of points
in $\Rspace^\Dim$ that have at least two closest points in $A$. The
{\em projection map} $\pi_A : \Rspace^\Dim \setminus \MA{A} \to A$
associates to each point $x$ its unique closest point in $A$. The {\em
  reach} of $A$ is the infimum of distances between $A$ and its medial
axis and is denoted as $\Reach A$. By definition, the projection map
$\pi_A$ is well-defined on every subset of $\Rspace^\Dim$ that does not intersect the
medial axis of $A$. In particular, letting the {\em $r$-tubular
  neighborhood} of $A$ be the set of points $\Offset A r = \{x
\in \Rspace^\Dim \mid d(x,A) \leq r \}$, the projection map $\pi_A$ is
well-defined on every $r$-tubular neighborhood of $A$ with $r < \Reach
A$. For short, we say that a subset $\sigma \subseteq \Rspace^\Dim$ is
$\rho$-small if it can be enclosed in a ball of radius $\rho$.

Throughout the paper, $\M$ designates a compact $C^2$ $d$-dimensional
submanifold of $\Rspace^\Dim$ for $d < \Dim$. For any point $m \in
\M$, the tangent plane to $m$ at \M is denoted as $\Tangent m
\M$. Because \M is $C^2$ and therefore $C^{1,1}$, the reach of \M is
positive \cite{federer-59}. We let $\reach$ be a fixed finite constant such that $0 <
\reach \leq \Reach{\M}$.

\subsection{Simplicial complexes}
%
In this section, we review some background notation on simplicial
complexes. For more details, the reader is referred to
\cite{munkres1993elements}. We also introduce the concept of faithful
reconstruction which encapsulates what we mean by a ``desirable''
approximation of a manifold.

All simplicial complexes that we consider are abstract. An
\emph{abstract simplicial complex} is a collection $K$ of finite
non-empty sets, such that if $\sigma$ is an element of $K$, so is
every non-empty subset of $\sigma$. The element $\sigma$ of $K$ is
called an \emph{abstract simplex} and its \emph{dimension} is one less
than its cardinality. The \emph{vertex set} of $K$ is the union of its
elements, $\Vertexset K = \bigcup_{\sigma \in K} \sigma$.  We are
interested in the situation where the vertex set of $K$ is a subset of
$\Rspace^\Dim$. In that situation, each abstract simplex $\sigma
\subseteq \Rspace^\Dim$ is naturally associated to a geometric simplex
defined as $\Conv{\sigma}$. The dimension of $\Conv{\sigma}$ is the
dimension of the affine space $\Aff{\sigma}$ and cannot be larger than
the dimension of the abstract simplex $\sigma$. When $\dim(\sigma) =
\dim(\Aff \sigma)$, we say that $\sigma$ is
\emph{non-degenerate}. Equivalently, the vertices of $\sigma$ form an affinely
independent set of points.

  Given a set of simplices $\Sigma$ with vertices in $\Rspace^\Dim$ (not necessarily forming a
simplicial complex), let us define the \emph{shadow} of $\Sigma$ as
the subset of $\Rspace^\Dim$ covered by the relative interior of
the geometric simplices associated to abstract simplices in $\Sigma$,
$\Shadow{\Sigma} = \bigcup_{\sigma \in \Sigma} \relint{\Conv{\sigma}}$.
  We shall
say that  $\Sigma$  is \emph{geometrically realized} (or
\emph{embedded}) if (1)
$\dim(\sigma)=\dim(\Aff\sigma)$ for all $\sigma \in \Sigma$ and (2)
$\Conv(\alpha \cap\beta) = \Conv\alpha \cap \Conv\beta$ for all
$\alpha, \beta \in \Sigma$.
%
%
%
%
%
%
%
\begin{definition}[Faithful reconstruction]
  Consider a subset $A \subseteq \Rspace^\Dim$ whose reach is positive,
  and a simplicial complex $K$ with a vertex set in $\Rspace^\Dim$. We
  say that $K$ \emph{reconstructs $A$ faithfully}  (or is a
  \emph{faithful reconstruction} of $A$) if the following three
  conditions hold:
  \begin{description}[itemsep=1pt,parsep=1pt,topsep=4pt]
  \item[\styleitem{Embedding:}] $K$ is geometrically realized;
  \item[\styleitem{Closeness:}] $\US K$ is contained in the $r$-tubular neighborhood of $A$ for some $0 \leq r < \Reach A$;
  \item[\styleitem{Homeomorphism:}] The restriction of $\pi_A: \Rspace^\Dim \setminus \MA{A} \to A$ to $\US K$ is a homeomorphism.
  \end{description}
\end{definition}
%


\subsection{Height, circumsphere and smallest enclosing ball}

All simplices we consider in the paper are abstract, unless
explicitely stated otherwise. The \emph{height} of a simplex $\sigma$
is $\height{\sigma} = \min_{v \in \sigma} d(v,\Aff(\sigma \setminus
\{v\}))$. The height of $\sigma$ vanishes if and only if $\sigma$ is
degenerate. If $\sigma$ is non-degenerate, then, letting
$d=\dim\sigma=\dim\Aff\sigma$, there exists a unique $(d-1)$-sphere
that circumscribes $\sigma$ and therefore at least one $(N-1)$-sphere
that circumscribes $\sigma$. Hence, if $\sigma$ is non-degenerate, it
makes sense to define $S(\sigma)$ as the smallest $(N-1)$-sphere that
circumscribes $\sigma$. Let $Z(\sigma)$ and $R(\sigma)$ denote the
center and radius of $S(\sigma)$, respectively.  Let $c_\sigma$ and
$r_\sigma$ denote the center and radius of the smallest $N$-ball
enclosing $\sigma$, respectively. Clearly, $r_\sigma \leq R(\sigma)$
and both $c_\sigma$ and $Z(\sigma)$ belong to $\Aff \sigma$. The
intersection $S(\sigma) \cap \Aff \sigma$ is a $(d-1)$-sphere which is
the unique $(d-1)$-sphere circumscribing $\sigma$ in $\Aff \sigma$.

\subsection{Delaunay complexes}

Consider a finite point set $Q \subseteq \Rspace^\Dim$. We say
  that an $(N-1)$-sphere is $Q$-empty if it is the boundary of a ball
  that contains no points of $Q$ in its interior. We say that $\sigma
  \subseteq Q$ is a \emph{Delaunay simplex} of $Q$ if there exists an
  $(N-1)$-sphere that circumscribes $\sigma$ and is $Q$-empty.  The
set of Delaunay simplices form a simplicial complex called the
\emph{Delaunay complex} of $Q$ and denoted as $\Del{Q}$.

\begin{definition}[General position]
  Let $d = \dim(\Aff Q)$. We say that $Q
  \subseteq \Rspace^\Dim$ is in \emph{general position} if no $d+2$
  points of $Q$ lie on a common $(d-1)$-dimensional sphere.
\end{definition}

\begin{lemma}
When $Q$ is in general position, $\Del{Q}$ is geometrically realized.
\end{lemma}

%
%
%
%
%

\section{Flat Delaunay complex}
\label{section:flat-delaunay-complex}

For simplicity, whenever $x \in \Rspace^\Dim \setminus \MA{\M}$, we
shall write $x^* = \pi_\M(x)$ for the projection of $x$ onto \M.
Afterwards, we assume once and for all that $P \subseteq \Rspace^\Dim
\setminus \MA{\M}$, so that the projection $p^* = \pi_\M(p)$ is
well-defined at every point $p \in P$. Given \M, $P$ and a scale
parameter $\rho \geq 0$, we introduce a construction which we call the
flat Delaunay complex of $P$ with respect to \M at scale $\rho$
(Section \ref{section:flat-delaunay-complex-definition}) and
  make some preliminary remarks (Section
  \ref{section:flat-delaunay-complex-preliminaries}).

\subsection{Definitions}
\label{section:flat-delaunay-complex-definition}

  \begin{figure}[htb]
    \begin{center}
      \includegraphics[width=0.495\linewidth]{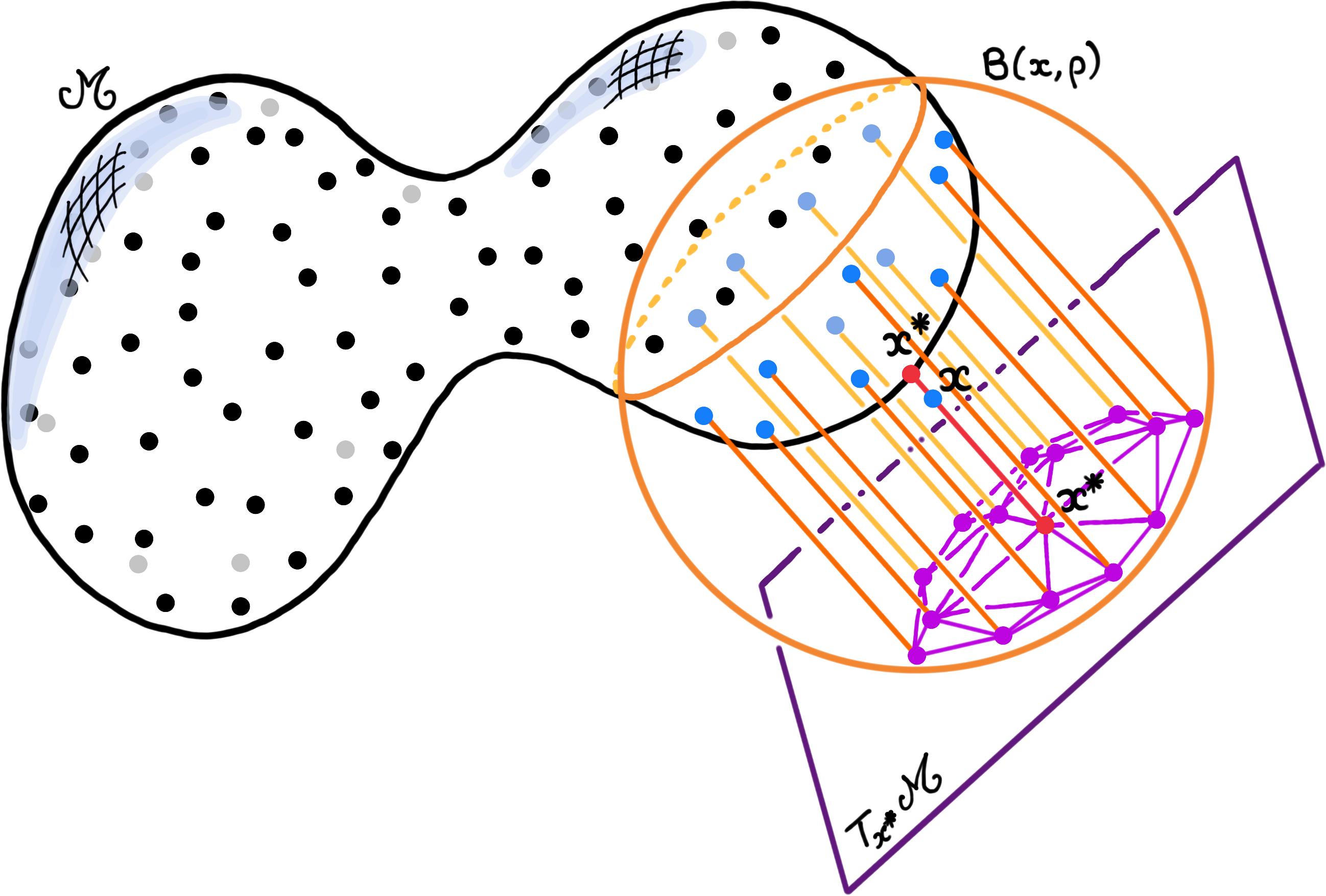}\hfill
      \includegraphics[width=0.495\linewidth]{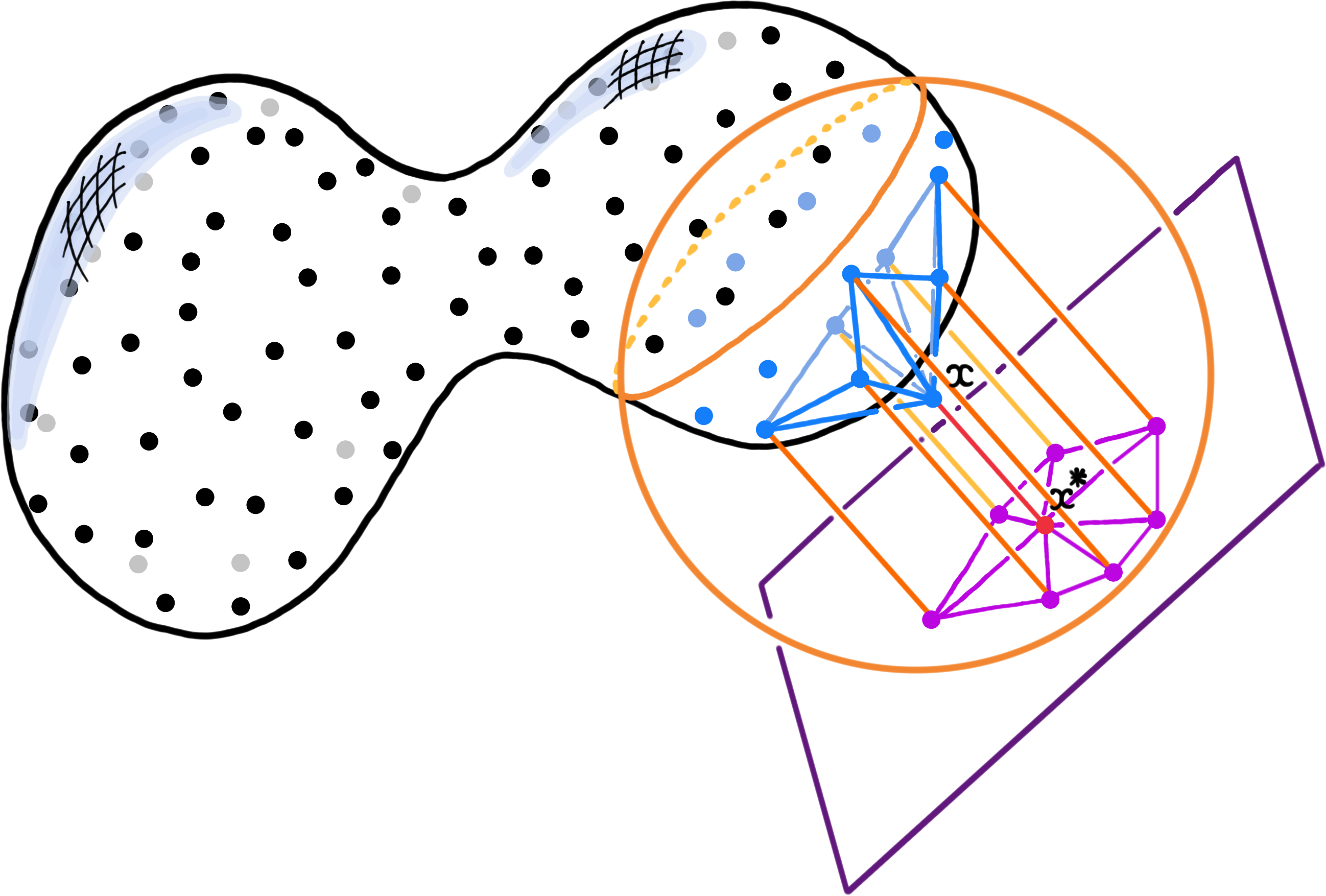}      
    \end{center}
    \caption{Construction of the prestar of $x$ at scale $\rho$. Left:
      Points in $P \cap B(x^\ast,\rho)$ are projected onto the tangent
      space $\Tangent {x^\ast} \M$ and the Delaunay complex of the
      projected points is computed. For clarity, we have translated
      $\Tangent {x^\ast} \M$. Right: The star of $x^\ast$ (in purple)
      is the set of simplices that cover $x^\ast$ and the prestar of
      $x$ (in blue) is the set of simplices $\sigma \in P \cap
      B(x^\ast,\rho)$ whose projection belongs to the star of $x^\ast$. \label{figure:building-flatdel}}
  \end{figure}

\begin{definition}[Stars and Prestars]
  Given a point $m \in \M$, we call the {\em star of $m$ at scale
    $\rho$} the set of simplices
  \[
  \Delstar m \rho = \{ \tau \in \Del{\pi_{\Tangent m \M}(P \cap B(m,\rho))} \mid m \in \Conv \tau \}.
  \]
  Given a point $x \in \Rspace^\Dim \setminus \MA{\M}$, we call the
  {\em prestar of $x$ at scale $\rho$} the set of simplices
  \[
  \begin{split}
    \Prestar x \rho = \{ &\sigma \subseteq P \cap B(x^*,\rho)
    \mid \pi_{\Tangent {x^*} \M}(\sigma) \in \Delstar {x^*} \rho   \}.
  \end{split}
  \]
\end{definition}
 Figure \ref{figure:building-flatdel} illustrates the
  construction of the prestar of a point $x$ in $P$.

\begin{remark}
  \label{remark:same-projections-same-prestar}
  By definition, if two points $x$ and $y$ share the same projection
  onto \M, that is, if $x^* = y^*$, then they also share the
  same prestar at scale $\rho$, that is, $\Prestar x \rho = \Prestar y
  \rho$. In particular, $\Prestar x \rho = \Prestar {x^*} \rho$
    whenever the projection at $x$ is well-defined.
\end{remark}

\begin{definition}[Flat Delaunay complex]
  \label{definition:flat-Delaunay-complex}
  The {\em flat Delaunay complex of $P$ with respect to \M at scale
    $\rho$} is the set of simplices
  \[
  \FlatDel P \rho = \bigcup_{p \in P} \Prestar p \rho.
  \]
\end{definition}

Note that the flat Delaunay complex is not necessarily a
simplicial complex but becomes one under the assumptions of our two main theorems
(Theorems \ref{theorem:homeomorphism-from-structural-conditions} and \ref{theorem:homeomorphism-from-sampling-conditions}).

\subsection{Preliminary remarks}
\label{section:flat-delaunay-complex-preliminaries}

\begin{remark}
  \label{remark:simplices-in-prestar-or-tangdel-are-small}
  By definition, if a simplex $\sigma$ belongs to the prestar of some
  point $x$ at scale $\rho$, then $\sigma$ fits in a ball of radius
  $\rho$ and therefore is $\rho$-small and so are simplices
  in $\FlatDel P \rho$.
\end{remark}

\begin{remark}
  \label{remark:same-projections}
   For all points $x$ at which $d(x,\M) < \Reach{M}$ and all $m \in
   \M$, we have that $m = \pi_{\Tangent {m} \M}(x) \iff m = \pi_\M(x)$.
\end{remark}

We now provide two alternate characterizations of simplices in
  the prestar. The first one is a direct consequence of the above
  remark and will be useful in the proof of
  Theorem~\ref{theorem:homeomorphism-from-sampling-conditions} and the
  second one will facilitate the proof of
  Lemma~\ref{lemma:delloc-characterization}.

\begin{remark}[First characterization of prestars]
  \label{remark-prestar-characterization}
  If $x$ is a point  at which $d(x,\M) < \Reach \M$ and $\rho < \Reach \M$, then:
  \[
    \sigma \in \Prestar x \rho ~\iff~
    \begin{cases}
    \sigma \subseteq P \cap B(x^*,\rho)\\
    \pi_{\Tangent {x^*} \M}(\sigma) \in \Del{\pi_{\Tangent {x^*} \M}(P \cap B(x^*,\rho))}\\
    x^* \in \pi_{\M}( \Conv\sigma)
    \end{cases}
  \]
\end{remark}

\begin{remark}[Second characterization of prestars]
  \label{remark-prestar-alternative-definition}
  For all simplices $\sigma$ such that $\Conv \sigma
    \subseteq \Offset \M \rho$ with $\rho < \Reach \M$
  and all $m \in \pi_\M(\Conv \sigma)$,
  \[
  \sigma \in \Prestar m \rho ~\iff~
  \begin{cases}
    \sigma \subseteq P \cap B(m,\rho) \\
    \pi_{\Tangent {m} \M}(\sigma) \in \Del{\pi_{\Tangent {m} \M}(P \cap B({m},\rho))}
  \end{cases}
  \]
\end{remark}
  Indeed, applying Remark \ref{remark-prestar-characterization} with $x = m$,
  we observe that the last condition on the right side of the
  equivalence is redundant because $m^* = m \in \pi_\M(\Conv \sigma)$.

%
%
%
%
%

\section{Faithful reconstruction from structural conditions}
\label{section:faithful-from-structural-conditions}

In this section, we exhibit a set of structural conditions
under which $\FlatDel P \rho$ is a faithful reconstruction of
\M. These conditions are encapsulated in our first reconstruction
theorem (Theorem
\ref{theorem:homeomorphism-from-structural-conditions} below). Among
the conditions, we find that every $\rho$-small $d$-simplex $\sigma
\subseteq P$ must have its prestars in agreement at scale $\rho$:

\begin{definition}[Prestars in agreement]
  We say that the {\em prestars of $\sigma$ are in agreement at scale
    $\rho$} if for all $x,y \in \Conv{\sigma}$, the following
  equivalence holds: $\sigma \in \Prestar x \rho \iff \sigma \in
  \Prestar y \rho$.
\end{definition}

Compare to the work in
\cite{boissonnat2018geometric,boissonnat2014manifold}, we define the
prestars everywhere in space, not merely at the data points $P$ and we
enforce the prestars to agree at every pair of points in $\Conv\sigma$
and not merely at the vertices of $\sigma$. This trick allows us to
provide a short proof that agreement of prestars, among other
lighter conditions, imply that the flat Delaunay complex is a faithful
reconstruction of $\M$.

\begin{theorem}[Faithful reconstruction from structural conditions]
  \label{theorem:homeomorphism-from-structural-conditions}
  Suppose that $P \subseteq \Offset \M \rho$ with $\rho <
  \frac{\reach}{2}$ and assume that the following structural conditions
  are satisfied:
  \begin{enumerate}[label=\styleitem{(\arabic*)},itemsep=1pt,parsep=1pt,topsep=4pt]

  \item \label{hypo:structural-proj-M}
    For every $\rho$-small $d$-simplex $\sigma \subseteq P$, the
    map $\restr{\pi_\M}{\Conv\sigma}$ is injective;

  \item \label{hypo:structural-proj-onto-tangent-plane}
    For all $m \in \M$, the map $\restr{\pi_{\Tangent m \M}}{P \cap B(m,\rho)}$ is injective;

  \item \label{hypo:structural-locally-manifold}
    For all $m \in \M$, $\US{\Delstar m \rho}$ is homeomorphic to $\Rspace^d$ in a neighborhood of $m$;

  \item \label{hypo:structural-geometrically-realized}
    For all $m \in \M$, $\Delstar m \rho$ is geometrically realized;
    
  \item \label{hypo:structural-agreement} Every $\rho$-small $d$-simplex
    $\sigma \subseteq P$ has its prestars in agreement at scale $\rho$.
  \end{enumerate}
  Then,
  \begin{itemize}
  \item $\FlatDel P \rho$ is a simplicial complex;
  \item For all $m \in \M$,  $\Prestar m \rho = \{ \sigma \in \FlatDel P \rho \mid m \in \pi_\M(\Conv\sigma) \}$;
  \item $\FlatDel P \rho$ is a faithful reconstruction of \M.
  \end{itemize}  
\end{theorem}

Before giving the proof, we start with a remark.

\begin{remark}
  \label{remark:projection-well-defined-for-simplex}
  For any simplex $\sigma \subseteq \Offset \M \rho$ that can be
  enclosed in a ball of radius $\rho < \frac{\reach}{2}$, then $\Conv
  \sigma \subseteq \Offset \M {2\rho}$.  Indeed, for all $x \in \Conv
  \sigma$, $d(x,\M) \leq d(x,\sigma) + \rho \leq 2\rho$. Hence, the
  map $\restr{\pi_\M}{\Conv \sigma}$ is well-defined.
\end{remark}



\begin{proof}[Proof of Theorem \ref{theorem:homeomorphism-from-structural-conditions}]
We prove the lemma in seven (short) stages:

\medskip

\noindent {\bf (a)} First, we prove the following implication:
\[
\begin{cases}
  \sigma \in \Prestar x \rho\\
  \text{for all $x\in\Conv\sigma$}
\end{cases}
\implies \quad
\begin{cases}
  \tau \in \Prestar x \rho\\
  \text{for all $\tau \subseteq \sigma$ and all $x \in \Conv\tau$}
\end{cases}
\]
Indeed, suppose that $\sigma \in \Prestar x \rho$ for all
$x\in\Conv\sigma$. Using Remark \ref{remark-prestar-characterization},
this is equivalent to saying that for all $x \in \Conv \sigma$:
\begin{equation*}
  \begin{cases}
    \sigma \subseteq P \cap B(x^*,\rho)\\
    \pi_{\Tangent {x^*} \M}(\sigma) \in \Del{\pi_{\Tangent {x^*} \M}(P \cap B(x^*,\rho))}\\
    x^* \in \pi_{\M}( \Conv\sigma)
  \end{cases}
\end{equation*}
Letting $x$ be any point of $\Conv \tau$ and using $\tau \subseteq \sigma$, we
obtain that:
\begin{equation*}
  \begin{cases}
    \tau \subseteq P \cap B(x^*,\rho)\\
    \pi_{\Tangent {x^*} \M}(\tau) \in \Del{\pi_{\Tangent {x^*} \M}(P \cap B(x^*,\rho))}\\
    x^* \in \pi_{\M}( \Conv\tau)
  \end{cases}
\end{equation*}
But, using again Remark \ref{remark-prestar-characterization}, this
translates into saying that $\tau \in \Prestar x \rho$ for all $x \in
\Conv \tau$ as desired.

\medskip

\noindent {\bf (b)} Second, we establish the following implication:
\[
\text{$\tau \in \Prestar m \rho$ for some $m\in \M$}
\implies
\text{$\tau \in
  \Prestar x \rho$ for all $x \in \Conv\tau$}.
\]
Consider a simplex $\tau \in \Prestar m \rho$ for some $m\in \M$ and
let us show that $\tau \in \Prestar x \rho$ for all $x \in \Conv\tau$.
Since $\tau \in \Prestar m \rho$, this implies that $\tau' =
\pi_{\Tangent {m} \M}(\tau) \in \Delstar m \rho$. Because of
hypothesis \refhypo{hypo:structural-locally-manifold}, $\US{\Delstar m \rho}$
is homeomorphic to $\Rspace^d$ and therefore contains at least
one $d$-simplex $\sigma' \supseteq \tau'$. Since $\sigma' \in
  \Delstar m \rho$, it follows from the definition of the star that $\sigma' \subseteq
  \pi_{\Tangent m \M}(P \cap B(m,\rho))$.  Because of hypothesis
  \refhypo{hypo:structural-proj-onto-tangent-plane}, the projection
  $\pi_{\Tangent m \M}$ restricted to $P \cap B(m,\rho)$ is injective
  and therefore there exists a unique $\sigma \subseteq P \cap
  B(m,\rho)$ such that $\sigma' = \pi_{\Tangent m \M}(\sigma)$ and
  furthermore $\sigma' \supseteq \tau'$ implies that $\sigma \supseteq
  \tau$. Since $\sigma \subseteq P \cap B(m,\rho)$ and $\sigma' \in
  \Delstar m \rho$, we get that $\sigma \in \Prestar m \rho$.  By
  hypothesis \refhypo{hypo:structural-agreement}, the prestars of $\sigma$
  are in agreement at scale $\rho$ and therefore $\sigma \in \Prestar
  m \rho$ implies $\sigma \in \Prestar x \rho$ for all $x \in \Conv
  \sigma$. Using the previous stage, we get that $\tau \in \Prestar x
  \rho$ for all $x \in \Conv \tau$ as desired.

\medskip

\noindent {\bf (c)} Third, we prove that $\FlatDel P \rho$ is
a simplicial complex. Consider $\sigma \in \FlatDel P \rho$ and $\tau
\subseteq \sigma$ and let us prove that $\tau \in \FlatDel P \rho$. By
definition of the flat Delaunay complex, we can find $m \in \M$ such
that $\sigma \in \Prestar m \rho$. Using Stage {\bf (b)}, we deduce
that $\sigma \in \Prestar x \rho$ for all $x \in \Conv
\sigma$. Using Stage {\bf (a)}, $\tau \in \Prestar x \rho$ for all
$x \in \Conv \tau$. By picking $x \in \tau \subseteq P$, this shows that $\tau \in \FlatDel P \rho$.

\medskip

\noindent {\bf (d)} Fourth, we claim that
\begin{equation}\label{equation:ClaimDefAlternativeFlatDel}
\FlatDel P \rho = \bigcup_{m\in \M} \Prestar m \rho,
\end{equation}
where the union is over all points $m$ of \M and not merely points of
$P$.  The direct inclusion is clear. To establish the reverse
inclusion, consider a simplex $\tau \in \Prestar m \rho$ for some
$m\in \M$ and let us show that $\tau \in \Prestar p \rho$ for some $p
\in P$. In Stage {\bf (b)}, we proved that $\tau \in \Prestar m
  \rho$ for some $m \in \M$ implies that $\tau \in \Prestar x \rho$ for all $x \in \Conv \sigma$ and
  thus, picking $x$ among the vertices of $\tau$ establishes the
  claim.

\medskip

\noindent {\bf (e)} Fifth, we establish that for all $m \in \M$, 
\begin{equation}
  \label{eq:prestar-proof}
  \Prestar m \rho = \{ \sigma \in \FlatDel P \rho \mid m \in \pi_\M(\Conv\sigma) \}.
\end{equation}
Let $m \in \M$. To establish the direct inclusion, consider a simplex
$\sigma \in \Prestar m \rho$. By
Equation~\eqref{equation:ClaimDefAlternativeFlatDel}, $\sigma \in
\FlatDel P \rho$, and by Remark \ref{remark-prestar-characterization},
we get that $m \in \pi_\M(\Conv\sigma)$.  To establish the reverse
inclusion, consider a simplex $\sigma \in \FlatDel P \rho$ such that
$m \in \pi_\M(\Conv\sigma)$. Because $\sigma \in \FlatDel P \rho$, we
can find $m' \in \M$ (which is the projection of a point of $P$) such
that $\sigma \in \Prestar {m'} \rho$. Applying Stage {\bf (b)}, we
deduce that $\sigma \in \Prestar {m'} \rho$ implies that $\sigma \in
\Prestar x \rho$ for all $x \in \Conv \sigma$ and picking $x \in \Conv
\sigma$ such that $m = \pi_\M(x)$ and using Remark
\ref{remark:same-projections-same-prestar}, we deduce that $\sigma \in
\Prestar m \rho$.

\medskip

\noindent {\bf (f)} Sixth, we prove that $\FlatDel P \rho$ is
geometrically realized. Consider a pair $(\alpha, \beta)$ of simplices
in $\FlatDel P \rho$ and let us prove that $\Conv \alpha \cap \Conv
\beta = \Conv (\alpha \cap \beta)$. Clearly, $\Conv \alpha \cap \Conv
\beta \supseteq \Conv (\alpha \cap \beta)$. To prove the converse
inclusion, suppose that there exists a point $x \in \Conv \alpha \cap
\Conv \beta$ and let us prove that $x \in \Conv(\alpha \cap
\beta)$. Because both $\alpha$ and $\beta$ are $\rho$-small, Remark
\ref{remark:projection-well-defined-for-simplex} implies that $\pi_\M$
is well-defined on both and we write $m = \pi_\M(x)$.  Because both
$\alpha$ and $\beta$ belong to $\FlatDel P \rho$ while
$\pi_\M(\Conv\alpha)$ and $\pi_\M(\Conv\beta)$ cover $m$, it follows
from Equation (\ref{eq:prestar-proof}) that both $\alpha$ and $\beta$ belong to
$\Prestar m \rho$ and therefore both $\alpha' = \pi_{\Tangent m
  \M}(\alpha)$ and $\beta' = \pi_{\Tangent m \M}(\beta)$ belong to
$\Delstar m \rho$. Since the latter is geometrically realized (hypothesis \refhypo{hypo:structural-geometrically-realized}), we have
$m \in \Conv \alpha' \cap \Conv \beta' = \Conv (\alpha' \cap \beta')$
and since we have assumed that the restrictiction of $\pi_{\Tangent
  {m} \M}$ to points in $P \cap B(m,\rho) \supseteq \alpha \cup \beta$
is injective, $m \in \Conv (\alpha' \cap \beta') = \Conv(\prTangent m
{\alpha \cap \beta}) = \prTangent m {\Conv (\alpha \cap
  \beta)}$. Using Remark \ref{remark:same-projections}, we get that $m
= \pi_\M(x) = \prTangent {m} x$ and using the injectivity of
$\pi_{\Tangent {m} \M}$ on $P \cap B(m,\rho)$, we get that $x \in
\Conv(\alpha \cap \beta)$. This proves that $\FlatDel P \rho$ is
geometrically realized.

  \medskip

\noindent {\bf (g)} Seventh, we prove that $\US{\FlatDel P
  \rho}$ is a $d$-manifold and that the map
\[\pi_\M : \US{\FlatDel P
  \rho} \to \M\] is injective. Consider a point $x \in \US{\FlatDel P
  \rho}$ and let $m = \pi_\M(x)$. Observe that in a small neighborhood
of $x$, the set $\US{\FlatDel P \rho}$ coincides with the set
$\US{\Prestar m \rho}$ because of Equation
(\ref{eq:prestar-proof}). Note that the map $\pi_{\Tangent m
  \M}$ is a bijective correspondence between $P \cap B(m,\rho)$ and
$\pi_{\Tangent m \M}(P \cap B(m,\rho))$ such that $\sigma \in \Prestar
m \rho$ if and only if $\pi_{\Tangent m \M}(\sigma) \in \Delstar m
\rho$. We note that $\sigma$ and $\pi_{\Tangent m \M}(\sigma)$ which
share the same dimension are both non-degenerate. Indeed,
$\pi_{\Tangent m \M}(\sigma)$ is non-degenerate because it belongs to
$\Delstar m \rho$ which we have assumed to be geometrically
realized
(hypothesis \refhypo{hypo:structural-geometrically-realized}). Simplex $\sigma$
is also non-degenerate since $\sigma$ has as many vertices as
$\pi_{\Tangent m \M}(\sigma)$ and the dimension of $\Aff\sigma$ cannot
be smaller than the dimension of its projection $\Aff \pi_{\Tangent m \M}(\sigma)$
which is full. Hence, $\pi_{\Tangent {m} \M}$ is an isomorphism
between $\Prestar m \rho$ and $\Delstar m \rho$ and both $\Prestar m
\rho$ and $\Delstar m \rho$ are geometrically realized. We deduce that
the induced simplicial map $\pi_{\Tangent {m} \M} : \US{\Prestar m
  \rho} \to \US{\Delstar m \rho}$ is a homeomorphism.  Since in a
neighborhood of $m$, $\US{\Delstar m \rho}$ is homeomorphic to
$\Rspace^d$, it follows that in a neighborhood of $x$, $\US{\FlatDel P
  \rho}$ which coincides with $\US{\Prestar m \rho}$ is also
homeomorphic to $\Rspace^d$. This proves that $\US{\FlatDel P \rho}$
is a $d$-manifold. Let us prove that $\pi_\M : \US{\FlatDel P \rho}
\to \M$ is injective. Consider two points $x$ and $y$ in $\US{\FlatDel
  P \rho}$ such that $\pi_\M(x) = \pi_\M(y) = m$. Then, by Remark
\ref{remark:same-projections}, $\pi_{\Tangent m \M}(x) = \pi_{\Tangent
  m \M}(y) = m$ and since we have just established that $\pi_{\Tangent m \M} : \US{\Prestar m \rho}
\to \US{\Delstar m \rho}$ is a homeomorphism, we deduce that $x = y$,
establishing the injectivity of $\pi_\M$.

  \medskip

  \noindent {\bf (h)} Finally, we prove that $\pi_\M$ is a
  homeomorphism between $\D = \US {\FlatDel P \rho}$ and \M. Recall
  that $\D$ and \M are two $d$-manifolds (without boundary) and that
  the restriction of $\pi_\M$ to $\D$ is an injective continuous
  map. Since for all $m \in \M$, $\US{\Delstar m
    \rho}$ is homeomorphic to $\Rspace^d$ in a  neighborhood of
  $m$, this implies that there exists $x \in \D$ such that
  $\pi_\M(x)=m$ and $\restr{\pi_\M}{\D}$ is surjective.  Applying
  the domain invariance theorem, we get that $\pi_\M : \D \to \M$ is
  open and therefore $\pi_\M$ is a homeomorphism between \D and
  $\pi_\M(\D)=\M$.
\end{proof}


%
%
%
%
%

\section{Faithful reconstruction from sampling and safety conditions}
\label{section:faithful-from-sampling-and-safety-conditions}

In this section, we state our second reconstruction theorem
  (Theorem \ref{theorem:homeomorphism-from-sampling-conditions}
  below). The theorem describes geometric conditions on $P$ under
  which (1) $\FlatDel P \rho$ is a faithful reconstruction of \M and
  (2) $\FlatDel P \rho$ satisfies certain properties that are needed
  in our companion paper \cite{socg22}. In particular, the theorem provides a
  characterization of $d$-simplices in $\FlatDel P \rho$ as the
  $d$-simplices that are delloc in $P$ at scale $\rho$:

\begin{definition}[Delloc simplex]
\label{definition:delloc}
  We say that a simplex $\sigma$ is \emph{delloc} in $P$ at
  scale $\rho$ if $\sigma \in \Del{\pi_{\Aff \sigma}(P \cap
    B(c_\sigma,\rho)}$.
\end{definition}

Note that deciding whether a simplex is delloc does not require the
knowledge of the manifold \M. We emphasize the fact that the
characterization of $d$-simplices in the flat Delaunay complex as the
one being delloc turns out to be crucial in our companion
paper \cite{socg22}.

In this section, we first introduce the necessary notations and
  definitions to describe the geometric conditions on $P$ that we
  need. We then state our second reconstruction theorem and sketch the proof.



\begin{definition}[Dense, accurate, and separated]
  We say that $P$ is an {\em $\varepsilon$-dense} sample of \M if for
  every point $m \in \M$, there is a point $p \in P$ with $\|p-m\|
  \leq \varepsilon$ or, equivalently, if $\M \subseteq \Offset P
  \varepsilon$. We say that $P$ is a {\em $\delta$-accurate} sample of
  \M if for every point $p \in P$, there is a point $m \in \M$ with
  $\|p-m\| \leq \delta$ or, equivalently, if $P \subseteq \Offset \M
  \delta$. Let $\Sep P = \min_{p,q \in P} \|p-q\|$.
\end{definition}
We stress that our definition of a protected simplex differs slightly
from the one in \cite{BOISSONNAT_2013,boissonnat2018geometric}.

\begin{definition}[Protection]
  We say that a non-degenerate simplex $\sigma \subseteq \Rspace^\Dim$
  is {\em $\zeta$-protected} with respect to $Q
  \subseteq \Rspace^\Dim$ if for all $q \in Q \setminus \sigma$, we
  have $d(q,S(\sigma)) > \zeta$. We shall simply say that $\sigma$ is
  \emph{protected} with respect to $Q$ when it is $0$-protected with respect
  to $Q$.
\end{definition}
Let $\Hcal(\sigma) = \{ \Tangent m \M \mid m \in \pi_\M(\Conv\sigma)\} \cup \{ \Aff \sigma \}$, and $\Theta(\sigma)
= \max_{H_0, H_1 \in \Hcal(\sigma)} \angle(H_0, H_1)$. To the pair $(P,\rho)$ we now
associate  three quantities that describe the
quality of $P$ at scale $\rho$:
\begin{itemize}
\item $\Height{P}{\rho} = \min_{\sigma} \height{\sigma}$, where the
  minimum is  over all $\rho$-small $d$-simplices $\sigma
  \subseteq P$;
\item $\Theta(P,\rho) = \max_{\sigma} \Theta(\sigma)$, where $\sigma$ ranges over all
  $\rho$-small $d$-simplices of $P$;
\item $\Protection{P}{\rho} = \min_\sigma\min_q d(q,S(\sigma))$, where
  the minima are over all $\rho$-small $d$-simplices $\sigma \subseteq
  P$ and all points $q \in \pi_{\Aff\sigma}(P \cap B(c_\sigma,\rho))
  \setminus \sigma$.
\end{itemize}



\begin{theorem}[Faithful reconstruction from sampling and safety conditions]
  \label{theorem:homeomorphism-from-sampling-conditions}
  Let $\varepsilon$, $\delta$, $\rho$, $\theta$ be non-negative real
  numbers and set $A = 4 \delta \theta + 4 \rho \theta^2$. Assume that
  $\theta \leq \frac{\pi}{6}$, $\delta \leq \varepsilon$, and
  $16\varepsilon \leq \rho < \frac{\reach}{4}$. Suppose that $P$ satisfies the following sampling conditions: $P$ is a
  $\delta$-accurate $\varepsilon$-dense sample of \M. Suppose furthermore that $P$ satisfies
  the following safety conditions:
  \begin{enumerate}[label=\styleitem{(\arabic*)},itemsep=1pt,parsep=1pt,topsep=4pt]


  \item 
    \label{hypo:safety-angle}
    $\Theta(P,\rho) \leq \theta - 2 \arcsin \left(\frac{\rho+\delta}{\reach}\right)$.
    
  \item 
    \label{hypo:safety-separation}
    $\Sep P > 2 A + 6 \delta + \frac{2\rho^2}{\reach}$;
    
  \item 
    \label{hypo:safety-protection}
    $\Height{P}{\rho}>0$ and $\Protection{P}{3\rho} > 2 A \left( 1 + \frac{4 d \varepsilon}{\Height{P}{\rho}} \right)$.
    
  \end{enumerate}
  Then we have the following properties:
  \begin{description}[itemsep=1pt,parsep=1pt,topsep=4pt]
  \item [\styleitem{Faithful reconstruction:}] $\FlatDel P \rho$ is a faithful reconstruction of \M;
  \item [\styleitem{Prestar formula:}] $\Prestar m \rho = \{ \sigma \in \FlatDel P \rho \mid m \in \pi_\M(\Conv\sigma) \},\; \forall m \in \M$;
  \item [\styleitem{Circumradii:}] For all $d$-simplices $\sigma \in \FlatDel
    P \rho$, we have that $R(\sigma) \leq \varepsilon$;
  \item [\styleitem{Characterization:}] For all $d$-simplices $\sigma$,
    $\sigma \in {\FlatDel P \rho} \iff \sigma$ delloc in $P$
    at scale $\rho$.
  \end{description}
\end{theorem}

The geometric conditions that we need for our result can be
  divided in two groups: the sampling conditions and the safety
  conditions. Roughly speaking, the sampling conditions say that $P$
  must be ``sufficiently'' dense and ``sufficiently'' accurate. The
  safety conditions say that (1) the angle that $\rho$-small
  $d$-simplices make with ``nearby'' tangent space to \M must be
  sufficiently small; (2) points in $P$ must be ``sufficiently'' well
  separated; (3) both the protection and the height of $P$ at scale
  $\rho$ must be ``sufficiently'' lower bounded.  Whereas it seems reasonable to assume
  that $P$ satisfies the samping conditions, it is less clear that, in
  practice, $P$ can satisfy both the sampling and safety
  conditions. We show in Section \ref{section:PerturbationProcedure}
  that starting from a situation where $P$ satisfies some ``strong''
  sampling conditions, it is always possible to perturb $P$ in such a
  way that after perturbation, $P$ satisfies both the sampling and
  safety conditions of
  Theorem~\ref{theorem:homeomorphism-from-sampling-conditions}.  


Before sketching the proof of our second theorem, we derive a
corollary that may have computational implications in low-dimensional
ambient spaces.  For this, we recall that $\sigma$ is a {\em Gabriel
  simplex} of $P$ if its smallest circumsphere $S(\sigma)$ does not
enclose any point of $P$ in its interior.

\begin{corollary}
  Under the assumptions of Theorem
  \ref{theorem:homeomorphism-from-sampling-conditions}, the
  $d$-simplices of $\FlatDel P \rho$ are Gabriel simplices and
  therefore $\FlatDel P \rho \subseteq \Del{P}$.
\end{corollary}

\begin{proof}
   It is easy to see that a delloc simplex $\sigma$ in $P$ at scale
   $\rho$ is also a \emph{Gabriel simplex} of $P$ whenever
   $2R(\sigma)\leq \rho$. The result follows because under the
     assumption of Theorem
     \ref{theorem:homeomorphism-from-sampling-conditions},
     $d$-simplices of $\FlatDel P \rho$ are delloc in $P$ at scale
     $\rho$.
\end{proof}

\paragraph*{Sketch of the proof.}
The proof consists in showing that the sampling and safety
conditions of
Theorem~\ref{theorem:homeomorphism-from-sampling-conditions} imply the
structural conditions of
Theorem~\ref{theorem:homeomorphism-from-structural-conditions}. Applying
Theorem~\ref{theorem:homeomorphism-from-structural-conditions}, we
then get that, amongst other properties, $\FlatDel P \rho$ is a
faithful reconstruction of \M. It is not too difficult to show that
the sampling and safety conditions of
Theorem~\ref{theorem:homeomorphism-from-sampling-conditions} imply the
first three structural conditions of
Theorem~\ref{theorem:homeomorphism-from-structural-conditions}. This
will be established in
Section~\ref{section:basics-on-projections}. The tricky part consists
in proving that the sampling and safety conditions imply the last two
structural conditions, and in particular imply that every $\rho$-small
$d$-simplex $\sigma \subseteq P$ has its prestarts in agreement at
scale $\rho$. Let us introduce the following definitions:

\begin{definition}[Delaunay stability at scale $\rho$]
  Let $\{ (h_i, H_i) \}_{i \in I}$ be a (possibly infinite) set, where
  $h_i$ designates a point of $\Rspace^\Dim$ and
  $H_i\subseteq \Rspace^\Dim$ designates a $d$-dimensional affine
  space.  We say that $\sigma$ is {\em Delaunay stable} for $P$ at scale
  $\rho$ with respect to the set $\{(h_i,H_i)\}_{i \in I}$ if the
  following two propositions are equivalent for all $a,b \in I$:
  \begin{itemize}
  \item $\sigma \subseteq P \cap B(h_a, \rho)$ and $\pi_{H_a}(\sigma) \in \Del{\pi_{H_a}(P \cap B(h_a,\rho))}$;
  \item $\sigma \subseteq P \cap B(h_b, \rho)$ and $\pi_{H_b}(\sigma) \in \Del{\pi_{H_b}(P \cap B(h_b,\rho))}$.
  \end{itemize}
\end{definition}

\begin{definition}[Standard neighborhood]
  We define the {\em standard neighborhood} of $\sigma$ as the set
  $
  \Hscr(\sigma) = \{
  (c_\sigma, \Aff \sigma) \} \cup \{ (x^\ast, \Tangent {x^\ast} \M)
  \}_{x \in \Conv \sigma}.
  $
\end{definition}

Roughly speaking, the next lemma tells us that the Delaunay stability
of a $\rho$-small $d$-simplex $\sigma$ with respect to its standard
neighborhood $\Hscr(\sigma)$ implies both agreement of prestars of
$\sigma$ and a characterization of the property for $\sigma$ to belong
to $\FlatDel P \rho$ in terms of being delloc. Precisely:

\begin{lemma}
  \label{lemma:agreement-from-stability}
  Suppose that $P \subseteq \Offset \M \rho$ with $\rho < \reach$ and
  that for all $m \in \M$, the restriction of map $\pi_{\Tangent m
    \M}$ to $P \cap B(m,\rho)$ is injective.  Consider a $\rho$-small
  $d$-simplex $\sigma \subseteq P$ and suppose that $\sigma$ is
  Delaunay stable for $P$ at scale $\rho$ with respect to its standard
  neighborhood. Then,
  \begin{itemize}
  \item the prestars of $\sigma$ are in agreement at scale $\rho$;
  \item $\sigma \in \FlatDel P \rho \iff$ $\sigma$ is delloc in $P$ at scale $\rho$.
  \end{itemize}
\end{lemma}

\begin{proof}
  Consider the following two propositions:
  \begin{description}[itemsep=1pt,parsep=1pt,topsep=4pt]
  \item [\ita] $\sigma \subseteq P \cap B(c_\sigma,\rho)$ and
    $\sigma \in \Del{\pi_{\Aff \sigma}(P \cap B(c_\sigma,\rho))}$;
  \item [\itb{x}] $\sigma \subseteq P \cap B(x^\ast,\rho)$ and
    $\pi_{\Tangent {x^\ast} \M}(\sigma) \in \Del{\pi_{\Tangent {x^\ast}
      \M}(P \cap B({x^\ast},\rho))}$.
  \end{description}
  Our Delaunay stability hypothesis is equivalent to saying that for all $x \in \Conv \sigma$, we have
  \ita $\iff$ \itb{x} and for all $x,y \in \Conv \sigma$, we have
  \itb{x} $\iff$ \itb{y}.  Using Definition \ref{definition:delloc}
  and Remark \ref{remark-prestar-alternative-definition}, we can
  rewrite Propositions \ita and \itb{x} respectively as:
  \begin{description}[itemsep=1pt,parsep=1pt,topsep=4pt]
  \item [\ita] $\sigma \text{ delloc in $P$ at scale $\rho$}$;
  \item [\itb{x}] $\sigma \in \Prestar x \rho$.
  \end{description}
  Since \itb{x} $\iff$ \itb{y} for all $x,y \in \Conv \sigma$, we get
  that $\sigma \in \Prestar x \rho$ $\iff$ $\sigma \in \Prestar y
  \rho$ for all $x,y \in \Conv \sigma$. In other words, the prestars
  of $\sigma$ are in agreement and the first item of the lemma holds.

  To see that we get the second item of the lemma as well, we claim
  that $\sigma \in \FlatDel P \rho$ $\iff$ there exists $v \in \sigma$
  such that $\sigma \in \Prestar v \rho$. The reverse inclusion is
  clear. To get the direct inclusion, consider $v \in P$ such that
  $\sigma \in \Prestar v \rho$ and let us prove that $v \in
  \sigma$. Because $P \subseteq \Offset \M \rho$ for $\rho < \reach$,
  $\pi_\M$ is well-defined at $v$ and letting $v^\ast = \pi_\M(v)$, we
  clearly have $v \in P \cap B(v^\ast,\rho)$. It follows from our definition
  of a prestar that
  \[
    \sigma \in \Prestar v \rho ~\iff~
    \begin{cases}
    \sigma \subseteq P \cap B(v^\ast,\rho)\\
    \pi_{\Tangent {v^\ast} \M}(\sigma) \in \Del{\pi_{\Tangent {v^\ast} \M}(P \cap B(v^\ast,\rho))}\\
    v^\ast \in \pi_{\Tangent {v^\ast} \M}( \Conv\sigma)
    \end{cases}
  \]
  By Remark~\ref{remark:same-projections},
  $v^\ast=\pi_\M(v)=\pi_{\Tangent {v^\ast} \M}(v)$ and therefore
  $\pi_{\Tangent {v^\ast} \M}(v) \in \Conv{(\pi_{\Tangent {v^\ast}
      \M}(\sigma))}$. Since $\pi_{\Tangent {v^\ast} \M}(\sigma) \in
  \Del{\pi_{\Tangent {v^\ast} \M}(P \cap B(v^\ast,\rho))}$, the only
  possibility is that $\pi_{\Tangent {v^\ast} \M}(v) \in \pi_{\Tangent
    {v^\ast} \M}(\sigma)$ and since $\pi_{\Tangent {v^\ast} \M}$ is
  injective on $P \cap B(v^\ast,\rho)$ (by hypothesis), it follows
  that $v \in \sigma$ as claimed. Hence, we have just proved that
  $\sigma \in \FlatDel P \rho$ $\iff$ there exists $v \in \sigma$ such
  that \itb{v}. Since the latter is equivalent to \ita by hypothesis
  and \ita can be rewritten as $\sigma$ is delloc in $P$ at scale
  $\rho$, we get the secend item of the lemma.
\end{proof}

 The above lemma suggests that we need first to establish the Delaunay
 stability of $\rho$-small $d$-simplices with respect to their standard
 neighborhood. We proceed in three steps. In
 Section~\ref{section:basics-on-projections}, we enunciate basic
 properties on projection maps. We also establish geometric conditions
 under which the first three structural conditions of Theorem
 \ref{theorem:homeomorphism-from-structural-conditions} hold.  In
 Section~\ref{section:stability-through-distortions}, we study the
 Delaunay stability of $d$-simplices with respect to a general set
 $\{(h_0,H_0), (h_1,H_1)\}$, where each pair $(h_i,H_i)$ consists of a
 point $h_i$ and a $d$-dimensional affine space $H_i$ through
 $h_i$. In Section~\ref{section:FlatDel-final}, we prove our second
 theorem by first establishing the Delaunay stability of $d$-simplices
 with respect to their standard neighborhood.  

\section{Basic properties on projection maps}
\label{section:basics-on-projections}

In this section, we enunciate basic properties on projection maps that
we need for the proof of Theorem
\ref{theorem:homeomorphism-from-sampling-conditions}.  We also
  establish geometric conditions under which the first three
  structural conditions of
  Theorem~\ref{theorem:homeomorphism-from-structural-conditions}
  hold. Those conditions are described respectively in Lemma
  \ref{lemma:injectivity-projection-manifold}, Lemma
  \ref{lemma:injectivity-projection-tangent-space} and Lemma
  \ref{lemma:surrounded}.

\begin{lemma}[Injectivity of $\restr{\pi_\M}{\Conv \sigma}$]
  \label{lemma:injectivity-projection-manifold}
  Consider $\sigma \subseteq \Rspace^\Dim$ such that $\Conv \sigma
    \subseteq \Rspace^\Dim \setminus \MA{\M}$.  If $\Theta(\sigma) <
  \frac{\pi}{2}$, then $\restr{\pi_\M}{\Conv \sigma}$ is injective.
\end{lemma}

\begin{proof}
  Suppose for a contradiction that there exist two points $x \neq y$
  in $\Conv \sigma$ that share the same projection $m$ onto \M, in
  other words, such that $x^\ast = y^\ast = m$. Then, the
  straight-line passing through $x$ and $y$ would be orthogonal to the
  tangent space $\Tangent m \M$, implying that $\angle (\Tangent m \M,
  \Conv \sigma) = \frac{\pi}{2}$ and therefore $\Theta(\sigma) =
  \max_{H_0,H_1 \in \Hcal(\sigma)} \angle(H_0,H_1) =
  \frac{\pi}{2}$. But this contradicts our assumption that
  $\Theta(\sigma) < \frac{\pi}{2}$.
\end{proof}

\begin{lemma}[Injectivity of $\restr{\pi_{\Tangent m \M}}{P \cap B(m,\rho)}$]
  \label{lemma:injectivity-projection-tangent-space}
  Suppose that $P \subseteq \Offset \M \delta$ with $16 \delta \leq
  \rho \leq \frac{\reach}{3}$ and $\Sep P > \frac{2\rho^2}{\reach} +
  2\delta$. Then, $\restr{\pi_{\Tangent m \M}}{P \cap B(m,\rho)}$ is
  injective for all $m \in \M$.
\end{lemma}


\begin{proof}
  Consider two points $a,b \in P \cap B(m,\rho)$ and
  let $\theta = \angle(\Tangent m \M, ab)$. We have
  $
  \cos \theta \times \|a-b\| \leq \| \pi_{\Tangent m \M}(a) - \pi_{\Tangent m \M}(b) \|
  $,
  showing that the restriction of $\pi_{\Tangent m \M}$ to
  $B(m,\rho)$ is injective as soon as $\theta < \frac{\pi}{2}$.
  Applying Lemma \ref{lemma:general-angle-bound} with
  $\tau = \{a,b\}$ and $z=m$, we obtain that $\theta$ is upper bounded by
  \[
  \theta \leq \arcsin \left( \frac{2}{\|a-b\|} \left(
  \frac{\rho^2}{\reach} + \delta\right)\right) \leq \arcsin \left( \frac{2}{\Sep P} \left(
  \frac{\rho^2}{\reach} + \delta\right)\right)
  \]
  and thus becomes smaller than $\frac{\pi}{2}$ for $\Sep P > \frac{2\rho^2}{\reach} + 2\delta$.
\end{proof}


\begin{lemma}[Local surjectivity of $\pi_H$]
  \label{lemma:projection-surjective}
  Suppose $\rho < \frac{\reach}{3}$. Let $H \subseteq \Rspace^\Dim$ be
  a $d$-dimensional affine space. Suppose that $H$ passes through a
  point $h$ such that $d(h,\M) \leq \frac{\rho}{4}$ and that there
  exists $\theta \leq \frac{\pi}{6}$ such that $ \angle(H,\Tangent
  {\pi_\M(h)} \M) + 2 \arcsin \frac{\rho}{\reach} \leq \theta$.
  Then,
  \[H \cap B(h,\frac{\rho}{4}) ~\subseteq~ \pi_H(\M \cap B(h,\frac{3\rho}{4})).\]
\end{lemma}


\begin{proof}
  Write $U = \M \cap B(h,\frac{3\rho}{4})$ and $V = H \cap
  B(h,\frac{\rho}{4})$; see Figure
  \ref{figure:projection-quasi-isometry}, right. We need to prove that
  $V \subseteq \pi_H(U)$. We start by establishing the following three
  propositions:
  \begin{enumerate}[label=\styleitem{(\alph*)},parsep=1pt,itemsep=1pt,topsep=4pt]
  \item \label{item:homeo} $\pi_H$ is a homeomorphism from $\M \cap B(h,\rho)$ to
    $\pi_H(\M \cap B(h,\rho))$;
  \item \label{item:proj-boundary} $\partial \pi_H(U) \cap V = \emptyset$;
  \item \label{item:proj-non-empty}  $\pi_H(U) \cap V \neq \emptyset$.
  \end{enumerate}

  \begin{figure}[htb]
    \begin{center}
      \includegraphics[width=0.5\linewidth]{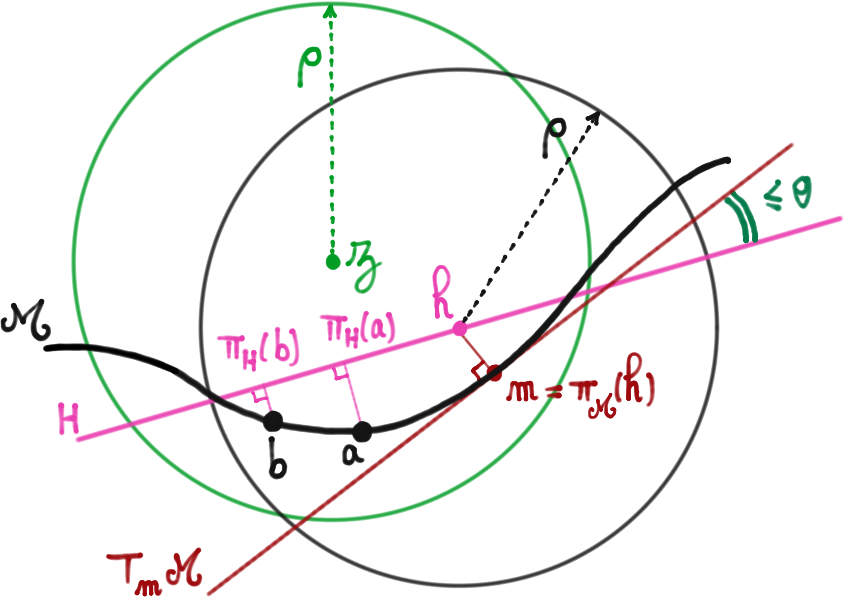}\hfill
      \includegraphics[width=0.4\linewidth]{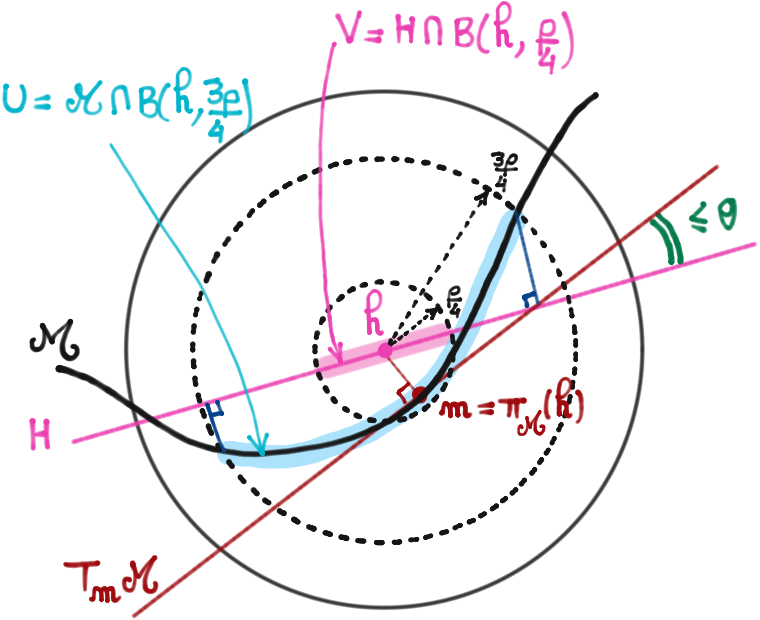}      
    \end{center}
    \caption{Notations for the proof of
      Lemma~\ref{lemma:projection-surjective}. \label{figure:projection-quasi-isometry}} 
  \end{figure}

  \medskip\noindent Let us prove Proposition \refitem{item:homeo}. For
  all $a,b \in \M \cap B(h,\rho)$, we start by bounding
  $\angle(H,ab)$. Letting $h^\ast = \pi_\M(h)$ and using Lemma \ref{lemma:tangdel:distanceToTangent} and Lemma \ref{lemma:tangdel:TangentSpaceVariation}, we obtain
  \begin{eqnarray*}
    \angle(H,ab) &\leq&  \angle(H,\Tangent {{h^\ast}} \M) + \angle(\Tangent {{h^\ast}} \M, \Tangent {a} \M) + \angle(\Tangent {a} \M, ab) \\
    &\leq& \angle(H,\Tangent {{h^\ast}} \M) + 2 \arcsin\left(\frac{\|a-{h^\ast}\|}{2\reach}\right) + \arcsin\left(\frac{\|a-b\|}{2\reach}\right)\\
    &\leq& \angle(H,\Tangent {h^\ast} \M) + \arcsin \left(\frac{\|a - h\| + \|h - h^\ast\|}{\reach} \right) +
    \arcsin \frac{\rho}{\reach}\\
    &\leq&  \angle(H,\Tangent {h^\ast} \M) + \arcsin \left(\frac{\rho + \frac{\rho}{4}}{\reach} \right) +
    \arcsin \frac{\rho}{\reach}\\
    &\leq& \angle(H,\Tangent {h^\ast} \M) + 2 \arcsin \frac{\rho}{\reach}\\
    &\leq&  \theta.
  \end{eqnarray*}
  Hence, for all $a,b \in \M \cap
  B(h,\rho)$:
  \begin{equation}
    \label{eq:proj-H}
    \cos\theta \times \|b-a\| ~~\leq~~ \cos\angle(H,ab) \times \|b-a\| = \| \pi_H(b) - \pi_H(a) \|
    ~~\leq~~ \| b - a\|,
  \end{equation}
  showing that the restriction of $\pi_H$ to $\M \cap B(h,\rho)$ is
  injective as soon as $\theta < \frac{\pi}{2}$. Thus, $\pi_H$ is a homeomorphism
  from $\M \cap B(z,\rho)$ to its range $\pi_H(\M \cap B(z,\rho))$.

   \medskip\noindent Let us prove Proposition
   \refitem{item:proj-boundary}. Because $\M \cap B(h,\rho)$ and
   $\pi_H(\M \cap B(h,\rho))$ are homeomorphic, we get in particular
   that $\partial \pi_H(U) = \pi_H(\partial U)$. Consider a point $u
   \in \partial U$, that is, a point $u \in \M$ such that $\|u-h\| =
   \frac{3\rho}{4}$ and let us prove that $\|\pi_H(u) - h \| >
   \frac{\rho}{4}$, in other words, that $\pi_H(u)$ not in $V$. By
   construction, both $u$ and $m = \pi_\M(h)$ belong to $\M \cap
   B(h,\rho)$ and thus $\angle(H,um) \leq \theta$. Using Equation
   \eqref{eq:proj-H} with $a=u$ and $b=m$, we get that $\|\pi_H(u) - \pi_H(m) \|
   \geq \cos \theta \times \|u-m\|$. We consider two cases:
    \begin{itemize}
    \item If $m=h$, we deduce immediately that $\|\pi_H(u) - h\| \geq
      \cos \theta \times \|u-h\| \geq \cos \frac{\pi}{6} \times
      \frac{3\rho}{4} > \frac{\rho}{4}$.
        
    \item If $m \neq h$, we claim that $\|\pi_H(u) - h\| >
      \frac{\rho}{4}$. To see this, denote by $\VectorSpace{A}$ the
      vector space associated to an affine space $A$ and let $V^\perp$
      designate the vector space orthogonal to a vector space $V$.
      Consider the straight-line $mh$ passing through $m$ and $h$.
      Note that $\pi_H(m)$ is also the orthogonal projection of $h$
      onto the affine space orthogonal to $H$ and passing through
      $m$. It follows that the vector $\pi_H(m) - m$ is the orthogonal
      projection onto $\VectorSpace{H}^\perp$ of the vector $h-m \in
      \VectorSpace{\Tangent m \M}^\perp$, so that $\angle(m \pi_H(m),mh)
      \leq \angle(\VectorSpace{H}^\perp,\VectorSpace{\Tangent m \M}^\perp)=
      \angle(H,\Tangent m \M) \leq \theta$. We thus get
    \begin{align*}
      \| \pi_H(u) - h \| &\geq \|\pi_H(u) - \pi_H(m) \| - \|h - \pi_H(m)\| \\
      &\geq \cos \theta \times \|u-m\| - \sin \theta \times \|h-m\| \\
      &\geq \cos \theta \times \left( \|u-h\| - \|m-h\| \right) - \sin \theta \times \|m - h\| \\
      &\geq \cos \theta \times \frac{3\rho}{4} - (\cos \theta + \sin \theta)\times \frac{\rho}{4} \\
      &\geq (2 \cos \frac{\pi}{6} - \sin \frac{\pi}{6}) \times \frac{\rho}{4}\\
      &> \frac{\rho}{4}.
    \end{align*}
    \end{itemize}
    Let us prove Proposition \refitem{item:proj-non-empty} by showing
    that $\pi_H(m) \in \pi_H(U) \cap V$. First, we show that $m \in U
    = \M \cap B(h,\frac{3\rho}{4})$. Because $\|h-m\| =
    d(h,\M) \leq \frac{\rho}{4}$, clearly $m \in \M \cap
    B(h,\frac{\rho}{4}) \subseteq U$. Second, we show that $\pi_H(m)
    \in V = H \cap B(h,\frac{\rho}{4})$. Since triangle $m h \pi_H(m)$
    has a right angle at vertex $\pi_H(m)$, the distance between any
    pair of points in this triangle is upper bounded by the lenght of its
    hypothenuse $mh$ and therefore, $\|h - \pi_H(m)\| \leq \|m - h \|
    \leq \frac{\rho}{4}$. Hence, $\pi_H(m) \in V$.

    \medskip \noindent
    We are now ready to conclude the second part of the proof. Since Propositions
    \refitem{item:proj-boundary} and \refitem{item:proj-non-empty}
    hold, we claim that $V \subseteq \pi_H(U)$. Indeed, suppose for a
    contradiction that $V \not \subseteq \pi_H(U)$. Then, we would be
    able to find two points $x$ and $y$ in $V$ such that $x$ lies
    inside $\pi_H(U)$ and $y$ lies outside $\pi_H(U)$. Consider a path
    connecting $x$ to $y$ in $V$ (for instance the segment with
    endpoints $x$ and $y$). This path would have to cross the
    boundary $\pi_H(U)$, contradicting the fact that the boundary of
    $\pi_H(U)$ lies outside $V$.
\end{proof}


\begin{lemma}[Small empty circumspheres]
  \label{lemma:bounding-circumradius}
  Assume $4\varepsilon < \rho < \frac{\reach}{3}$. Let $P$ to be an
  $\varepsilon$-dense sample of $\M$. Let $H \subseteq \Rspace^\Dim$
  be a $d$-dimensional affine space passing through a point $h$ such
  that $d(h,\M) \leq \frac{\rho}{4}$ and $\angle(H,\Tangent
  {\pi_\M(h)} \M) + 2 \arcsin \frac{\rho}{\reach} \leq
  \frac{\pi}{6}$. Then,
  \begin{itemize}
  \item $h$ lies in the relative interior of $\Conv{\pi_H(P \cap B(h,\rho))}$.
  \item For any $d$-simplex $\sigma \subseteq P$ such that $h \in
    \pi_H(\Conv \sigma)$ and $\pi_H(\sigma) \in \Del{\pi_H(P \cap
      B(h,\rho))}$, we have that $R(\pi_H(\sigma)) \leq
      \varepsilon$.
  \end{itemize}
\end{lemma}




\begin{proof}
  Let $Q = P \cap B(h,\rho)$ and $Q' = \pi_H(Q)$. The two items follow
  from a claim that we make: for all $r \in
  (\varepsilon,\frac{\rho}{4})$, any $d$-ball of radius $r$ contained
  in $H$ and covering $h$ must contain in its interior some point of
  $Q'$. Suppose for a contradiction that this in not the case and let
  $H \cap B(c,r)$ be a $d$-ball covering $h$ and containing no point
  of $Q'$ in its interior. Notice that the center $c$ of this $d$-ball
  belongs to $H \cap B(h,r)$ and since $r < \frac{\rho}{4}$,
  Lemma~\ref{lemma:projection-surjective} entails that
  \[
  c \in H \cap B(h,r) \subseteq \pi_H(\M \cap B(h,\rho - r)).
  \]
  Hence, there would exist $m \in \M \cap B(h,\rho - r)$
  such that $\pi_H(m) = c$ and therefore $p \in P \cap B(h,\rho)$ such
  that $\|p-m\| \leq \varepsilon$ and consequently such that
  $\|\pi_H(p)-c\| \leq \varepsilon$. Thus, we would have a point $p \in P
  \cap B(h,\rho)$ whose projection onto H is contained in the interior
  of $B(c,r)$, which contradicts our claim.

  Let us prove that $h$ lies in the relative interior of
  $\Conv{Q'}$. Suppose for a contradiction that
  this is not the case. Then, we would be able to find an open
  $d$-dimensional half-space of $H$ whose boundary passes through $h$
  and which avoids $Q'$, contradicting our claim.

  Consider now a $d$-simplex $\sigma \subseteq P$ such that $h \in
  \pi_H(\Conv \sigma)$ and $\sigma' = \pi_H(\sigma) \in
  \Del{Q'}$. Because $\sigma'$ is a Delaunay simplex, $S(\sigma')$ is well-defined. Write $Z' =
  Z(\sigma')$ and $R' = R(\sigma')$. Since $\sigma' \in \Del{Q'}$,
  this means that no point $p \in P \cap B(h,\rho)$ has a projection
  onto H that is contained in the interior of $B(Z',R')$. Let us prove
  that $R' \leq \varepsilon$. Suppose for a contradiction that $R' >
  \varepsilon$. Noting that $h \in \pi_H(\Conv \sigma) =
  \Conv{\pi_H(\sigma)} = \Conv{\sigma'} \subseteq B(Z',R')$ and
  letting $r \in (\varepsilon,\frac{\rho}{4})$, we would be able to
  find a $d$-ball of radius $r$ contained in the $d$-ball $H \cap
  B(Z',R')$, covering $h$ and containing no point of $Q'$, hence
  contradicting our claim.
\end{proof}


\begin{lemma}
  \label{lemma:surrounded}
  Suppose that $4\varepsilon < \rho < \frac{\reach}{3}$ and $2
  \arcsin \frac{\rho}{\reach} \leq \frac{\pi}{6}$. Let $P$ be an
  $\varepsilon$-sample of \M. Then, for all $m \in \M$,
  the domain $\US{\Delstar m \rho}$ is homeomorphic to $\Rspace^d$.
\end{lemma}

  \begin{proof}
    Applying Lemma \ref{lemma:bounding-circumradius} with $(h,H) = (m,
    \Tangent m \M)$, we get that each point $m \in \M$ lies in the
    relative interior of $\Conv{ \prTangent m {P \cap B(m,\rho)}
    }$. Hence, $\US{\Delstar m \rho}$ contains $m$ in its relative
    interior and the result follows.
\end{proof}

\section{Stability of Delaunay simplices through distortions}
\label{section:stability-through-distortions}

The goal of this section is to establish a technical lemma
  (Lemma~\ref{lemma:stability-from-plane-zero-to-plane-one}) which
  provides conditions under which a $d$-simplex $\sigma$ is Delaunay
  stable for $P$ at scale $\rho$ with respect to the set $\{ (h_i,H_i)
  \}_{i \in \{0,1\}}$, where $h_i$ is a point of $\Rspace^\Dim$ and
  $H_i \subseteq \Rspace^\Dim$ a $d$-dimensional space passing through
  $h_i$. Recall that a simplex $\sigma$ is Delaunay stable for $P$ at
  scale $\rho$ with respect to $\{ (h_i,H_i) \}_{i \in \{0,1\}}$ if
  the following two propositions are equivalent:
\begin{itemize}
\item $\sigma \subseteq P \cap B(h_0, \rho)$ and $\pi_{H_0}(\sigma) \in \Del{\pi_{H_0}(P \cap B(h_0,\rho))}$;
\item $\sigma \subseteq P \cap B(h_1, \rho)$ and $\pi_{H_1}(\sigma) \in \Del{\pi_{H_1}(P \cap B(h_1,\rho))}$.
\end{itemize}
Letting $\sigma_i = \pi_{H_i}(\sigma)$ and $Q_i = \pi_{H_0}(P \cap
B(h_0,\rho))$, we thus have to answer the following question: under
which conditions do we have $\sigma_0 \in \Del{Q_0} \iff \sigma_1 \in
\Del{Q_1}$?  We find that binary relations are the right concept to
compare Delaunay complexes $\Del{Q_0}$ and $\Del{Q_1}$ when $P$
  is noisy. In Section~\ref{section:general-distortions}, building on
  the work of Boissonnat et al. \cite{BOISSONNAT_2013}, we first
  consider a general binary relation over sets $Q_0$ and $Q_1$ and
  find that this relation must be a ``sufficiently'' small distortion
  to ensure the equivalence $\sigma_0 \in \Del{Q_0} \iff \sigma_1 \in
  \Del{Q_1}$
  (Lemma~\ref{lemma:delaunay-stability-through-distortion}). In
  Section~\ref{section:specific-distortions}, we then turn our
  attention to some specific restrictions of the binary relation
  $\{(\pi_{H_0}(p), \pi_{H_1}(p)) \mid p \in P\}$ and quantify their
  distortion (Lemma
  \ref{lemma:perfect-relation-from-plane-zero-to-plane-one} and
  Lemma \ref{lemma:relation-from-plane-zero-to-plane-one}). In
  Section~\ref{section:technical-lemma}, we state and prove our
  technical lemma.

\subsection{General distortions}
\label{section:general-distortions}

Recall that a {\em (binary) relation} $\Rrel$ over sets $X_0$ and
$X_1$ is a subset of the Cartesian product $X_0 \times X_1$. The {\em
  range} of $\Rrel$, denoted as $\Range{\Rrel}$ is the set of all $x_1
\in X_1$ for which there exists at least one $x_0 \in X_0$ such that
$(x_0,x_1) \in \Rrel$. The {\em domain} of $\Rrel$, denoted as
$\Domain{\Rrel}$ is the set of all $x_0 \in X_0$ for which there
exists at least one $x_1 \in X_1$ such that $(x_0,x_1) \in \Rrel$.

\begin{definition}[Multiplicative distortion]
  We say that $\Rrel$ is a {\em multiplicative $M$-distortion} for
  some $M \geq 0$ if for all $(x_0,x_1), (y_0, y_1) \in \Rrel$, we have
  \[
  \frac{1}{1+M} \| x_0 - y_0 \| \leq \|x_1 - y_1\| \leq (1 + M) \| x_0 - y_0 \|.
  \]
\end{definition}

\begin{remark}
  \label{remark:multiplicative-distortion-one-to-one}
  Notice that a multiplicative $M$-distortion $\Rrel$ is injective
  because for all $(x_0,x_1), (y_0, y_1) \in \Rrel$, the following
  implication holds: $x_1 = y_1 \implies x_0 = y_0$. Also, if $\Rrel$
  is a multiplicative $M$-distortion, so is the converse relation
  $\Rrel^{-1} = \{ (x_1,x_0) \mid (x_0, x_1) \in \Rrel \}$. Hence, the
  converse relation $\Rrel^{-1}$ is injective, or equivalently,
  $\Rrel$ is functional. Thus, $\Rrel$ being both injective and
  functional is one-to-one.
\end{remark}

\begin{definition}[Additive distortion]
  We say that $\Rrel$ is an {\em additive $A$-distortion} for some $A$
  if for all $(x_0,x_1), (y_0, y_1) \in \Rrel$ we have
  \[
  \abs{ \|x_1 - y_1\| - \|x_0-y_0\| } \leq A
  \]
\end{definition}

\begin{lemma}[Going from  multiplicative to additive, and vice versa]
  \label{lemma:relation-multiplicative-additive}
  Consider a relation $\Rrel$ over sets $X_0$ and $X_1$.
  \begin{itemize}
  \item
    If $\Rrel$ is a multiplicative $M$-distortion for some $M
    \geq 0$, then $\Rrel$ is an additive $A$-distortion for any $A \geq M \times \Diam{X_0}$.
  \item
    If $\Rrel$ is an additive $A$-distortion map for some $A < \Sep{X_0}$, then
    $\phi$ is a multiplicative $M$-distortion map for any $M \geq \frac{A}{\Sep{X_0}-A}$.
  \end{itemize}
\end{lemma}

\begin{proof}
  To show the first part of the lemma, suppose that $\Rrel$ is a
  multiplicative $M$-distortion for some $M \geq 0$. For all $(x_0,
  x_1), (y_0,y_1) \in \Rrel$, we thus have
  \[
  \frac{1}{1+M} \|x_0 - y_0\| \leq \|x_1 - y_1\| \leq (1+M) \|x_0 - y_0\|.
  \]
  Subtracting from each side $\|x_0 - y_0\|$, we get that
  \[
  -M \|x_0 - y_0\| \leq \frac{-M}{1+M} \|x_0 - y_0\| \leq \|x_1 - y_1\| - \|x_0 - y_0\| \leq M \|x_0 - y_0\|.
  \]
  and therefore
  \[
  \abs{\,\|x_1 - y_1\| - \|x_0 - y_0\|\,} \leq M \|x_0 - y_0\| \leq M \times \Diam{X_0},
  \]
  showing the first part of the lemma. To establish the second part of
  the lemma, set $S = \Sep{X_0}$ and suppose that $\Rrel$ is an
  additive $A$-distortion map for some $A < S$. Then, for all $(x_0,
  x_1), (y_0,y_1) \in \Rrel$, we have by definition that
  \[
  \|x_0 - y_0\| - A \leq \| x_1 - y_1 \| \leq \|x_0 - y_0\| + A
  \]
  Rearranging the left and right sides and using $S < \|x_0 - y_0\|$,
  we get that
  \[
  \left(\frac{1}{1 + \frac{A}{S-A}} \right) \|x_0 - y_0\|  = \left( 1 - \frac{A}{S} \right) \|x_0 - y_0\| \leq \| x_1 - y_1\| \leq  \left( 1 + \frac{A}{S} \right) \|x_0 - y_0\|.
  \]
  For any $M \geq \frac{A}{S-A} \geq \frac{A}{S}$, we thus get that
  \[
  \frac{1}{1+M} \| x_0 - y_0 \| \leq \|x_1 - y_1\| \leq (1 + M) \|x_0 - y_0 \|,
  \]
  showing that $\Rrel$ is an $M$-distortion map. This proves the second part of the lemma.
\end{proof}

Let us recall a nice result \cite[Lemma 4.1]{BOISSONNAT_2013} which
bounds the displacement that undergoes the circumcenter of a simplex
when its vertices are perturbed.

\begin{lemma}[Location of almost circumcenters {\cite[Lemma 4.1]{BOISSONNAT_2013}}]
  \label{lemma:locating-centers}
  Let $X \subseteq \Rspace^\Dim$ be a $d$-dimensional affine space. If
  $\sigma \subseteq X$ is a $d$-simplex, and $x \in X$ is such that
  \[
  \abs{ \|x - a\|^2 - \|x - a'\|^2 } \leq \xi^2 \quad \text{for all $a,a' \in \sigma$},
  \]
  then
  \[
  \|Z(\sigma) - x\| \leq \frac{d \xi^2}{2\height\sigma}.
  \]
\end{lemma}

Notice that the above bound becomes meaningless as the simplex $\sigma$
becomes degenerate because then the right side of the inequality tends
to $+\infty$. Applying the above lemma in our context, we get the
following lemma:

\begin{lemma}[Stability of Delaunay simplices through distortion]
  \label{lemma:delaunay-stability-through-distortion}
Let $H_0$ and $H_1$ be two $d$-dimensional affine spaces in
$\Rspace^\Dim$. Consider a binary relation $\Rrel \subseteq H_0 \times
H_1$ and suppose that $\Rrel$ is an additive $A$-distortion for some
$A \geq 0$. Let $\Qrel \subseteq \Rrel$ be a finite one-to-one
relation. Let $Q_0 = \Domain{\Qrel}$ and $Q_1 =
\Range{\Qrel}$. Consider  $\Srel \subseteq \Qrel$
such that both $\sigma_0 = \Domain{\Srel}$ and $\sigma_1 =
\Range{\Srel}$ are non-degenerate abstract $d$-simplices. Suppose
$\sigma_0$ is $\zeta$-protected with respect to $Q_0$.  Suppose there
exists $\varepsilon \geq 0$ such that for $i \in \{ 0, 1\}$
\begin{align*}
  2 A \left( 1 + \frac{2 d \varepsilon}{\height{\sigma_i}} \right) < \zeta.
\end{align*}
Then, we have the following two implications:
\begin{itemize}
\item $R(\sigma_0) \leq \varepsilon$, $Z(\sigma_0) \in \Domain{\Rrel}$ and $\sigma_0 \in \Del{Q_0}$ $\implies$
  $\sigma_1 \in \Del{Q_1}$ and is protected with respect to $Q_1$; 
\item $R(\sigma_1) \leq \varepsilon$, $Z(\sigma_1) \in \Range{\Rrel}$ and $\sigma_1 \in \Del{Q_1}$
  $\implies$ $\sigma_0 \in \Del{Q_0}$.
\end{itemize}
\end{lemma}


\begin{proof}
   Suppose first that $R(\sigma_0) \leq \varepsilon$, $Z(\sigma_0) \in
   \Domain{\Rrel}$ and $\sigma_0 \in \Del{Q_0}$ and let us prove that
   $\sigma_1 \in \Del{Q_1}$ and is protected with respect to $Q_1$. In other words, we need to prove that for
   all $(a_0,a_1) \in \Srel$ and all $(q_0,q_1) \in \Qrel \setminus
   \Srel$, we have $\|a_1 - Z(\sigma_1)\| < \|q_1 -
   Z(\sigma_1)\|$. Let $z_1 \in \Range{\Rrel}$ such that $(Z(\sigma_0),z_1) \in
   \Rrel$. On one hand, for all $(a_0,a_1) \in \Srel$, we have:
  \begin{align}
    \|a_1 - Z(\sigma_1)\|
    &\leq \|a_1 - z_1 \| + \|z_1 - Z(\sigma_1)\| \notag\\
    &\leq A + \|a_0 - Z(\sigma_0)\| + \|z_1 - Z(\sigma_1)\| \notag\\
    &\leq A + R(\sigma_0) + \|z_1 - Z(\sigma_1)\|. \label{eq:circumradius-from-above-additive}
  \end{align}
  On the other hand, for all $(q_0,q_1) \in \Qrel \setminus \Srel$, we have:
  \begin{align}
    \|q_1 - Z(\sigma_1)\|
    &\geq \|q_1 - z_1\| - \|z_1 - Z(\sigma_1)\| \notag\\
    &\geq \|q_0 - Z(\sigma_0)\| - A - \|z_1 - Z(\sigma_1)\| \notag\\
    &\geq R(\sigma_0) + \zeta - A - \|z_1 - Z(\sigma_1)\|. \label{eq:circumradius-from-below-additive}
  \end{align}

  Thus, we obtain that $\|a_1 - Z(\sigma_1)\| < \|q_1 - Z(\sigma_1)\|$ as soon
  as the right side of (\ref{eq:circumradius-from-above-additive}) is smaller
  than the right side of
  (\ref{eq:circumradius-from-below-additive}), that is, as soon as:
  \begin{equation}
    \label{eq:condition-del-additive}
    2\|z_1-Z(\sigma_1)\| + 2A < \zeta.
  \end{equation}
  Because for all $a_1 \in \sigma_1$ we have $R(\sigma_0) - A \leq
  \|z_1-a_1\| \leq R(\sigma_0) + A$, we get that for all
  $a_1,a_1' \in \sigma_1$:
  \[
  \abs{ \|z_1-a_1\|^2 - \|z_1-a_1'\|^2 } \leq (R(\sigma_0) + A)^2 - (R(\sigma_0) - A)^2 = 4 A R(\sigma_0).
  \]
  Applying Lemma~\ref{lemma:locating-centers}, we obtain that
  \[
  \|z_1-Z(\sigma_1)\| \leq
  \frac{2 A d R(\sigma_0)}{\height{\sigma_1}}.
  \]
  Using this inequality, we get that Inequality
  (\ref{eq:condition-del-additive}) holds as soon as
  \[
  2 A \left( 1 + \frac{2 d R(\sigma_0)}{\height{\sigma_1}} \right) < \zeta
  \]
  which follows directly from our assumptions. Thus, $\sigma_1 \in
  \Del{Q_1}$.  Suppose now that $R(\sigma_1) \leq \varepsilon$,
  $Z(\sigma_1) \in \Range{\Rrel}$ and $\sigma_0 \not \in \Del{Q_0}$
  and let us prove that $\sigma_1 \not \in \Del{Q_1}$. Because
  $\sigma_0 \not \in \Del{Q_0}$, there exists $(q_0 ,q_1) \in \Qrel$
  such that $\|q_0 - Z(\sigma_0)\|<R(\sigma_0)$ and because $\sigma_0$
  is $\zeta$-protected with respect to $Q_0$, we have
  $\|q_0-Z(\sigma_0)\| < \|a_0-Z(\sigma_0)\| - \zeta$ for all pairs
  $(a_0,a_1) \in \Srel$. Let us prove that $\|q_1 - Z(\sigma_1)\| <
  \|a_1 - Z(\sigma_1)\|$ for any $(a_0,a_1) \in \Srel$. Let $z_0 \in \Domain{\Rrel}$ such that
  $(z_0,Z(\sigma_1)) \in \Rrel$. On one hand, we have:
  \begin{align}
    \|q_1 - Z(\sigma_1)\| &\leq \|q_0 - z_0 \| + A \notag \\
    &\leq \| z_0 - Z(\sigma_0) \| + \| Z(\sigma_0) - q_0 \| + A \notag \\
    &\leq \| z_0 - Z(\sigma_0) \| + R(\sigma_0) - \zeta + A \label{eq:phi-circumradius-from-above-additive}
  \end{align}
  On the other hand, for any $(a_0,a_1) \in \Srel$, we have:
  \begin{align}
    \|a_1-Z(\sigma_1)\|
    &\geq \| a_0 - z_0 \| - A \notag\\
    &\geq  \|a_0 - Z(\sigma_0)\| - \| z_0 - Z(\sigma_0) \|  - A \notag\\
    &\geq  R(\sigma_0) - \| z_0 - Z(\sigma_0) \|  - A \label{eq:phi-circumradius-from-below-additive}
  \end{align}
  Thus, we obtain that $\|q_1 - Z(\sigma_1)\| < \|a_1 - Z(\sigma_1)\|$ as soon as the
  right side of (\ref{eq:phi-circumradius-from-above-additive}) is smaller than the right side of
  (\ref{eq:phi-circumradius-from-below-additive}), that is, as soon as:
  \begin{equation}
  \label{eq:condition-not-del-additive}
  2 \|z_0 - Z(\sigma_0)\| + 2A < \zeta.
  \end{equation}
  Because for all $(a_0,a_1) \in \Srel$, we have
  $R(\sigma_1) - A \leq \| z_0 - a_0 \| \leq R( \sigma_1) + A$, we get that for all $a_0, a_0' \in \sigma_0$:
  \[
  \abs{ \| z_0 - a_0 \|^2 - \| z_0 - a_0'\|^2} \leq (R(\sigma_1)+A)^2 - (R(\sigma_1)-A)^2
  = 4 A R(\sigma_1).
  \]
  Applying Lemma \ref{lemma:locating-centers}, we obtain that
  \[
  \| z_0 - Z(\sigma_0) \| \leq \frac{2 A d R(\sigma_1)}{\height{\sigma_0}}
  \]
  Using this inequality, we get that Inequality
  (\ref{eq:condition-not-del-additive}) holds as soon as
  \[
  2 A\left( 1 + \frac{2 d R(\sigma_1)}{\height{\sigma_0}} \right) < \zeta
  \]
  which follows directly from our assumptions. Thus, $\sigma_1
  \not \in \Del{Q_1}$.
\end{proof}


\subsection{Specific distortions}
\label{section:specific-distortions}

We now consider relations of the form $\Rrel = \{
  (\pi_{H_0}(x),\pi_{H_1}(x)) \mid x \in X\}$ and find values of $A$
  for which $\Rrel$ is an additive $A$-distortion. We consider first
  the case of a set $X$ contained in \M in Lemma
  \ref{lemma:perfect-relation-from-plane-zero-to-plane-one} (non-noisy
  case) before handling the case of a set $X$ contained in $\Offset M
  \delta$ in Lemma \ref{lemma:relation-from-plane-zero-to-plane-one}
  (noisy case). 

\begin{lemma}[Distortion in the non-noisy case]
  \label{lemma:perfect-relation-from-plane-zero-to-plane-one}
  Consider a subset $U \subseteq \M$ and two $d$-dimensional spaces
  $H_0$ and $H_1$. Suppose that there is $\theta \leq 1$ such that for
  $i \in \{0,1\}$
  \[
   \sup_{u,u' \in U} \angle(H_i,uu') \leq  \theta.
   \]
   Then, 
   $
   \Rrel = \{ (\pi_{H_0}(u),\pi_{H_1}(u)) \mid u \in U\}
   $
   is an additive $(\Diam{U}\times\theta^2)$-distortion.
\end{lemma}

\begin{proof}
  Note that for all
  $u,u' \in U$: 
  \begin{alignat*}{2}
    \cos \theta \times \| u' - u \| &~~\leq~~ \| \pi_{H_0}(u') - \pi_{H_0}(u) \| &~~\leq~~ \| u' - u\|, \\
    \cos \theta \times \| u' - u \| &~~\leq~~   \| \pi_{H_1}(u') - \pi_{H_1}(u) \| &~~\leq~~ \| u' - u \|.
  \end{alignat*}
  Hence, for all $u, u' \in U$,
  \[
  \cos \theta \times \| \pi_{H_0}(u') - \pi_{H_0}(u) \| \leq \| \pi_{H_1}(u') - \pi_{H_1}(u)
  \| \leq \frac{1}{\cos \theta} \times \| \pi_{H_0}(u') - \pi_{H_0}(u) \|
  \]
  Thus, $\Rrel$ is a multiplicative $\left(\frac{1 - \cos \theta}{\cos
    \theta}\right)$-distortion.  Noting that for all $t$, we
  have $1 - \cos(t) \leq \frac{t^2}{2}$ and using $\theta \leq 1$, we
  obtain that $\frac{1 - \cos\theta}{\cos\theta} \leq
  \frac{\theta^2}{2 - \theta^2} \leq \theta^2$ and therefore $\Rrel$
  is a multiplicative $\theta^2$-distortion. Applying Lemma
  \ref{lemma:relation-multiplicative-additive}, it follows that $\Rrel$
  is an additive $(\Diam{U} \times \theta^2)$-distortion.
\end{proof}

\begin{lemma}[Distortion in the noisy case]
  \label{lemma:relation-from-plane-zero-to-plane-one}
  Suppose $P \subseteq \Offset \M \delta$ for some $\delta \leq \frac{\reach}{2}$. Consider a point $z
  \in \Rspace^\Dim$ and two $d$-dimensional spaces $H_0$ and
  $H_1$. Suppose that there is $\theta \leq 1$ such that for $i \in
  \{0,1\}$
  \[
   \sup_{m,m' \in \pi_\M(\Offset \M \delta \cap B(z,\rho))} \angle(H_i,mm') \leq  \theta.
   \]
   Then, the binary relation $ \Rrel = \{(\pi_{H_0}(a), \pi_{H_1}(a))
   \mid a \in \Offset \M \delta \cap B(z,\rho) \} $ is an additive
   $A$-distortion for $A = 4\delta\theta+4\rho\theta^2$. If
   furthermore $\Sep P > 2 A + 6 \delta$, the
   restricted relation $ \Qrel = \{(\pi_{H_0}(p), \pi_{H_1}(p)) \mid
   p \in P \cap B(z,\rho) \} $ is one-to-one.
\end{lemma}


\begin{proof}
  Whenever the projection of a point $a \in \Rspace^\Dim$ onto $\M$ is well-defined,
  let us write $a^* = \pi_\M(a)$ for short. Observe that for all $a
  \in \Offset \M \delta \cap B(z,\rho)$ and all $i \in \{0,1\}$,
  $\angle(H_i, \Tangent {a^*} \M) \leq \theta$ and consequently,
  \begin{equation*}
  \|\pi_{H_i}(a) - \pi_{H_i}(a^*)\| \leq \delta \sin \theta.  
  \end{equation*}
  Let us bound from above the diameter of the set $U = \pi_\M(\Offset
  \M \delta \cap B(z,\rho))$. We know from \cite[page 435]{federer-59}
  that for $0 \leq \delta < \reach$ the projection map $\pi_\M$
  onto $\M$ is $\left( \frac{\reach}{\reach - \delta}
  \right)$-Lipschitz for points at distance less than $\delta$ from
  $\M$. For any two points $a,b \in \Offset \M \delta \cap B(z,\rho)$,
  we thus have
  \begin{equation*}
    \| a^* - b^* \| \leq \frac{\reach}{\reach - \delta} \times \| a - b\| \leq 4 \rho
  \end{equation*}
  and therefore $\Diam{U} \leq 4 \rho$. Applying Lemma
  \ref{lemma:perfect-relation-from-plane-zero-to-plane-one}, we get that for all
  $a,b \in \Offset \M \delta \cap B(z,\rho)$:
  \begin{equation*}
    \abs{ \| \pi_{H_1}(a^*) - \pi_{H_1}(b^*) \| - \| \pi_{H_0}(a^*) - \pi_{H_0}(b^*) \| } \leq 4 \rho \theta^2.
  \end{equation*}
  Let us introduce $\Delta_i = \| \pi_{H_i}(a) - \pi_{H_i}(b) \|$ and
  $\Delta_i^* = \| \pi_{H_i}(a^*) - \pi_{H_i}(b^*) \|$. We have
  \[
  | \Delta_i - \Delta_i^* | \leq \|\pi_{H_i}(a) - \pi_{H_i}(a^*)\| + \|\pi_{H_i}(b) - \pi_{H_i}(b^*)\| \leq 2 \delta \theta
  \]
  and therefore $| \Delta_1 - \Delta_0 | \leq | \Delta_1 - \Delta_1^*
  | + | \Delta_1^* - \Delta_0^* | + | \Delta_0^* - \Delta_0 | \leq 4
  \delta \theta + 4 \rho \theta^2$. It follows that $\Rrel$ is an
  additive $A$-distortion for $A=4\delta\theta+4\rho\theta^2$
  and so is its restricted relation $\Qrel$. Writing $Q_0 =
  \Domain{\Qrel}$, we now suppose in addition that $\Sep P > 2A+6\delta$ and deduce that $\Sep{Q_0} > A$. Using
  $\cos\theta \geq \frac{1}{2}$, we get that for all $a,b \in P \cap
  B(z,\rho)$
  \begin{align*}
    \| \pi_{H_0}(a) - \pi_{H_0}(b) \| &\geq \| \pi_{H_0}(a^*) - \pi_{H_0}(b^*) \| -  \|\pi_{H_0}(a) - \pi_{H_0}(a^*)\|
    - \|\pi_{H_0}(b^*) - \pi_{H_0}(b)\| \\
    &\geq  \|a^*-b^*\| \cos \theta - 2 \delta \theta \\
    &\geq \left( \|a - b\| - \|a - a^*\| - \|b - b^*\| \right) \cos \theta  - 2 \delta \theta \\
    &\geq \frac{1}{2} \left( \|a -b\| - 2 \delta \right) - 2\delta\theta.
  \end{align*}
  Thus, $\Sep{Q_0} \geq \frac{1}{2} \Sep{P} - 3 \delta >
  A$. Applying Lemma \ref{lemma:relation-multiplicative-additive}, we
  get that $\Qrel$ is a multiplicative $\Psi$-distortion for
  $\Psi = \frac{A}{\Sep{Q_0}-A}$ and using Remark~\ref{remark:multiplicative-distortion-one-to-one}, we conclude that
  $\Qrel$ is one-to-one.
\end{proof}


\subsection{Technical lemma}
\label{section:technical-lemma}

The next lemma provides conditions under which $\sigma$ is Delaunay
stable for $P$ at scale $\rho$ with respect to $\{(h_i,H_i)\}_{i \in
  \{0,1\}}$. Roughly speaking, our conditions say that for each pair
$(h_i, H_i)$, we need $h_i$ to be ``close'' to $\M$, $h_0$ and $h_1$
to be ``close'' to one another and $H_i$ to make a
``small'' angle with \M ``near'' $\Conv \sigma$.  Precisely:

\begin{lemma}[Technical lemma]
  \label{lemma:stability-from-plane-zero-to-plane-one}
  Let $\delta \geq 0$, $0 \leq \varepsilon \leq \frac{\rho}{16}$, $0
  \leq \theta \leq \frac{\pi}{6}$ and $A = 4\delta\theta + 4
  \rho\theta^2$ and assume that $\rho + \delta <
  \frac{\reach}{3}$. Suppose that $P \subseteq \Offset \M \delta$, $\M
  \subseteq \Offset P \varepsilon$ and $\Sep{P} > 2 A + 6
  \delta$. Consider a $d$-simplex $\sigma \subseteq P$, a
  $d$-dimensional space $H_0$ passing through a point $h_0$ and a
  $d$-dimensional space $H_1$ passing through a point $h_1$. For $i
  \in \{ 0, 1 \}$, write $\sigma_i = \pi_{H_i}(\sigma)$. Suppose that
  $\sigma_0$ is $\zeta$-protected with respect to $\pi_{H_0}(P \cap
  B(h_0,2\rho))$ and assume furthermore that the
  following hypotheses are satisfied:
  \begin{enumerate}[label=\styleitem{(\arabic*)},parsep=1pt,itemsep=1pt,topsep=4pt]
  \item \label{hypo:non-degenerate} For $i \in \{0,1\}$, $\sigma_i$ has dimension $d$;
  \item \label{hypo:cin-conv} For $i \in \{0,1\}$, $h_i \in \Conv(\sigma_i)$;
  \item \label{hypo:close-manifold} For $i \in \{0,1\}$, $d(h_i,\M) \leq \frac{\rho}{4}$;
  \item \label{hypo:angle}
    For $0 \leq i,j \leq 1$, $  \sup_{m,m' \in \pi_\M(\Offset \M \delta \cap B(h_j,\rho))} \angle(H_i,mm') \leq \theta$.
  \item \label{hypo:small-circumradius} $\|h_0-h_1\| \leq 4 \varepsilon$ whenever $R(\sigma_0) \leq
    \varepsilon$ or $R(\sigma_1) \leq \varepsilon$;
  \item \label{hypo:simplex-enclosed-in-ball} For $0 \leq i,j \leq 1$ with $i \neq j$, the following holds:
    $R(\sigma_i) \leq \varepsilon \implies \sigma \subseteq
    B(h_{j},\rho)$;
  \item \label{hypo:protection}  $2 A \left( 1 + \frac{2 d \varepsilon}{\height{\sigma_i}}
      \right) < \zeta$ for $i \in \{0,1\}$.
  \end{enumerate}
  Then, $\sigma$ is Delaunay stable for $P$ at scale $\rho$ with respect to $\{(h_i,H_i)\}_{i \in \{0,1\}}$. Equivalently, the following two propositions are equivalent:
  \begin{itemize}
  \item $\sigma \subseteq P \cap B(h_0, \rho)$ and $\sigma_0 \in \Del{\pi_{H_0}(P \cap B(h_0,\rho))}$;
  \item $\sigma \subseteq P \cap B(h_1, \rho)$ and $\sigma_1 \in \Del{\pi_{H_1}(P \cap B(h_1,\rho))}$.
  \end{itemize}
  Furthermore, whenever one of the two above propositions
    holds, $\sigma_1$ is protected with respect to $\pi_{H_1}(P \cap B(h_1,\rho))$, $R(\sigma_0) \leq \varepsilon$ and $R(\sigma_1) \leq \varepsilon$.
\end{lemma}


\begin{proof}
  We prove the lemma by showing that the following four propositions are equivalent:
  \begin{enumerate}[label=\styleitem{(\alph*)},parsep=1pt,itemsep=1pt,topsep=4pt]
  \item \label{item:del-X-x} $\sigma \subseteq P \cap B(h_0, \rho)$ and $\sigma_0 \in \Del{\pi_{H_0}(P \cap B(h_0,\rho))}$;

  \item \label{item:del-X-y}  $\sigma \subseteq P \cap B(h_1, \rho)$,
    $\sigma_0 \in \Del{ \pi_{H_0}(P \cap B(h_1,\rho))}$ and $R(\sigma_0) \leq \varepsilon$;

  \item \label{item:del-Y-y} $\sigma \subseteq P \cap B(h_1, \rho)$
    and $\sigma_1 \in \Del{\pi_{H_1}(P \cap B(h_1,\rho))}$ with
      $\sigma_1$ being protected with respect to $\pi_{H_1}(P \cap B(h_1,\rho))$;

  \item \label{item:del-Y-x} $\sigma \subseteq P \cap B(h_0, \rho)$,
    $\sigma_1 \in \Del{ \pi_{H_1}(P \cap B(h_0,\rho))}$ and $R(\sigma_1) \leq \varepsilon$.

  \end{enumerate}


  \smallskip \noindent Let us prove \refitem{item:del-X-x} $\implies$
  \refitem{item:del-X-y}.  Suppose $\sigma \subseteq P \cap B(h_0,
  \rho)$ and $\sigma_0 \in \Del{\pi_{H_0}(P \cap B(h_0,\rho))}$.
  Applying Lemma~\ref{lemma:bounding-circumradius} with $(H,h) =
  (H_0,h_0)$, we obtain that $R(\sigma_0) \leq \varepsilon$.
  Using $\|h_0 - h_1\| \leq 4 \varepsilon$ and
  $\|Z(\sigma_0)-h_0\| \leq R(\sigma_0) \leq \varepsilon$, we obtain
  \[
  B(Z(\sigma_0),R(\sigma_0)) \subseteq B(h_0,2\varepsilon) \subseteq B(h_1, 6\varepsilon)
  \subseteq B({h_0},\rho) \cap B(h_1,\rho)
  \]
  and therefore $\sigma_0 \in \Del{\pi_{H_0}(P \cap
    B({h_1},\rho))}$. Since $R(\sigma_0) \leq \varepsilon$, our sixth
  hypothesis implies $\sigma \subseteq B(h_0,\rho)$.  This proves
  \refitem{item:del-X-x} $\implies$ \refitem{item:del-X-y}.


  \smallskip \noindent Let us prove \refitem{item:del-X-y} $\implies$
  \refitem{item:del-Y-y}. Suppose $\sigma \subseteq P \cap
  B(h_1,\rho)$, $\sigma_0 \in \Del{\pi_{H_0}(P \cap B(h_1,\rho))}$ and
  $R(\sigma_0) \leq \varepsilon$. Consider the relations
  \begin{align*}
  \Rrel &= \{(\pi_{H_0}(a), \pi_{H_1}(a)) \mid a \in \Offset \M \delta
  \cap B(h_1,\rho) \},\\
  \Qrel &= \{(\pi_{H_0}(p), \pi_{H_1}(p)) \mid p \in P \cap B(h_1,\rho) \},\\
  \Srel &= \{(\pi_{H_0}(v), \pi_{H_1}(v)) \mid v \in \sigma \}.
  \end{align*}
  Let $Q_0 = \pi_{H_0}(P \cap B(h_1,\rho))$ and $Q_1 = \pi_{H_1}(P
  \cap B(h_1,\rho))$. By construction, $Q_0 = \Domain \Qrel$, $Q_1 =
  \Range \Qrel$, $\sigma_0 = \Domain \Srel$ and $\sigma_1 = \Range
  \Srel$.  Note that $\|h_1 - h_i\| \leq 4\varepsilon$ and for $i \in
  \{0,1\}$, we have $d(h_i,\M) \leq \frac{\rho}{4}$ and
  \[
  \sup_{m,m' \in \pi_\M(\Offset \M \delta \cap B(h_1,\rho))} \angle(H_i,mm') \leq \theta.
  \]
  Applying Lemma~{\ref{lemma:relation-from-plane-zero-to-plane-one}}
  with $z=h_1$, the relation $\Rrel$ is an additive $A$-distortion and
  the relation $\Qrel$ is one-to-one.  Let us prove that $Z(\sigma_0)
  \in \Domain{\Rrel}$.  Using $\|Z(\sigma_0) - h_0 \| \leq \varepsilon
  \leq \frac{\rho}{4}$ and $\|h_0 - h_1\| \leq 4 \varepsilon \leq
  \frac{\rho}{4}$ and applying Lemma \ref{lemma:projection-surjective}
  with $(H,h)=(H_0,h_0)$, we get that
  \[
  Z(\sigma_0) \in H_0 \cap B(h_0,\frac{\rho}{4}) \subseteq \pi_{H_0}( \M \cap
  B(h_0,\frac{3\rho}{4}) ) \subseteq \pi_{H_0}(\M \cap B(h_1, \rho)) \subseteq
  \Domain{\Rrel}.
  \]
  Note that $\sigma_0$ is $\zeta$-protected with respect to
  $Q_0$. Applying
  Lemma~\ref{lemma:delaunay-stability-through-distortion}, we get that
  $R(\sigma_0) \leq \varepsilon$, $Z(\sigma_0) \in \Domain{\Rrel}$ and
  $\sigma_0 \in \Del{Q_0}$ imply $\sigma_1 \in \Del {Q_1}$ and $\sigma_1$ is protected with respect to $Q_1$. This
  proves \refitem{item:del-X-y} $\implies$ \refitem{item:del-Y-y}.


  \smallskip \noindent For proving \refitem{item:del-Y-y} $\implies$
  \refitem{item:del-Y-x}, we proceed as in the proof of
  \refitem{item:del-X-x} $\implies$ \refitem{item:del-X-y}, 
  switching the role of indices $0$ and $1$.


  \smallskip \noindent Let us prove \refitem{item:del-Y-x} $\implies$
  \refitem{item:del-X-x}.  Suppose $\sigma \subseteq P \cap
  B(h_0,\rho)$, $\sigma_1 \in \Del{\pi_{H_1}(P \cap B(h_0,\rho))}$ and
  $R(\sigma_1) \leq \varepsilon$. Consider the relations
  \begin{align*}
  \Rrel &= \{(\pi_{H_0}(a), \pi_{H_1}(a)) \mid a \in \Offset \M \delta
  \cap B(h_0,\rho) \},\\
  \Qrel &= \{(\pi_{H_0}(p), \pi_{H_1}(p)) \mid p \in P \cap B(h_0,\rho) \},\\
  \Srel &= \{(\pi_{H_0}(v), \pi_{H_1}(v)) \mid v \in \sigma \}.
  \end{align*}
  Let $Q_0 = \pi_{H_0}(P \cap B(h_0,\rho))$ and $Q_1 = \pi_{H_1}(P
  \cap B(h_0,\rho))$. By construction, $Q_0 = \Domain \Qrel$, $Q_1 =
  \Range \Qrel$, $\sigma_0 = \Domain \Srel$ and $\sigma_1 = \Range
  \Srel$.  Note that $\|h_0 - h_i\| \leq 4\varepsilon$ and for $i \in
  \{0,1\}$, we have $d(h_i,\M) \leq \frac{\rho}{4}$ and
  \[
  \sup_{m,m' \in \pi_\M(\Offset \M \delta \cap B(h_0,\rho))} \angle(H_i,mm') \leq \theta.
  \]
  Applying Lemma~{\ref{lemma:relation-from-plane-zero-to-plane-one}}
  with $z=h_0$, the relation $\Rrel$ is an additive $A$-distortion
   and the relation $\Qrel$ is one-to-one.  Let us prove that
  $Z(\sigma_1) \in \Range{\Rrel}$.  Using $\|Z(\sigma_1) - h_1 \| \leq
  \varepsilon \leq \frac{\rho}{4}$ and $\|h_0 - h_1\| \leq 4
  \varepsilon \leq \frac{\rho}{4}$ and applying Lemma
  \ref{lemma:projection-surjective} with $(H,h)=(H_1,h_1)$, we get
  that
  \[
  Z(\sigma_1) \in H_1 \cap B(h_1,\frac{\rho}{4}) \subseteq
  \pi_{H_1}( \M \cap B(h_1,\frac{3\rho}{4}) ) \subseteq
  \pi_{H_1}(\M \cap B(h_0,\rho)) \subseteq \Range{\Rrel}.
  \]
  Because $\sigma_0$ is $\zeta$-protected with respect to $Q_0$,
  we can  apply Lemma~\ref{lemma:delaunay-stability-through-distortion} and get that
  $R(\sigma_1) \leq \varepsilon$, $Z(\sigma_1) \in \Range{\Rrel}$ and $\sigma_1 \in
  \Del{Q_1}$ imply $\sigma_0 \in \Del{Q_0}$. This proves
  \refitem{item:del-Y-x} $\implies$ \refitem{item:del-X-x}.
\end{proof}


\section{Proof of the second reconstruction theorem}
\label{section:FlatDel-final}

In this section, we first show that under the assumptions of Theorem
\ref{theorem:homeomorphism-from-sampling-conditions}, $\rho$-small
$d$-simplices of $P$ are Delaunay stable for $P$ at scale $\rho$ with respect
to their standard neighborhood
(Lemma~\ref{lemma:delloc-characterization}). We then show that
whenevery the assumptions of
Theorem~\ref{theorem:homeomorphism-from-sampling-conditions} are
verified, so are the assumptions of
Theorem~\ref{theorem:homeomorphism-from-structural-conditions}
(Lemma~\ref{lemma:from-sampling-to-structural-conditions}). Finally,
we assemble the pieces and prove
Theorem~\ref{theorem:homeomorphism-from-sampling-conditions}.

Next lemma strengthens Remark
\ref{remark:projection-well-defined-for-simplex}. It says that if a
subset $\sigma \subseteq \Rspace^\Dim$ is sufficiently small and sufficiently close to a
subset $A \subseteq \Rspace^\Dim$ compare to its reach, then the convex hull of
$\sigma$ is not too far away from $A$.

\begin{lemma}
  \label{lemma:small-tubular-neigborhood}
  Let $16\delta \leq \rho \leq \frac{\Reach A}{3}$. If the subset $\sigma
  \subseteq \Offset A \delta$ is $\rho$-small, then $\Conv \sigma
  \subseteq \Offset A {\frac{\rho}{4}}$.
\end{lemma}

\begin{proof}
  Let $\reach = \Reach A$. Applying Lemma 14 in \cite{socg10-convex},
  we get that $\Conv \sigma \subseteq \Offset A r$ for $r = \reach -
  \sqrt{(\reach - \delta)^2 - \rho^2}$. Since $\delta \leq
  \frac{\rho}{16}$, we deduce that $\frac{r}{\reach} \leq 1 -
  \sqrt{\left(1 - \frac{\rho}{16\reach}\right)^2 -
    \frac{\rho}{\reach}^2}$ and since for all $0 \leq t \leq
  \frac{1}{3}$, we have $1 - \sqrt{\left(1-\frac{t}{16}\right)^2-t^2}
  \leq \frac{t}{4}$, we obtain the result.
\end{proof}

\begin{lemma}
  \label{lemma:delloc-characterization}
  Under the assumptions of Theorem
  \ref{theorem:homeomorphism-from-sampling-conditions}, every
  $\rho$-small $d$-simplex $\sigma \subseteq P$ is Delaunay stable for
  $P$ at scale $\rho$ with respect to its standard
  neighborhood. Furthermore, whenever $\sigma \in \Prestar x \rho$ for
  some $x \in \Conv \sigma$, we have that $R(\sigma) \leq \varepsilon$
  and $\pi_{\Tangent {x^\ast} \M}(\sigma)$ is protected with respect
  to $\pi_{\Tangent {x^\ast} \M}(P \cap B({x^\ast},\rho))$.
\end{lemma}


\begin{proof}
  Consider a $\rho$-small $d$-simplex $\sigma \subseteq P$. We note that
  $\sigma$ is Delaunay stable with respect to its standard neighborhood
  if, for all $x \in \Conv \sigma$, the following two propositions are
  equivalent:
  \begin{description}[itemsep=1pt,parsep=1pt,topsep=4pt]
    \item [\ita] $\sigma \subseteq P \cap B(c_\sigma,\rho)$ and
      $\sigma \in \Del{\pi_{\Aff \sigma}(P \cap B(c_\sigma,\rho))}$;
    \item [\itb{x}] $\sigma \subseteq P \cap B(x^\ast,\rho)$ and
      $\pi_{\Tangent {x^\ast} \M}(\sigma) \in \Del{\pi_{\Tangent {x^\ast}
        \M}(P \cap B({x^\ast},\rho))}$.
  \end{description}
  Pick a point $x \in \Conv\sigma$ and set $(h_0,H_0) = (c_\sigma,
  \Aff \sigma)$ and $(h_1,H_1)= (x^\ast, \Tangent {x^\ast} \M)$. We
  thus have to prove that $\sigma$ is Delaunay stable with respect to
  $\{(h_0,H_0), (h_1,H_1)\}$. We do this by applying Lemma
  \ref{lemma:stability-from-plane-zero-to-plane-one}. Let us check
  that the assumptions of Lemma
  \ref{lemma:stability-from-plane-zero-to-plane-one} are indeed
  satisfied for our choice of $h_0$, $H_0$, $h_1$, $H_1$ and with
  $\zeta = \Protection{P}{3\rho}$.

  Let $\sigma_0 = \pi_{H_0}(\sigma)$ and $\sigma_1 =
  \pi_{H_1}(\sigma)$ and note that $\sigma_0 = \sigma$ and $\sigma_1 =
  \pi_{\Tangent {x^\ast} \M}(\sigma)$.  Before we start, let us make
  two observations. Since $H_0, H_1 \in \Hcal(\sigma)$, our assumption
  that $\Theta(P,\rho) \leq \theta - \arcsin
  \frac{\rho+\delta}{\reach}$ implies that $ \angle(H_0,H_1) =
  \angle(\Aff \sigma, \Tangent {x^\ast} \M) \leq \theta $.  Second,
  \begin{equation}
    \label{eq:bounding-SEB-radius}
    r_\sigma \leq \frac{2\varepsilon}{\sqrt{3}},
    \quad\quad \text{whenever $R(\sigma_i) \leq \varepsilon$ for some $i \in \{0,1\}$}.
  \end{equation}
  Indeed, assume $R(\sigma_i) \leq \varepsilon$ for some $i \in
  \{0,1\}$. Then, applying Lemma \ref{lemma:seb-preimage}, we get that
  $r_\sigma = r_{\sigma_0} \leq \frac{r_{\sigma_i}}{\cos
    \angle(H_0,H_i)} \leq \frac{R(\sigma_i)}{\cos\theta} \leq
  \frac{\varepsilon}{\cos\frac{\pi}{6}} =
  \frac{2\varepsilon}{\sqrt{3}}$.  We are now ready to show that the
  hypotheses of Lemma
  \ref{lemma:stability-from-plane-zero-to-plane-one} are satisfied.

  \smallskip \noindent \refitem{hypo:non-degenerate}
  \styleitem{$\sigma_i$  has dimension $d$ for $i
    \in \{0,1\}$.} This is clear for $i=0$ since $\sigma_0 = \sigma$
  and we have assumed that $\height{\sigma} > 0$. For $i=1$, note that
  $\sigma_1 = \pi_{H_1}(\sigma)$ and since
  $\angle(\Aff\sigma,H_1) = \angle(H_0,H_1) < \frac{\pi}{2}$, $\sigma_1$ 
  has dimension $d$.

  \smallskip \noindent \refitem{hypo:cin-conv} \styleitem{$h_i \in \Conv(\sigma_i)$ for $i \in \{0,1\}$.}
    Note that $h_0 \in \Conv \sigma_0$ is equivalent to $c_\sigma \in
    \Conv \sigma$ which is clearly true and $h_1 \in \Conv \sigma_1$
    is equivalent to $x^\ast \in \Conv( \pi_{\Tangent {x^\ast} \M}(\sigma))$
    which is also true because $x^\ast = \pi_{\Tangent {x^\ast} \M}(x) \in
    \pi_{\Tangent {x^\ast} \M}(\Conv \sigma) = \Conv( \pi_{\Tangent {x^\ast}
      \M}(\sigma))$. 

  \smallskip \noindent \refitem{hypo:close-manifold} \styleitem{$d(h_i,\M) \leq
    \frac{\rho}{4}$ for $i \in \{0,1\}$.} This is clearly true for
  $i=1$ since $d(x^\ast,\M)=0$. For $i=0$, we have to
  show that $d(c_\sigma,\M) \leq \frac{\rho}{4}$ which is also true by Lemma \ref{lemma:small-tubular-neigborhood}.

  \smallskip \noindent \refitem{hypo:angle} \styleitem{ For $0 \leq i,j
    \leq 1$, $ \sup_{m,m' \in \pi_\M(\Offset \M \delta \cap
      B(h_j,\rho))} \angle(H_i,mm') \leq \theta$.} Consider $m,m' \in
  \pi_\M(\Offset \M \delta \cap B(h_j,\rho))$. Applying Lemma
  \ref{lemma:general-angle-bound} with $\tau = \{m,m'\}$ and $z =
  h_j$, we obtain
  \begin{align*}
    \angle(H_i, mm') \leq \angle(H_i,H_j) + \angle(H_j,mm') 
    \leq \Theta(\sigma) + \arcsin \left(\frac{\rho+\delta}{\reach}\right) \leq \theta.
  \end{align*}

  \smallskip \noindent  \refitem{hypo:small-circumradius} \styleitem{$\|h_0 - h_1\| \leq 4
    \varepsilon$ whenever $R(\sigma_0) \leq \varepsilon$ or $R(\sigma_1)
    \leq \varepsilon$.} This boils down to showing that $\|c_\sigma - {x^\ast}\| \leq
  4\varepsilon$ whenever there exists a space $H \in \{\Aff \sigma,
  \Tangent {x^\ast} \M\}$ such that $R(\pi_H(\sigma)) \leq \varepsilon$. Since $\|x-{x^\ast}\| =
  d(x,\M) \leq d(x,\pi_{\M}(\sigma)) \leq d(x,\sigma) + \delta \leq
  r_\sigma + \varepsilon$ and $\|c_\sigma - x \| \leq r_\sigma$, it
  follows from (\ref{eq:bounding-SEB-radius}) that $\|c_\sigma -
  {x^\ast}\| \leq \|c_\sigma - x\| + \|x - {x^\ast}\| \leq 2 r_\sigma
  + \varepsilon \leq 4 \varepsilon$.

  \smallskip \noindent \refitem{hypo:simplex-enclosed-in-ball} \styleitem{For $0 \leq i,j \leq 1$
    with $i \neq j$, $R(\sigma_i) \leq \varepsilon \implies \sigma
    \subseteq B(h_{j},\rho)$.} Let us prove it for
  $(i,j)=(0,1)$. Assume $R(\sigma) \leq \varepsilon$. Using $\|c_\sigma -
  x^\ast\| = \|h_0 - h_1\| \leq 4\varepsilon$, we obtain that $R(\sigma) \subseteq
  B(Z(\sigma),R(\sigma)) \subseteq B(c_\sigma,2\varepsilon) \subseteq
  B(x^\ast,6\varepsilon) \subseteq B(h_1,\rho)$.  Let us prove it for
  $(i,j)=(1,0)$. Assume $R(\pi_{\Tangent {x^\ast} \M}(\sigma)) \leq
  \varepsilon$. Then, using (\ref{eq:bounding-SEB-radius}), $r_\sigma
  \leq \frac{2\varepsilon}{\sqrt{3}} < \rho$ and $\sigma \subseteq
  B(c_\sigma,r_\sigma) \subseteq B(c_\sigma,\rho) =
  B(h_0,\rho)$.

  \smallskip \noindent \refitem{hypo:protection} \styleitem{$2 A \left(
    1 + \frac{2d\varepsilon}{\height{\sigma_i}}
    \right) < \zeta$ for $i \in \{0,1\}$.} The inequality is clearly true for
  $i=0$ since $\sigma_0 = \sigma$ and $\zeta = \Protection{P}{3\rho}$. Let us prove it for $i=1$. Recalling that
  $\angle(H_0,H_1) \leq \frac{\pi}{6}$
  and applying Lemma \ref{lemma:delta-thickness-under-distortion}, we obtain
  $
  \height{\sigma_1} \geq \cos \angle(H_0,H_1) \height{\sigma_0}
  \geq \frac{\sqrt{3}}{2} \height{\sigma_0}
  $. Hence,
  \[
  2 A \left(1 + \frac{2d\varepsilon}{\height{\sigma_1}} \right)
  \leq 2 A \left(1 + \frac{4d\varepsilon}{\height{\sigma_0}} \right)
  < \zeta = \Protection{P}{3\rho},
  \]
  showing the inequality for $i=1$.

  \smallskip Applying Lemma
  \ref{lemma:stability-from-plane-zero-to-plane-one}, we get that \ita
  $\iff$ \itb{x} and furthermore, whenever \ita or \itb{x} holds, then
  $\pi_{\Tangent {x^\ast} \M}(\sigma)$ is protected with respect to $\pi_{\Tangent {x^\ast}
        \M}(P \cap B({x^\ast},\rho))$ and $R(\sigma)
  \leq \varepsilon$. This concludes the proof.
\end{proof}


\begin{lemma}
  \label{lemma:from-sampling-to-structural-conditions}
  Whenever the assumptions of Theorem
  \ref{theorem:homeomorphism-from-sampling-conditions} are verified,
  so are the assumptions of Theorem \ref{theorem:homeomorphism-from-structural-conditions}.
\end{lemma}


\begin{proof}
  Assume that the assumptions of Theorem
  \ref{theorem:homeomorphism-from-sampling-conditions} are satisfied
  and let us verify that the five structural conditions of Theorem
  \ref{theorem:homeomorphism-from-structural-conditions} are met.

  \smallskip \noindent \refhypo{hypo:structural-proj-M} By Lemma
  \ref{lemma:injectivity-projection-manifold}, for every $\rho$-small
  $d$-simplex $\sigma \subseteq P$, the map $\restr{\pi_\M}{\Conv\sigma}$ is injective.

  \smallskip \noindent \refhypo{hypo:structural-proj-onto-tangent-plane} By Lemma
  \ref{lemma:injectivity-projection-tangent-space}, for all $m \in
  \M$, the map $\restr{\pi_{\Tangent m \M}}{P \cap B(m,\rho)}$ is
  injective.

  \smallskip \noindent \refhypo{hypo:structural-locally-manifold} By Lemma \ref{lemma:surrounded}, for all $m \in \M$,
  the domain $\US{\Delstar m \rho}$ is homeomorphic to $\Rspace^d$.

  \smallskip \noindent
  \refhypo{hypo:structural-geometrically-realized} Let us show that
  for all $m \in \M$, $\Delstar m \rho$ is geometrically
  realized. Since the domain $\US{\Delstar m \rho}$ is homeomorphic to
  $\Rspace^d$, $\Delstar {m} \rho$ contains at least a $d$-simplex and
  it suffices to show that all $d$-simplices in $\Delstar {m} \rho$
  are protected with respect to $\pi_{\Tangent m \M}(P \cap B(m,\rho))$ to deduce that $\Delstar m \rho$ is geometrically
  realized. Consider a $d$-simplex $\sigma' \in \Delstar {m} \rho$.
  By definition of the star, there exists a $d$-simplex $\sigma \in P
  \cap B(m,\rho)$ such that $\sigma' = \pi_{\Tangent m
    \M}(\sigma)$. In other words, $\sigma \in \Prestar m \rho$. Note
  that we can find $x \in \Conv \sigma$ such that $m = \pi_{\Tangent m
    \M}(x)$. By Remark \ref{remark:same-projections}, $m = \pi_\M(x)$ and by Remark
  \ref{remark:same-projections-same-prestar}, $\Prestar m \rho = \Prestar x \rho$. Thus, $\sigma \in
  \Prestar x \rho$ for some $x \in \Conv \sigma$ and applying Lemma
  \ref{lemma:delloc-characterization}, we get that $\sigma'$ is
  protected with respect to $\pi_{\Tangent m \M}(P \cap B(m,\rho))$.

  \smallskip \noindent \refhypo{hypo:structural-agreement} By Lemma
  \ref{lemma:delloc-characterization}, every $\rho$-small $d$-simplex
  $\sigma$ is Delaunay stable for $P$ at scale $\rho$ with respect to
  its standard neighborhood. Applying Lemma
  \ref{lemma:agreement-from-stability}, we deduce that $\sigma$ has its
  prestars in agreement at scale $\rho$.
\end{proof}



\begin{proof}[Proof of Theorem \ref{theorem:homeomorphism-from-sampling-conditions}]
  By Lemma \ref{lemma:from-sampling-to-structural-conditions}, the
  assumptions of Theorem
  \ref{theorem:homeomorphism-from-structural-conditions} are
  satisfied. We thus deduce that (1) $\FlatDel P \rho$ is a faithful
  reconstruction and (2) the prestar formula holds.  Applying Lemma
  \ref{lemma:delloc-characterization}, it is not difficult to see that
  (3) $R(\sigma) \leq \varepsilon$ for all $d$-simplices $\sigma \in
  \FlatDel P \rho$. Applying Lemma \ref{lemma:delloc-characterization}
  again, we deduce that every $\rho$-small $d$-simplex $\sigma$ is
  Delaunay stable for $P$ at scale $\rho$ with respect to its standard
  neighborhood and applying Lemma
  \ref{lemma:agreement-from-stability}, we get that (4) a $d$-simplex
  $\sigma$ belongs $\FlatDel P \rho$ if and only if $\sigma$ delloc in $P$ at
  scale $\rho$.
\end{proof}


\section{Perturbation procedure for ensuring safety conditions}
\label{section:PerturbationProcedure}

While assuming the sample to be $\varepsilon$-dense and
$\delta$-accurate, seems realistic enough (perhaps after filtering
outliers), conditions \refhypo{hypo:safety-angle},
\refhypo{hypo:safety-separation} and \refhypo{hypo:safety-protection}
in Theorem \ref{theorem:homeomorphism-from-sampling-conditions} seem
less likely to be satisfied by natural data. In fact, it is not even
obvious that there exists a point set $P$ satisfying the conditions of
Theorem \ref{theorem:homeomorphism-from-sampling-conditions}.  Note
that condition \refhypo{hypo:safety-separation} that imposes a lower
bound on the separation of the data points can easily be satisfied, at
the price of doubling the density parameter $\varepsilon$; see
\cite[Section 5.1]{boissonnat2018geometric} for a standard procedure
that extracts an $\varepsilon$-net. In this section, we assume that
$P$ is a $\delta$-accurate $\varepsilon$-dense sample of \M and
perturbe it to obtain a point set $P'$ that satisfies the assumptions
of our main theorem. For this, we use the Moser Tardos Algorithm
\cite{moser2010constructive} as a perturbation scheme in the spirit of
what is done in \cite[Section 5.3.4]{boissonnat2018geometric}.

The perturbation scheme is parametrized with real numbers $\rho \geq
0$, $r_{\operatorname{pert.}} \geq 0$, $\mathrm{Heigh}_{\min}> 0$, and
$\mathrm{Prot}_{\min}> 0$. To describe it, we need some notations and
terminology. Let $\tilde{T}_{p} = T_{p}(P,3\rho)$ be the $d$-dimensional
affine space passing through $p$ and parallel to the $d$-dimensional
vector space $V_{p}(P,3\rho)$ defined as follows: $V_{p}(P, 3\rho)$
is spanned by the eigenvectors associated to
the $d$ largest eigenvalues of the inertia tensor of $(P \cap B(p,
3\rho)) - c$, where $c$ is the center of mass of $P \cap B(p,3\rho)$. 
To each point $p \in P$, we associate a
perturbed point $p' \in P'$, computed by applying a sequence of
elementary operations called reset. Precisely, given a point $p' \in P'$
associated to the point $p \in P$, the {\em reset} of $p'$ is the
operation that consists in drawing a point $q$ uniformely at random in
$V_{p} \cap B(p,r_{\operatorname{pert.}})$ and assigning $q$ to
$p'$. Finally, we call any of the two situations below a {\em bad
  event}:
\begin{description}
\item[\styleitem{Violation of the height condition by $\sigma'$:}] A $\rho$-small $d$-simplex
$\sigma' \subseteq P'$ such that $\height{\sigma'}<\mathrm{Heigh}_{\min}$;
\item[\styleitem{Violation of the protection condition by $(p',\sigma')$:}] A pair $(p',\sigma')$ made of a point $p' \in P'$
  and a $d$-simplex $\sigma' \subseteq P' \setminus \{p'\}$ such that $p' \in
  B(c_{\sigma'},3\rho)$ and $\sigma'$ is not
  $\mathrm{Prot}_{\min}$-protected with respect to $\{\pi_{\Aff
    {\sigma'}}(p')\}$.
\end{description}
In both situations, we associate to the bad event $E$ a set of points
called the points {\em correlated} to $E$. In the first
situation, the points correlated to $E$ are the $d+1$ vertices of
$\sigma'$ and in the second situation, they are the $d+2$ points of
$\{p'\} \cup \sigma'$.

\bigskip

\noindent \fbox{
    \begin{minipage}{0.97\textwidth}
      \noindent {\bf Moser-Tardos Algorithm: }\\
      \texttt{
      1. For each $p \in P$, compute the $d$-dimensional affine space $\tilde{T}_{p}$ \\
      2. For each point $p' \in P'$, reset $p'$ \\
      3. WHILE (some bad event $E$ occurs):\\
      --------------- For each point $p'$ correlated to $E$, reset $p'$ \\
      ----- END WHILE\\
      4. Return $P'$
      }
    \end{minipage}
}

\bigskip

Roughly speaking, in our context, the Moser Tardos Algorithm reassigns
new coordinates to any point $p \in P$ that is correlated to a bad
event as long as a bad event occurs. A beautiful result from
\cite{moser2010constructive} tells us that the Moser-Tardos Algorithm
terminates in a number of steps that is expected to be linear in the size of $P$. We thus
have:

\begin{lemma}\label{lemma:PerturbationTunedForTheorem}
Let $\varepsilon \geq 0$, $\eta >0$, and $\rho =
C_{\operatorname{ste}} \varepsilon$, where $C_{\operatorname{ste}}
  \geq 32$.  Let $\delta = \frac{\rho^2}{\reach}$,
$r_{\operatorname{pert.}} = \frac{\eta \varepsilon}{20}$,
$\varepsilon' = \frac{21}{20}\varepsilon$, and $\delta' = 2\delta$.
 There are
positive constants $c_1$, $c_2$, $c_3$, and $c_4$ that depend only upon
$\eta$, $C_{\operatorname{ste}}$, and $d$ such that if
$\frac{\varepsilon}{\reach} < c_1$ then, given a point set $P$
such that $\M \subseteq \Offset {P} {\varepsilon}$, $P
\subseteq \Offset \M {\delta}$, and $\Sep{P} > \eta
\varepsilon$, the point set $P'$ obtained after resetting each of its points
satisfies $\M \subseteq \Offset {(P')} {\varepsilon'}$, $P'
\subseteq \Offset \M {\delta'}$, and $\Sep{P'} > \frac{9}{10} \eta
\varepsilon$.  Moreover, whenever we apply the Moser-Tardos Algorithm with
$\mathrm{Heigh}_{\min} = c_2 \left( \frac{\rho}{\reach}
\right)^{\frac{1}{2}} \rho $ and $\mathrm{Prot}_{\min} =c_3 \left(
\frac{\rho}{\reach} \right)^{\frac{1}{2}} \rho$, the algorithm terminates with
expected time $O( \sharp P)$ and returns a point set $P'$ that
satisfies:
\begin{align*}
\Height{P'}{\rho} &\geq c_2 \left( \frac{\rho}{\reach} \right)^{\frac{1}{2}} \rho \\
\Protection{P'}{\rho} &\geq  c_3 \left( \frac{\rho}{\reach} \right)^{\frac{1}{2}}  \rho 
\end{align*}
As a consequence of the above lower bound on
$\Height{P'}{\rho}$, we have:
 \[
 \Theta(P',\rho)  \leq c_4 \left( \frac{\rho}{\reach} \right)^{\frac{1}{2}}.
 \]
The point set $P'$ returned by the Moser-Tardos Algorithm is a $\delta'$-accurate $\varepsilon'$-dense sample of \M that
satisfies the assumptions of Theorem \ref{theorem:homeomorphism-from-sampling-conditions} with
parameters $\varepsilon'$, $\delta'$, $\rho$, and some $\theta \geq 0$.
\end{lemma}

The proof is given in Appendix \ref{section:ProofOfLemmaPerturbationTunedForTheorem}.

\bibliography{../geofi}


\clearpage
\appendix


\section{Angle between affine spaces}

In this appendix, we present basic upper bounds on the angle between
affine spaces spanned by simplices close to a manifold and nearby
tangent spaces to that manifold. We start by recalling how the angle
between two affine spaces is defined \cite{BSMF_1875__3__103_2}:

\begin{definition}[Angle between affine spaces]
  Consider two affine spaces $H_0, H_1 \subseteq \Rspace^\Dim$. Let
  $V_0$ and $V_1$ be the vector spaces associated respectively to
  $H_0$ and $H_1$.  The {\em angle between $V_0$ and $V_1$} is defined
  as
  \begin{equation*}
    \angle (V_0, V_1) = \sup_{\Above{v_0\in V_0}{\|v_0\|=1}} \inf_{\Above{v_1\in V_1}{\|v_1\|=1}}  \angle v_0, v_1
    = \max_{\Above{v_0\in V_0}{\|v_0\|=1}} \min_{\Above{v_1\in V_1}{\|v_1\|=1}}  \angle v_0, v_1
  \end{equation*}
  and the {\em angle between $H_0$ and $H_1$} is defined as $\angle (H_0,
  H_1) = \angle (V_0, V_1)$.
\end{definition}
Note that by definition, $\angle (H_0, H_1) \in [0, \frac{\pi}{2}]$. We recall a classical result:
\[
\dim H_0 = \dim H_1 \quad \implies \quad \angle (H_0, H_1) = \angle (H_1,
H_0).
\]
In other words, whenever $H_0$ and $H_1$ share the same dimension, the
angle definition is symmetric in $H_0$ and $H_1$. Skipping details,
this is because, in that case, there exists an isometry that swaps the
associated vector spaces $V_0$ and $V_1$ while preserving the angles.

We are now ready to state a few lemmas. As usual, we assume the reach
of \M to be positive and let $\reach$ be a fixed finite constant such
that $0 < \reach \leq \Reach \M$ so as to handle the case where \M has
an infinite reach. We start by enunciating a lemma due to Federer
\cite{federer-59} which bounds the distance of a point $q\in \M$ to
the tangent space at a point $p \in \M$.  It holds for any set with a
positive reach:

\begin{lemma}[{\cite[Theorem 4.8]{federer-59}}]
  \label{lemma:tangdel:distanceToTangent}
For any $p,q \in \M$ such that $\|p-q\| < \reach$, we
have
\begin{align*} 
  \sin \angle (pq, \Tangent p \M ) \leq \frac{\|p-q\|}{2 \reach}
  \qquad \text{ and } \qquad d(q , \Tangent p \M )  \leq \frac{\|p-q\|^2}{2 \reach }.
\end{align*}
\end{lemma}
Next lemma bounds the angle variation between two tangent spaces for $C^2$-manifolds and can be
found for instance in \cite{boissonnat2019reach}\footnote{A slightly weaker condition is given for
$C^{1,1}$-manifolds in that paper.}:
\begin{lemma}[{\cite[Corollary~3]{boissonnat2019reach}}]
  \label{lemma:tangdel:TangentSpaceVariation}
  For any $p,q \in \M$, we have
  \begin{align*}
    \sin \left(\frac{\angle (\Tangent p \M, \Tangent q \M )}{2} \right)
    \leq \frac{\|p-q\|}{2\reach}.
  \end{align*}
\end{lemma}
We shall also need the Whitney angle bound established in \cite{BOISSONNAT_2013}.

\begin{lemma}[Whitney angle bound {\cite[Lemma 2.1]{BOISSONNAT_2013}}]
  \label{lemma:whitney-angle-bound}
  Consider a $d$-dimensional affine space $H$ and a simplex $\sigma$
  such that $\dim\sigma \leq d$ and $\sigma \subseteq \Offset H t$ for
  some $t \geq 0$. Then
  \[
  \sin \angle(\Aff\sigma,H) \leq \frac{2 t \dim(\sigma)}{\height{\sigma}}
  \]
\end{lemma}

Building on these results, we derive various bounds between the affine
space spanned by a simplex and a nearby tangent space.

\begin{lemma}
  \label{lemma:general-angle-bound}
  Consider a non-degenerate $\rho$-small simplex $\tau \subseteq
  \Offset \M \delta$ with $16\delta \leq \rho \leq
  \frac{\reach}{3}$. Let $z$ be a point such that $\tau \subseteq
  B(z,\rho)$ and $d(z,\M) \leq \frac{\rho}{4}$. Then,
  \[
  \angle(\Aff \tau, \Tangent {\pi_\M(z)} \M) \leq
  \arcsin \left(
  \frac{2\dim(\tau)}{\height{\tau}}
  \left(\frac{\rho^2}{\reach} + \delta \right) \right).
  \]
\end{lemma}


\begin{proof}
  Let $v \in \tau$. Write $v^\ast = \pi_\M(v)$ and $z^\ast =
  \pi_\M(z)$.  We know from \cite[page 435]{federer-59} that for $0
  \leq h < \Reach{\M}$, the projection map $\pi_\M$ onto $\M$ is
  $\left( \frac{\reach}{\reach - h} \right)$-Lipschitz for points at
  distance less than $h$ from $\M$. Since both $z$ and $v$ belong to
  $\Offset \M h$ for $h = \frac{\rho}{4}$, we thus have
  \begin{equation*}
    \| v^* - z^* \| \leq \frac{\reach}{\reach - \frac{\rho}{4}} \times \| v - z\| 
    \leq  \frac{\reach}{\reach - \frac{\reach}{3\times 4}} \times \|v - z\| \leq  \sqrt{2}\rho.
  \end{equation*}
  Applying Lemma \ref{lemma:tangdel:distanceToTangent}, we get that
  \[
    d(v,\Tangent {z^\ast} \M) \leq d(v^\ast,\Tangent {z^\ast} \M) + \|v - v^\ast\| 
    \leq \frac{\|v^\ast - z^\ast\|^2}{2\reach} + \delta 
    \leq \frac{\rho^2}{\reach} + \delta.
  \]
  Hence, $\tau \subseteq \Offset {(\Tangent {z^\ast} \M)} t$ for $t =
  \frac{\rho^2}{\reach} + \delta$ and applying Whitney angle bound
  (Lemma \ref{lemma:whitney-angle-bound}), we conclude that $\sin
  \angle(\Aff\tau,\Tangent {z^\ast} \M) \leq
  \frac{2\dim(\tau)}{\height{\tau}} \left(\frac{\rho^2}{\reach} +
  \delta \right) $.
\end{proof}


\begin{corollary}
  \label{lemma:angle-span-simplex-tangent-plane}
  For any non-degenerate $\rho$-small simplex $\sigma \subseteq \Offset \M
  \delta$ with $16\delta \leq \rho \leq \frac{\reach}{3}$:
  \[
  \Theta(\sigma) \leq
  \arcsin \left(
  \frac{2\dim(\sigma)}{\height{\sigma}}
  \left(\frac{4 \rho^2}{\reach} + \delta \right) \right) +
  \arcsin \left( \frac{\rho+\delta}{\reach}\right).
  \]
\end{corollary}


\begin{proof}
  Letting $x^\ast = \pi_\M(x)$ and $y^\ast = \pi_\M(y)$, the angle $\Theta(\sigma)$ can be expressed as follows:
  \[
  \Theta(\sigma) = \max \left\{ \max_{x \in \Conv \sigma} \angle(\Aff \sigma, \Tangent {x^\ast} \M),
  \max_{x,y \in \Conv \sigma} \angle(\Tangent {x^\ast} \M, \Tangent {y^\ast} \M) \right\}.
  \]
  By Lemma \ref{lemma:small-tubular-neigborhood}, for any $x \in \Conv
  \sigma$, $d(x,\M) \leq \frac{\rho}{4}$ and $\sigma \subseteq
  B(x,2\rho)$. Applying Lemma~\ref{lemma:general-angle-bound} with
  $\tau = \sigma$ and $z=x$, we get that $\max_{x \in \Conv \sigma}
  \angle(\Aff \sigma, \Tangent {x^\ast} \M)$ is upper bounded by the
  first term in the above sum. Applying
  Lemma~\ref{lemma:tangdel:TangentSpaceVariation}, we get that
  $\max_{x,y \in \Conv \sigma} \angle(\Tangent {x^\ast} \M, \Tangent
  {y^\ast} \M)$ is upper bounded by the second term in the above
  sum.
\end{proof}


\begin{lemma}
  \label{lemma:seb-preimage}
  Let $\sigma \subseteq \Rspace^\Dim$ be a $d$-simplex. Let $H$ be a
  $d$-dimensional affine space and suppose $\angle(\Aff\sigma,H) <
  \frac{\pi}{2}$. Then, $r_\sigma \leq \frac{1}{\cos
    \angle(\Aff\sigma,H)} \times r_{\pi_H(\sigma)}$.
\end{lemma}


\begin{proof}
  For short, write $\theta = \angle(\Aff\sigma,H)$ and $\sigma' =
  \pi_H(\sigma)$. Because $\theta < \frac{\pi}{2}$, the restriction of
  $\pi_H$ to $\Aff \sigma$ is a homeomorphism. Consider the point $z
  \in \Aff\sigma$ such that $\pi_H(z) = c_{\sigma'}$. For all $v \in
  \sigma$, we have
  \[
  \|z-v\| \leq \frac{1}{\cos \theta} \times \|c_{\sigma'} - \pi_H(v) \| = \frac{1}{\cos \theta} \times r_{\sigma'}.
  \]
  The result follows.
\end{proof}


\begin{lemma}
  \label{lemma:delta-thickness-under-distortion}
    Let $\sigma \subseteq \Rspace^\Dim$ be a non-degenerate
    $d$-simplex. Let $H$ be a $d$-dimensional affine space and suppose
    $\angle(\Aff\sigma,H) \leq \theta$ for some $\theta <
    \frac{\pi}{2}$. Then, $\pi_H(\sigma)$ is a non-degenerate
    $d$-simplex whose height is lower bounded as follows:
    \begin{align*}
    \cos \theta \times  \height{\sigma} &\leq \height{\pi_H(\sigma)}.
    \end{align*}
\end{lemma}


\begin{proof}
  For any two points $x, y \in \Aff \sigma$, we have
  \[
  \cos\theta \times \|x - y\| \leq \| \pi_{H}(x) - \pi_{H}(y) \|.
  \]
  Because $\cos\theta \neq 0$, the restriction of $\pi_H$ to $\Aff
  \sigma$ is injective and therefore a homeomorphism. Hence,
  $\pi_H(\sigma)$ is a non-degenerate $d$-simplex.  Consider a vertex
  $v \in \sigma$ and a point $x \in \Aff \sigma \setminus \{v\}$ such
  that $\pi_H(x)$ is the point of $\pi_H(\Aff \sigma \setminus \{v\})$
  closest to $\pi_H(v)$. If follows from $ \cos \theta \times \|v - x
  \| \leq \|\pi_H(v) - \pi_H(x)\| $ that $ \cos \theta \times
  \height{\sigma} \leq \height{\pi_H(\sigma)}$.
\end{proof}



\clearpage

\section{Perturbation ensuring fatness and protection}\label{section:PerturbationAndLLL}
In this section, following \cite[Section 5.3.4]{boissonnat2018geometric}, we make use of the Lov{\'a}sz local lemma 
\cite{alon2016probabilistic} and its algorithmic avatar \cite{moser2010constructive},
to show how to effectively perturb the point sample in order to ensure the required protection and fatness conditions.

However, we have specific constraints here that does not occur  in 
the context of  \cite[Section 5.3.4]{boissonnat2018geometric}.
Indeed, with respect to the projection of neihboring points on the affine hull of the $d$-simplex itself, the required protection depends on the
angle variability of these simplices affine hulls, which itself depends on the simplices minimal heights. It follows that the required minimal protection and heigh cannot be 
defined independently  in what follows, which constraints the choice of events upon which  Lov{\'a}sz local lemma application relies.

\subsection{Approximate tangent space computed by PCA}

\begin{lemma}\label{lemma:PerturbationAndLLL_pseudoTangent}
Let $\delta \geq 0$, $0 < \varepsilon \leq \frac{\rho}{16}$, $10 \rho < R \leq  \Reach{\M}$ and suppose that $P \subseteq \Offset \M \delta$ 
for $\delta < \frac{\rho^2}{4 R}$, and 
$\M \subseteq \Offset P \varepsilon$ and $\Sep{P} > \eta \varepsilon$ for $\eta >0$.

For any  $p\in P$, if $c_p$ is the center of mass of $P \cap B(p, \rho)$ and $V_p$ the linear space spanned by the $n$ eigen vectors corresponding 
to the $n$ largest eigenvalues of the inertia tensor of $\big( P \cap B(p, \rho) \big) - c_p$, 
then one has:
\[
\angle V_p , \Tangent {\pi_\M(p)} \M \,< \,\Xi _0( \eta, d) \frac{\rho}{R}
\]
where the function $\Xi_0$  is a polynomial in $\eta$ and exponential in $d$.
\end{lemma}


\begin{proof} 
Consider a frame centered at $c_p$ with an orthonormal basis of $\Rspace^N$
whose $d$ first vectors  $e_1,\ldots,e_d$ belong to $\Tangent {\pi_{\M}(p)} \M$ and the $N-d$ last vectors  $e_{d+1},\ldots,e_N$ to the normal fiber $\Normal {\pi_{\M}(p)} \M$.
Consider the symmetric $N\times N$ normalized inertia tensor $\mathbf{A}$ of $ P \cap B(p, \rho) $ in this frame:
\[
\mathbf{A}_{ij} \defunder{=} \frac{1}{\sharp P \cap B(p, \rho)} \sum_{p \in P \cap B(p, \rho)} \langle v_i, p - c_p\rangle  \langle v_j, p - c_p\rangle 
\]
The symmetric matrix $A$ decomposes into $4$ blocs:
\begin{equation}
\mathbf{A}=
\begin{pmatrix}
\mathbf{A}_{TT} & \mathbf{A}_{TN} \\
\mathbf{A}_{TN}^t & \mathbf{A}_{NN}
\end{pmatrix}
\end{equation}
where $\mathbf{A}_{TT}$, the tangental inertia  is $d\times d$ symmetric define positive. Because of the sampling conditions, 
we claim\footnote{But this claim has to be  detailed if one want to give an explicit expression of the quantity $\Xi _0( \eta, d)$ of the lemma.} 
there is a constant $C_{TT}>0$ depending only on $\eta$ and $d$ such that the smallest eigenvalue of 
$\mathbf{A}_{TT}$ is at least $C_{TT} \rho^2$:
\begin{equation}\label{equation:LowerBoundOnATT}
\forall u \in \Rspace^d, \|u\| = 1 \Rightarrow \: u^t \, \mathbf{A}_{TT} \, u \geq C_{TT} \rho^2 
\end{equation}

 Observe that, by Lemma \ref{lemma:tangdel:distanceToTangent}, 
 the points in $P \cap B(p, \rho)$ are at distance less than $\frac{(\rho+\delta)^2}{2R}$
from space $\pi_{\M}(p)+ \Tangent {\pi_{\M}(p)} \M$, and therefore so is $c_p$. It follows that the 
points in $P \cap B(p, \rho)$ are at distance less than $2 \frac{(\rho+\delta)^2}{2R} \leq 2 \frac{\rho^2}{R} $
 (assuming $(\rho+ \delta)^2 \leq 2 \rho^2$) 
from space $c_p+ \Tangent {\pi_{\M}(p)} \M$.
Then, 
 there are constants $C_{TN}$ and $C_{NN}$ such that the operator norm induced by euclidean  vector norm, of 
 $\mathbf{A}_{TN}$ and $\mathbf{A}_{NN}$ are upper bounded as :
\begin{equation}\label{equation:LowerBoundOnATN}
\forall u,v \in \Rspace^d,  \|u\| =  \|v\| = 1  \Rightarrow  \: v^t \, \mathbf{A}_{TN} \, u \leq C_{TN}  \frac{ \rho^2}{R} \rho 
\end{equation}
and:
\begin{equation}\label{equation:LowerBoundOnANN}
\forall u \in \Rspace^d,  \|u\| = 1 \Rightarrow  \:  u^t \,\mathbf{A}_{NN} \, u  \leq C_{NN}   \frac{ \rho^2}{R}  \frac{ \rho^2}{R} 
\end{equation}
Let $v\in \Rspace^N$ be a unit eigenvector of $\mathbf{A}$ with eigenvalue $\lambda$:
\begin{equation}\label{equation:EigneVectorOfA}
\mathbf{A}\, v = \lambda \, v
\end{equation}
Define $T \defunder{=} \Rspace^d \times \{0\}^{N-d} \subset \Rspace^N$ and $N \defunder{=}  \{0\}^d \times  \Rspace^{N-d} \subset \Rspace^N$,
corresponding, in the space of coordinates, respectively to $\Tangent {\pi_{\M}(p)} \M$ and $\Normal {\pi_{\M}(p)} \M$.

Let $\theta$ be the angle between $v$ and $T$. There are unit vectors $v_T \in \Rspace^d$ and $v_N \in  \Rspace^{N-d}$
such that:
\[
v = ((\cos \theta) v_T, (\sin \theta) v_N)^t .
\] 
where for a matrix $u$, $u^t$ denotes the transpose of $u$,
 and \eqref{equation:EigneVectorOfA} can be rewritten as:
 \begin{equation}\label{equation:EigneVectorOfA_2X2}
\begin{pmatrix}
\mathbf{A}_{TT} & \mathbf{A}_{TN} \\
\mathbf{A}_{TN}^t & \mathbf{A}_{NN}
\end{pmatrix}\,  
\begin{pmatrix}
(\cos \theta) v_T \\
 (\sin \theta) v_N
\end{pmatrix}  
= \lambda \,
\begin{pmatrix}
(\cos \theta) v_T \\
 (\sin \theta) v_N
\end{pmatrix}\,  
\end{equation}
equivalently:
\begin{align}
 (\cos \theta) \mathbf{A}_{TT}  v_T &+  (\sin \theta) \mathbf{A}_{TN}   v_N = \lambda (\cos \theta) v_T \label{equation:LambdaCosThetaV_T} \\
 (\cos \theta)\mathbf{A}_{TN}^t  v_T &+  (\sin \theta)  \mathbf{A}_{NN}  v_N = \lambda (\sin \theta) v_N  \label{equation:LambdaSinThetaV_N} 
\end{align}
 multiplying, on the left, the two equations respectively by $(\sin \theta)v_T^t$ and $ (\cos \theta)  v_N^t$ we get:
 \begin{align*}
(\sin \theta) (\cos \theta)  v_T^t \mathbf{A}_{TT}  v_T +  (\sin \theta)^2  v_T^t  \mathbf{A}_{TN}   v_N &= \lambda  (\sin \theta) (\cos \theta)   \\
 (\cos \theta)^2 v_N^t \mathbf{A}_{TN}^t  v_T +   (\cos \theta) (\sin \theta) v_N^t \mathbf{A}_{NN}  v_N &= \lambda (\cos \theta)  (\sin \theta) 
\end{align*}
So that:
\[
(\sin \theta) (\cos \theta)  v_T^t \mathbf{A}_{TT}  v_T +  (\sin \theta)^2  v_T^t  \mathbf{A}_{TN}   v_N  
=  (\cos \theta)^2 v_N^t \mathbf{A}_{TN}^t  v_T +   (\cos \theta) (\sin \theta) v_N^t \mathbf{A}_{NN}  v_N
\]
Using  \eqref{equation:LowerBoundOnATN} and \eqref{equation:LowerBoundOnANN}, 
we get:
\[
(\sin \theta) (\cos \theta)  v_T^t \mathbf{A}_{TT}  v_T  \leq 2  C_{TN}  \frac{ \rho^2}{R} \rho  +  C_{NN}   \frac{ \rho^2}{R}  \frac{ \rho^2}{R} 
\]
 Using \eqref{equation:LowerBoundOnATT}, we get:
 \[
(\sin \theta) (\cos \theta)   \leq 2    \frac{ C_{TN}}{C_{TT}}  \frac{ \rho}{R}   +  C_{NN}   \frac{ C_{NN}}{C_{TT}}   \frac{ \rho^2}{R^2} 
\]
Using $\sin 2\theta = 2 \sin \theta \cos \theta$, we get:
\begin{equation}\label{equation:UpperBoundOnSinTwoTheta}
\frac{1}{2} \sin 2 \theta \leq 2    \frac{ C_{TN}}{C_{TT}}  \frac{ \rho}{R}   +   C_{NN}   \frac{ C_{NN}}{C_{TT}}   \frac{ \rho^2}{R^2} = \mathcal{O}\left(  \frac{ \rho}{R} \right)
\end{equation}
We get that:
\[
\theta \in [0, t] \cup [\pi/2 - t, \pi/2]
\]
with:
\[
t = \frac{1}{2}  \arcsin 2 \left(2    \frac{ C_{TN}}{C_{TT}}  \frac{ \rho}{R}   
+   C_{NN}   \frac{ C_{NN}}{C_{TT}}   \frac{ \rho^2}{R^2} \right) = \mathcal{O}\left(  \frac{ \rho}{R} \right)
\]
This means that eigenvectors of $\mathbf{A}$ (for the non generic situation of a multiple eigenvalue,
 one chose arbitrarily the vectors of an orthogonal base of the corresponding eigenspace)
make an angle less than $ \mathcal{O}\left(  \frac{ \rho}{R} \right)$,
with either $T$, either $N$. Since no more than $d$ pairwise orthogonal vectors can make a small angle with the $d$-dimensional space $T$,
and the same for the $N-d$-dimensional space $T$, we know that $d$ eigenvectors make an angle $\mathcal{O}\left(  \frac{ \rho}{R} \right)$
with $T$ and $N-d$ with $N$.
Multiplying on the left \eqref{equation:LambdaCosThetaV_T} by $v_T^t$ and  \eqref{equation:LambdaSinThetaV_N} by $v_N^t$  we get:
\begin{align*}
 (\cos \theta)  v_T^t \mathbf{A}_{TT}  v_T &+  (\sin \theta) v_T^t \mathbf{A}_{TN}   v_N = \lambda (\cos \theta)\\
 (\cos \theta*)  v_N^t \mathbf{A}_{TN}^t  v_T &+  (\sin \theta) v_N^t \mathbf{A}_{NN}  v_N = \lambda (\sin \theta) 
\end{align*}
When the angle between the eigenvector $v$ and $T$ is in $\mathcal{O}\left(  \frac{ \rho}{R} \right)$, 
then $|1 - \cos \theta| = \mathcal{O}\left( \left( \frac{ \rho}{R} \right)^2\right)$ 
and $|\sin \theta | = \mathcal{O}\left(  \frac{ \rho}{R} \right)$,
  and  the first equation gives that $\lambda$ is close to 
$ v_T^t \mathbf{A}_{TT}  v_T  \geq  C_{TT} \rho^2$.
When    the angle between the eigenvector $v$ and $N$ is in $\mathcal{O}\left(  \frac{ \rho}{R} \right)$, 
the second equation gives that $\lambda =\mathcal{O}\left( \left( \frac{ \rho}{R} \right)^2\right)$,
which is smaller than $C_{TT} \rho^2$ for $\frac{ \rho}{R}$ small enough. 

We have so far proven that the $d$ orthonormal eigenvectors $v_1,\ldots, v_d$ corresponding to the $d$ largest eigenvalues of $\mathbf{A}$
make an angle upper bounded by $C \left(\frac{ \rho}{R} \right)$ with $\Tangent {\pi_{\M}(p)} \M$, for some constant $C$ that depends only on $d$ and $\eta$.
For any unit vector  $u$, its angle with $\Tangent {\pi_{\M}(p)} \M$ satisfies:
\[
\sin \angle u, \Tangent {\pi_{\M}(p)} \M = \| u - \pi_{\Tangent {\pi_{\M}(p)} \M} (u)\|
\]
Now if $u= \sum_{i=1,d} a_i v_i$ is a unit vector in the $d$-space spanned by $v_1,\ldots, v_d$, then:
\begin{align*}
\sin \angle u, \Tangent {\pi_{\M}(p)} \M &= \left\|  \sum_{i=1,d} a_i v_i - \pi_{\Tangent {\pi_{\M}(p)} \M} \left( \sum_{i=1,d} a_i v_i\right) \right\| \\
& =    \left\|  \sum_{i=1,d} a_i  \left(v_i - \pi_{\Tangent {\pi_{\M}(p)} \M} (  v_i)  \right)\right\| \\
& \leq \sum_{i=1,d} | a_i |  \left\| v_i - \pi_{\Tangent {\pi_{\M}(p)} \M} (  v_i) \right\|  \\
& \leq  \sum_{i=1,d} | a_i | C \left(\frac{ \rho}{R} \right)   \leq \sqrt{d} \,C \left(\frac{ \rho}{R} \right)
\end{align*}
since $\sum_{i=1,d} a_i^2 = 1 \Rightarrow \sum_{i=1,d} |a_i | \leq \sqrt{d}$.
\end{proof}


\subsection{Perturbation}

In the context of Lemma \ref{lemma:PerturbationAndLLL_pseudoTangent},
we define the \textbf {random perturbation}  $f(P)$ of $P$ of amplitude $r_{\operatorname{pert.}}>0$ as follows.
 For each point $p\in P$,
$f(p)$ is drawn, independently,  uniformly in the ball $V_p \cap B(p, r_{\operatorname{pert.}})$.

Since $d_H(P, f(P)) \leq r_{\operatorname{pert.}}$, we can take $\varepsilon' = \varepsilon + r_{\operatorname{pert.}}$ to guarantee $\M \subset \Offset {f(P)} {\varepsilon'}$.
In order to guarantee a lower bound on $\Sep{f(P)} > \eta' \varepsilon' > 0$,  we require:
\begin{equation}\label{equation:relationRPerEtaPrime}
2 r_{\operatorname{pert.}} \leq \eta \varepsilon - \eta' \varepsilon' 
\end{equation}
In fact it will be convenient to assume that:
\begin{equation}\label{equation:MajorationRPerEtaEpsilon}
 r_{\operatorname{pert.}} \leq \frac{ \eta \varepsilon}{20}
\end{equation}

 and  since we assume $\varepsilon \leq \frac{\rho}{16}$ in the context of Lemma \ref{lemma:PerturbationAndLLL_pseudoTangent} we have:
\begin{equation}\label{equation:UpperBoundOnRPer}
r_{\operatorname{pert.}} \leq \frac{\rho}{32}
\end{equation}

Denotes by $\Sigma_d(p)$  the set of $d$-simplices  in $P \cap B(p, \rho)$:
\[
\Sigma_d(p) \defunder{=} \{ \sigma \subset P \cap B(p, \rho) \mid \sharp \sigma = d+1 \}
\]
Then, for $h, \zeta>0$, and $p\in P$,
we consider the  event $E^{\operatorname{good}}_p(h, \zeta)$ as
the set of possible perturbations $f$ such that:
(1) for any simplex $\sigma \in \Sigma_d(p)$, $\pi_{V_p} (f(\sigma))$  has minimal height greater than $h$, and,
(2) for any simplex $\sigma \in \Sigma_d(p)$, $\pi_{V_p} (f(\sigma))$ is $\zeta$-protected in  $\pi_{V_p} (f(P \cap B(p, \rho)))$.

Our goal is to find a perturbation $f$ that belongs to $E^{\operatorname{good}}_p(h, \zeta)$ for any $p$:
\begin{equation}\label{equation:GoodPerturbation}
f \in \bigcap_{p\in P} E^{\operatorname{good}}_p(h, \zeta)
\end{equation}
where:
\begin{equation}\label{equation:UpperBoundOnZeta}
\zeta > 2 A \left( 1 + \frac{4d\varepsilon}{h}  \right) 
\end{equation}
with:
\begin{equation}\label{equation:RelationAwrtThetaM}
 A = 4\delta(C \theta_m) + 4  \rho( C \theta_m)^2
\end{equation}
where $C\geq 1$ is a constant to choose to meet your need.


and $\theta_m$ is an upper bound $\Theta(P,\rho)$:
\begin{equation}\label{equation:DefinitionAngleThetaM}
\theta_m \geq \Theta(P,\rho) 
\end{equation}

since for any simplex with vertices in $P$ on has  $\sigma \subset \Offset \M \delta$, one has:
\[
f(\sigma) \subset \Offset \M {\delta + \Xi_0(\eta,d)\frac{\rho}{R}  r_{\operatorname{pert.}} + \frac{r_{\operatorname{pert.}}^2}{R}}
\]
So that, with the assumption made in Lemma \ref{lemma:PerturbationAndLLL_pseudoTangent}
that $\delta < \frac{\rho^2}{4 R}$, we have:
\begin{equation}\label{equation:perturbedPointIn2DeltaThickening}
r_{\operatorname{pert.}} \leq \frac{\rho }{8 \max (\Xi_0(\eta,d), 1) } \Rightarrow f(\sigma) \subset \Offset \M {2 \delta}
\end{equation}
It is now possible to gives an upper bound $\theta_m$ defined in \eqref{equation:DefinitionAngleThetaM} relying on a lower bound on
simplices maximal height $h$.

Let us assume that:
\begin{equation}\label{equation:LowerBound1OnH}
h >  20 \delta
\end{equation}
If $h$ is the smallest height of $\pi_{V_p} (f(\sigma))$ for some $p \in P$ and $\sigma \in \Sigma_d(p)$,
then, the smallest heigh of $f(\sigma)$, before projection on $V_p$, is at least $h$ as this projection cannot increase distances.

Assuming the requested upper bound on $ r_{\operatorname{pert.}}$, 
we know from \eqref{equation:perturbedPointIn2DeltaThickening}
  that $f(\sigma) \subset \Offset \M {2 \delta}$, and we have with \eqref{equation:LowerBound1OnH} that the height of 
$\pi_\M (f(\sigma))$ is at least $h - 4 \delta > \frac{4}{5} h$. 
Since $L( \pi_\M (f(\sigma)) < 2 \rho + 2 r_{\operatorname{pert.}} + 4 \delta < 3 \rho$, 
we can apply Corollary \ref{lemma:angle-span-simplex-tangent-plane},
 where the bound $2\rho$ on the diameter of the simplex is replaced by $3 \rho$, allowing us to chose for angle 
$\theta_m$ satisfying \eqref{equation:DefinitionAngleThetaM}:
\begin{align*}
\theta_m &= \arcsin \left( \frac{2d}{ \frac{4}{5} h} \left( \frac{( 3 \rho )^2}{\mathcal{R}} + 2 \delta \right) \right) \\
 &\leq \frac{\pi}{2}  \left( \frac{2d}{ \frac{4}{5} h} \left( \frac{( 3 \rho )^2}{\mathcal{R}} + 2 \delta \right) \right) 
<   \frac{\pi}{2}  \left( \frac{2d}{ \frac{4}{5} h} \left( \frac{( 3 \rho )^2}{\mathcal{R}} + 2 \frac{\rho^2}{4 \mathcal{R}} \right) \right) = \frac{95 \pi}{8} \frac{d \rho^2}{h \mathcal{R}}
\end{align*}

so that:

\begin{equation}\label{equation:UpperBoundOnThetaMAsFctOfLowerBoundOnH}
\theta_m < 38 \frac{d \rho^2}{h \mathcal{R}}
\end{equation}
Substituting this in \eqref{equation:RelationAwrtThetaM} we get (using $C\geq 1 \Rightarrow C^2 \geq C$):

\begin{align*}
 A &\leq   4 C^2(\delta \theta_m +   \rho\theta_m^2 ) \\
 & \leq 4 C^2 \left(  \frac{\rho^2}{4 R} \theta_m +   \rho\theta_m^2 \right) \\
 & < 4 C^2 \theta_m  \left(  \frac{\rho}{4 R} +   38 d  \frac{\rho}{h} \frac{\rho}{R}  \right) \rho \\
 & < 4 C^2  15  \frac{\rho}{h} \frac{\rho}{R}   \left(  \frac{\rho}{4 R} +    38 d   \frac{\rho}{h} \frac{\rho}{R}  \right) \rho \\
 & < 60  C^2   \frac{\rho}{h} \left( \frac{\rho}{R} \right)^2   \left(  \frac{1}{4} +    38 d   \frac{\rho}{h} \right) \rho  \\
 & < 60  C^2   \frac{\rho}{h} \left( \frac{\rho}{R} \right)^2   \left(   39 d  \frac{\rho}{h} \right) \rho \\
 & < 2500  C^2  d \left( \frac{\rho}{R} \right)^2   \left(\frac{\rho}{h}\right)^2 \rho 
\end{align*}

It follows that \eqref{equation:UpperBoundOnZeta} is satisfied if:
\[
\zeta > 5000   C^2 d \left( \frac{\rho}{R} \right)^2   \left(\frac{\rho}{h}\right)^2 \rho  \left( 1 + \frac{ d \rho}{h}  \right) 
\]
A stronger but still sufficient condition to guarantee  \eqref{equation:UpperBoundOnZeta} is, since $1 + \frac{ d \rho}{h} < 2 \frac{ d \rho}{h}$,
is to set the required protection to be:
\begin{equation}\label{equation:UpperBoundOnZetaFctOfH}
\zeta \geq 10^4  C^2  d^2 \left( \frac{\rho}{R} \right)^2   \left(\frac{\rho}{h}\right)^3 \rho  
\end{equation}
Which is optimal up to a multiplicative constant.

One can already see on \eqref{equation:UpperBoundOnZetaFctOfH} that, by assuming $ \frac{\rho}{R} $ small enough, the condition
can be made arbitrarily weak, which is still not a proof but is a good omen for the existence of perturbations satisfying it.
One can see also that we cannot require $\zeta$ and $h$ to be simultaneously arbitrarily small for a given value of $ \frac{\rho}{R}$.
In order to quantify the optimal tradeoff between $\zeta$ and $h$ lower bounds, we need first to evaluate 
the  probabilities of  events related to these protection and fatness constraints.

Denote by $E^{\operatorname{fat}}_p(h)$ the  event which is the set of perturbations $f$ such that at least one perturbed simplex 
has a height smaller than $h$:
\[
E^{\operatorname{fat}}_p(h) \defunder{=} \left\{ f, \exists \sigma \in  \Sigma_d(p) , \textrm{height}(\pi_{V_p} (f(\sigma))) \leq h \right\}
\]
where, for a $d$-simplex $\sigma'$, $ \textrm{height}(\sigma')$ is its minimal height.

Denote by $E^{\operatorname{protect}}_p(\zeta)$ the  event which is the set of perturbations $f$ such that  at least one perturbed simplex  
is not $\zeta$-protected in $\pi_{V_p} (f(\sigma))$:
\[
E^{\operatorname{protect}}_p(\zeta) \defunder{=}
 \left\{ f, \exists \sigma \in  \Sigma_d , \pi_{V_p} (f(\sigma))\quad  
 \textrm{is not} \quad \zeta\textrm{-protected} \quad \textrm{in} \quad \pi_{V_p} (f( P \cap B(p, \rho))) \right\}
\]
One has:
\[
\neg E^{\operatorname{good}}_p(h, \zeta) = E^{\operatorname{fat}}_p(h) \lor E^{\operatorname{protect}}_p(\zeta)
\]
so that:
\begin{equation}\label{equation:ProbabiltyUpperBoundedBySum1}
\mathbf{P}( \neg E^{\operatorname{good}}_p(h, \zeta) ) \leq  \mathbf{P}( E^{\operatorname{fat}}_p(h) ) + \mathbf{P} (  E^{\operatorname{protect}}_p(\zeta) )
\end{equation}

In order to be able to apply the perturbation algorithm \cite{moser2010constructive} 
we have to derive an upper bound on \eqref{equation:ProbabiltyUpperBoundedBySum1}.
We start by $ P( E^{\operatorname{fat}}_p(h) )$.

{\bf Upper bound on  $\mathbf{P}( E^{\operatorname{fat}}_p (h) )$}

Denote by $N_p$ the number of points in $P \cap B(p, \rho)$.
\begin{proposition}\label{equation:BoundOnNearByPoints}
In the context of Lemma \ref{lemma:PerturbationAndLLL_pseudoTangent}, one has:
\[
N_p <  \left(  \frac{4 \rho}{\eta \varepsilon}\right)^d
\]
\end{proposition}


\begin{proof}
Since $\delta << \eta \varepsilon$, the angle:
\[
\max_{{q_1,q_2 \in P \cap B(p, \rho)} \atop {q_1 \neq q_2}} \angle q_1q_2, V_p
\]
can be upper bounded, by, say, $\pi/3$, so that its cosine is lower bounded by $1/2$.
Therefore the  projection of points in $P \cap B(p, \rho)$  on $V_p$ remain at pairwise  distances at least $\frac{1}{2}\eta \varepsilon$.
The balls $V_p \cap B(\pi_{V_p}(q), \frac{1}{4} \eta \varepsilon)$ are disjoint and included in $V_p \cap B(\pi_{V_p}(q), \rho)$ which gives the upper bound of the lemma.
\end{proof}


Given $q_0 \in P \cap B(p, \rho)$ and a $(d-1)-simplex$ $\{q_1,\ldots, q_d\} \subset P \cap B(p, \rho)$, denote by $E^{\operatorname{fat}}_p(h,q_0, \{q_1,\ldots q_d\})$ the event made of all $f$
such that $d ( \pi_{V_p} (f(q_0)), [ \pi_{V_p}( f(\{q_1, \ldots , q_d\} ))] ) \leq h$.
 Where $[\ldots]$ denotes the hyperplane,  affine hull of the $d$ points  (generic with probability $1$).

The event  $E^{\operatorname{fat}}_p(h)$ is the union of all such events whose number is   $N_p {N_p -1 \choose{d}}$.


The probability of $E^{\operatorname{fat}}_p(h,q_0, \{q_1,\ldots q_d\})$   can be upper bounded by a uniform upper bound on 
conditional probabilities. A given sample $f_0(P \setminus \{q_0\})$ of all other points, defines the condition $\forall q \in P \setminus \{q_0\} ,f(q) = f_0(q)$,
 under which we can consider  the conditional probability of  the event $E^{\operatorname{fat}}_p(h,q_0, \{q_1,\ldots q_d\})$ and we have:
\[
 \mathbf{P}( E^{\operatorname{fat}}_p(h,q_0, \{q_1,\ldots q_d\}) )  \leq \sup_{f_0} \mathbf{P}( E^{\operatorname{fat}}_p(h,q_0, \{q_1,\ldots q_d\}) \mid \forall q \in P \setminus \{q_0\} ,f(q) = f_0(q) )
\]
This conditional probability is easy to upper bound. Indeed, since all points beside $q_0$ have a given position, and since the projection of $f(q_0)$
on $V_p$, obey a uniform law (as the Jacobian of the projection from a $d$-flat to a $d$-flat is constant)  inside the projection $\pi_{V_p} ( B(0, r_{\operatorname{pert.}}) \cap V_{q_0})$,
this conditional probability can be estimated as a ratio of two $d$-volumes.

If we denote by $\mathcal{V}_k$ the $k$-volume of the euclidean ball with radius $1$, then:
\[
\mathbf{Vol} \Big( \pi_{V_p} ( B(q_0, r_{\operatorname{pert.}}) \cap V_{q_0} ) \Big) \geq \alpha_d \left( r_{\operatorname{pert.}}  \cos \angle V_p, V_{q_0}\right)^d  
\]
Also, for a given $(d-1)$-flat $[ \pi_{V_p}( f(\{q_1, \ldots , q_d\} ))] $ in $V_p$ , one can upper bound the area of the subset of $\pi_{V_p} ( B(0, r_{\operatorname{pert.}}) \cap V_{q_0} )$
at distance at most $h$ from $[ \pi_{V_p}( f(\{q_1, \ldots , q_d\} ))] $ by:
\[
\mathbf{Vol} \Big(\Offset  {[ \pi_{V_p}( f(\{q_1, \ldots , q_d\} ))]} {h} \cap \pi_{V_p} ( B(q_0, r_{\operatorname{pert.}}) \cap V_{q_0} ) \Big) \leq 2h \alpha_{d-1} r_{\operatorname{pert.}}^{d-1}
\]
we get:
\begin{align*}
&\mathbf{P}( E^{\operatorname{fat}}_p(h,q_0, \{q_1,\ldots q_d\}) \mid \forall q \in P \setminus \{q_0\} ,f(q) = f_0(q) ) \\
& \quad = \frac{ \mathbf{Vol} \Big(\Offset  {[ \pi_{V_p}( f(\{q_1, \ldots , q_d\} ))]} {h} \cap \pi_{V_p} ( B(q_0, r_{\operatorname{pert.}}) \cap V_{q_0} ) \Big)}{ \mathbf{Vol} \Big( \pi_{V_p} ( B(q_0, r_{\operatorname{pert.}}) \cap V_{q_0} ) \Big)}
  \leq  \frac{2 \alpha_{d-1}}{\alpha_d \left( \cos \angle V_p, V_{q_0}\right)^d} \frac{h}{r_{\operatorname{pert.}}}
\end{align*}
With the bound of lemma \ref{lemma:PerturbationAndLLL_pseudoTangent} we have get:
\begin{align*}
 \angle V_p, V_{q_0} &\leq  \angle V_p , \Tangent {\pi_\M(p)} \M + \angle \Tangent {\pi_\M(p)} \M , \Tangent {\pi_\M(q_0)} \M + \angle \Tangent {\pi_\M(q_0)} \M, V_{q_0} \\
 & <  (2 \Xi_0(\eta, d) + 2) \frac{\rho}{R}
\end{align*}

So that, assuming:
\begin{equation}\label{equation:AssumtionRLArgeEnoughWRTXid}
\frac{\rho}{R} < \frac{1}{2 \Xi_0(\eta, d) + 2}  \arccos \left( \frac{1}{2}\right)^{1/d}
\end{equation}
One has, since $t \mapsto (\cos t)^d$ is decreasing:
\[
 \left( \cos \angle V_p, V_{q_0}\right)^d \geq \cos \left( (2 \Xi_0(\eta, d) + 2) \frac{\rho}{R}  \right)^d  > \frac{1}{2}
\]
So that:
\begin{equation}\label{equation:ProbabiltyUpperBoundOnFatnessEvent}
 \mathbf{P}  \Big( E^{\operatorname{fat}}_p ( h)  \Big) < N_p {N_p -1 \choose {d}} \frac{4 \alpha_{d-1}}{\alpha_d} \:  \frac{h}{r_{\operatorname{pert.}}}
\end{equation}

{\bf Upper bound on  $\mathbf{P}( E^{\operatorname{protect}}_p )$}
The computation is the same. The number of corresponding individual events for a given $q_0 \in P \cap B(p, \rho)$
and a $d$-simplex $\sigma \in \Sigma_d(p)$ is now $ N_p {N_p -1 \choose {d+1}}$.

The $d$-volume of the intersection of the $\zeta$-offset  of a $(d-1)$-sphere with radius at least $\eta \varepsilon - 2 r_{\operatorname{pert.}}$ 
with $\pi_{V_p} ( B(q_0, r_{\operatorname{pert.}}) \cap V_{q_0} )$
can be  upper bounded as follows.

The radius of the $d-1$-circumsphere $\tilde{S}$ of $\pi_{V_p}(f(q_1),\ldots,f(q_{d+1}) )$ is,
  thanks to \ref{equation:MajorationRPerEtaEpsilon}, at least $\eta \varepsilon - 2 r_{\operatorname{pert.}} > \frac{9}{10} \eta \varepsilon$.

Since $\pi_{V_p} ( B(q_0, r_{\operatorname{pert.}}) \cap V_{q_0} ) \subset B( \pi_{V_p} (q_0), r_{\operatorname{pert.}}) \cap V_p $,
it is enough to bound the volume of the intersection of  $\Offset {\tilde{S}} {\zeta}$
with  $B( \pi_{V_p} (q_0), r_{\operatorname{pert.}}) \cap V_p$. This set is included in the set of
points in  $\Offset {\tilde{S}} {\zeta}$ whose closest point on $\tilde{S}$ is inside the ball 
$B( \pi_{V_p} (q_0), r_{\operatorname{pert.}} + \zeta) \cap V_{q_0} )$. The $d$-volume of this last set can be upper bounded by the
$(d-1)$-volume of the outer shell times $2 \zeta$, 
in other words, if $\tilde{r}$ is the radius of $\tilde{S}$,
and $\tilde{a}$ is the $(d-1)$-volume (area) of the spherical cap $\tilde{C}$ defined as:
\[
\tilde{C} \defunder{=} \tilde{S} \cap B_{V_p}( \pi_{V_p} (q_0), r_{\operatorname{pert.}} + \zeta)
\]
We can bound our volume by:
\[
2 \zeta \left( \frac{\tilde{r} + \zeta} { \tilde{r}}\right)^{d-1}  \tilde{a}
\]
as, here, $\left( \frac{\tilde{r} + \zeta} { \tilde{r}}\right)^{d-1} $ is the ratio between 
the area of $\tilde{C}$ and the area of the corresponding outer shell of the $\zeta$-offset.

Now the ratio between $\tilde{a}$ and the $(d-1)$-volume of the $(d-1)$-disk $\tilde{D}$
subset of $B_{V_p}( \pi_{V_p} (q_0), r_{\operatorname{pert.}} + \zeta)$, 
with same boundary as $\tilde{C}$ is upper bounded by $\frac{d \alpha_d}{2 \alpha_{d-1}}$, 
where $\frac{d \alpha_d}{2}$ is the $(d-1)$-volume of the half $(n-1)$-sphere bounding the unit $n$-ball,
and $\alpha_{d-1}$ the $(d-1)$-volume of the unit of $(d-1)$-ball with the same boundary which is the equator of the the unit $n$-ball.
This ratio can be made as near as $1$ as  wanted if the ratio $r_{\operatorname{pert.}}/\tilde{r}$ is assumed small enough, but,
since we don't care too much about constants, we keep the ratio $\frac{d \alpha_d}{2 \alpha_{d-1}}$ so that we get:
\[
\tilde{a} < \frac{n \alpha_d}{2 \alpha_{d-1}} \alpha_{d-1} (r_{\operatorname{pert.}} + \zeta)^{d-1} =  \frac{d \alpha_d}{2 } (r_{\operatorname{pert.}} + \zeta)^{d-1} 
\]
which gives:
\[
 \mathbf{P}  \Big( E^{\operatorname{protect}}_p ( \zeta)  \Big) < N_p {N_p -1 \choose {d+1}} \frac{ 2 \zeta \left( \frac{\tilde{r} + \zeta} { \tilde{r}}\right)^{d-1} \frac{d \alpha_d}{2 } (r_{\operatorname{pert.}} + \zeta)^{d-1} } {\alpha_d \left( r_{\operatorname{pert.}}  \cos \angle V_p, V_{q_0}\right)^d  }
\]
Observe that since $(1+ 1/n)^n < e$, assuming:
\begin{equation}\label{equation:LowerBondOnZeta}
 \zeta < \frac{r_{\operatorname{pert.}}}{d-1} <  \frac{\tilde{r}}{d-1} 
\end{equation}
one has: $\left( \frac{\tilde{r} + \zeta} { \tilde{r}}\right)^{d-1} < e$ and $(r_{\operatorname{pert.}} + \zeta)^{d-1} < e r_{\operatorname{pert.}}^{d-1}$, so that,
assuming \eqref{equation:AssumtionRLArgeEnoughWRTXid} we get:
\begin{equation}\label{equation:ProbabiltyUpperBoundOnPotectionEvent}
 \mathbf{P}  \Big( E^{\operatorname{protect}}_p ( \zeta)  \Big) < N_p {N_p -1 \choose {d+1}}  2 d e^2 \frac{\zeta}{r_{\operatorname{pert.}}}
\end{equation}
We have now write an explicit upper bound on equation \ref{equation:ProbabiltyUpperBoundedBySum1}:
\begin{equation}\label{equation:ProbabiltyUpperBoundedBySum2}
\mathbf{P}( \neg E^{\operatorname{good}}_p(h, \zeta) ) \leq   N_p {N_p -1 \choose {d}} \frac{4 \alpha_{d-1}}{\alpha_d} \:  \frac{h}{r_{\operatorname{pert.}}} + N_p {N_p -1 \choose {d+1}}  2 d e^2 \: \frac{\zeta}{r_{\operatorname{pert.}}}
\end{equation}
Let us denote respectively by $c_{\operatorname{fat}}$ and $c_{\operatorname{protect}}$ the respective coeffients in front of $\frac{h}{r_{\operatorname{pert.}}} $ and $\frac{\zeta}{r_{\operatorname{pert.}}}$ is the above upper bound:
\begin{align*}
c_{\operatorname{fat}} &\defunder{=}   N_p {N_p -1 \choose {d}} \frac{4 \alpha_{d-1}}{\alpha_d} \\
c_{\operatorname{protect}} &\defunder{=}   N_p {N_p -1 \choose {d+1}}  2 d e^2 
\end{align*}
With this simplified notation, we can substitute in \eqref{equation:ProbabiltyUpperBoundedBySum2}  the smallest possible value
  \eqref{equation:UpperBoundOnZetaFctOfH} of $\zeta$ and we get:
\begin{equation}\label{equation:expressionIntermNotPOfEgood}
\mathbf{P}( \neg E^{\operatorname{good}}_p(h, \zeta) ) \leq   c_{\operatorname{fat}} \:  \frac{h}{r_{\operatorname{pert.}}} + c_{\operatorname{protect}} \:  \frac{10^4 d^2  C^2  \left( \frac{\rho}{R} \right)^2   \left(\frac{\rho}{h}\right)^3 \rho}{r_{\operatorname{pert.}}}
\end{equation}
taking:
\[
h_{\min}^4 = \frac{3 c_{\operatorname{protect}}  10^4 d^2  C^2  \left( \frac{\rho}{R} \right)^2  \rho^4 }{ c_{\operatorname{fat}} }
\]
As a function of $h$, the right hand term  of \eqref{equation:expressionIntermNotPOfEgood}
 is decreasing for $h< h_{\min}$ and increasing for $h>h_{\min}$.

Setting $h= h_{\min}$, \eqref{equation:expressionIntermNotPOfEgood} and defining the constant $c_\star$, which, 
as $c_{\operatorname{fat}}$ and $c_{\operatorname{protect}}$, depends  on $d$ and $\eta$ only, as:
\[
c_{\star} \defunder{=} \left( 10^4 \frac{c_{\operatorname{protect}}}{c_{\operatorname{fat}}}  d^2 C^2 \right)^{\frac{1}{4}},
\]
one has:
\begin{equation}\label{equation:OptimalValueForH}
h= h_{\min} =  3^{\frac{1}{4}} \:  c_{\star} \left( \frac{\rho}{R} \right)^{\frac{1}{2}}  \rho
\end{equation}
 and we get, 
by substituting in \eqref{equation:expressionIntermNotPOfEgood}:  
\begin{equation}\label{equation:expressionIntermNotPOfEgood2}
\mathbf{P}( \neg E^{\operatorname{good}}_p(h, \zeta) ) \leq  (3^{\frac{1}{4}} + 3^{-\frac{3}{4}})  c_{\operatorname{fat}} c_{\star}\left( \frac{\rho}{R} \right)^{\frac{1}{2}}  \frac{\rho}{r_{\operatorname{pert.}}} 
< 2 c_{\operatorname{fat}} c_{\star}\left( \frac{\rho}{R} \right)^{\frac{1}{2}}  \frac{\rho}{r_{\operatorname{pert.}}}
\end{equation}
In order to apply  Lov{\'a}sz local lemma we need to upper bound the number of events $  \neg E^{\operatorname{good}}_q(h, \zeta) , q \in P$
which are nit independent of $ \neg E^{\operatorname{good}}_p(h, \zeta)$. Since, as soon as $\| p-q \| > 2 \rho$, $P\cap B(p, \rho)$ and 
 $P\cap B(q, \rho)$ being disjoint, the events are independent. The number of dependant event is then bounded by,
 following the same argument as for proposition \ref{equation:BoundOnNearByPoints} :
 \[
 N_{indep} +1 \leq  \left(  \frac{8 \rho}{\eta \varepsilon}\right)^d
\]
So that Lov{\'a}sz local Lemma  may apply if:
\[
e \mathbf{P}( \neg E^{\operatorname{good}}_p(h, \zeta) )  (N_{indep} +1) < 1
\]
that is if:
\[
e \left(  \frac{8 \rho}{\eta \varepsilon}\right)^d 2 c_{\operatorname{fat}} c_{\star}\left( \frac{\rho}{R} \right)^{\frac{1}{2}}  \frac{\rho}{r_{\operatorname{pert.}}} < 1
\]
Since we have assumed $ r_{\operatorname{pert.}} \leq \frac{ \eta \varepsilon}{20}$ in \eqref{equation:MajorationRPerEtaEpsilon} and 
$r_{\operatorname{pert.}} \leq \frac{\rho }{8 \max (\Xi_0(\eta,d), 1) }$ in \eqref{equation:perturbedPointIn2DeltaThickening}, we see that
the required perturbation is possible if:
\begin{equation}\label{equation:ReachCondtionForPerturbation}
\left( \frac{\rho}{R} \right)^{\frac{1}{2}}  < \min \left( \frac{  \eta \varepsilon} { 20 }, 
 \frac{  \rho} { 8 \max (\Xi_0(\eta,d), 1) } \right) \frac{1}{  e  \rho  \left(\frac{8 \rho}{\eta \varepsilon}\right)^d 2 c_{\operatorname{fat}} c_{\star}}
\end{equation}
Since the right hand side depends only on $\eta,\frac{\epsilon}{\rho}$ and $d$, 
and since $\frac{\epsilon}{\rho}$ can be chosen equal to $\frac{1}{16}$,
assuming  $\frac{\rho}{R}$ small enough enforces the inequality to hold.
When it holds, we can take:
\begin{align*}
r_{\operatorname{pert.}} &= \min \left( \frac{ \eta \varepsilon}{20},  \frac{\rho }{8 \max (\Xi_0(\eta,d), 1) } \right) \\
h &= 3^{1/4} c_{\star} \left( \frac{\rho}{R} \right)^{\frac{1}{2}} \rho \\
\zeta &= 3^{-3/4} \: c_{\star}\frac{c_{\operatorname{fat}}}{c_{\operatorname{protect}}} \left( \frac{\rho}{R} \right)^{\frac{1}{2}}  \rho 
\end{align*}

We can give now the perturbation lemma, that refers to Moser Tardos Algorithm, see
 Algorithm 4 and Theorem 5.22 in \cite[Section 5.3.4]{boissonnat2018geometric}.
 For convenience of use, the radius $\rho$ is denoted $\tilde{\rho}$ in the 
the lemma setting.

\begin{lemma}\label{lemma:PerturbationLemma}
Let $C\geq 1$, $\delta \geq 0$, $0 < \varepsilon \leq \frac{\tilde{\rho}}{16}$, $10 \tilde{\rho} < \reach \leq  \Reach{\M}$ and suppose that $P \subseteq \Offset \M \delta$ 
for $\delta = \frac{\tilde{\rho}^2}{4 R}$, and 
$\M \subseteq \Offset P \varepsilon$ and $\Sep{P} > \eta \varepsilon$ for $\eta >0$.
For each $p \in P$ defines $V_p$ as in Lemma \ref{lemma:PerturbationAndLLL_pseudoTangent}.
Then, given  fixed constants $\eta$,$\frac{\varepsilon}{\tilde{\rho}}$ and $d$, for  $\frac{\tilde{\rho}}{\reach}$ small enough, 
\eqref{equation:ReachCondtionForPerturbation} holds, so that 
Moser Tardos algorithm produces a perturbation $f(P)$ of $P$ such that $f(p) \in V_p \cap B(p, r_{\operatorname{pert.}})$
and such that $\M \subset   \Offset {f(P)} {\varepsilon'}$ and $f(P) \subseteq \Offset \M {\delta'}$ with $\delta' = 2 \delta$,  $\varepsilon' = \varepsilon + r_{\operatorname{pert.}}$
and $\Sep{f(P)} > \eta' \varepsilon'$, for $\eta' = \eta \frac{ 9 \varepsilon}{ 10 (\varepsilon +  r_{\operatorname{pert.}})}$.

Moreover the perturbed cloud $f(P)$ is $\zeta$ protected at scale $\tilde{\rho}$ and:

\begin{align}
r_{\operatorname{pert.}} &= \min \left( \frac{ \eta \varepsilon}{20},  \frac{\tilde{\rho} }{8 \max (\Xi_0(\eta,d), 1) } \right) \\
\Height{f(P)}{\tilde{\rho}} \geq h &= 3^{1/4} c_{\star} \left( \frac{\tilde{\rho}}{\reach} \right)^{\frac{1}{2}} \tilde{\rho} \\
\Protection{f(P)}{\tilde{\rho}} \geq  \zeta &=  3^{-3/4} \: c_{\star}\frac{c_{\operatorname{fat}}}{c_{\operatorname{protect}}} \left( \frac{\tilde{\rho}}{\reach} \right)^{\frac{1}{2}}  \tilde{\rho} \\
\zeta &> 2 A \left( 1 + \frac{4d\varepsilon}{h}  \right) \\
 A &= 4\delta(C \theta_m) + 4  \tilde{\rho}( C \theta_m)^2 \label{equation:EquationDefA}
\end{align}
for some $\theta_m$ such that $\theta_m > 4 \frac{\tilde{\rho}}{\reach} $  and $\theta_m \geq \Theta(P,\tilde{\rho})$.
\end{lemma}

\subsection{Proof of Lemma \ref{lemma:PerturbationTunedForTheorem} }\label{section:ProofOfLemmaPerturbationTunedForTheorem}

Lemma \ref{lemma:PerturbationTunedForTheorem} is a  corollary of Lemma \ref{lemma:PerturbationLemma}, applied with adapted parameters.
The radius $\tilde{\rho}$ of  Lemma \ref{lemma:PerturbationLemma} is $\tilde{\rho} = 3 \rho$,
where $\rho$ is the radius of Lemma \ref{lemma:PerturbationTunedForTheorem} 
and Theorem \ref{theorem:homeomorphism-from-sampling-conditions}.
 The angle $\theta_m$ of  Lemma \ref{lemma:PerturbationLemma} gives the $\frac{1}{3}  \theta$,
 where $\theta$ is the angle in  condition \refhypo{hypo:safety-angle} of 
Theorem \ref{theorem:homeomorphism-from-sampling-conditions}.
The constant $C$ of  \ref{lemma:PerturbationLemma}  is set to $3$.
So that, with 
the value  $C=3$ in \eqref{equation:EquationDefA}
the angle  $\theta = C \theta_m =3 \theta_m$  gives us the angle $\theta$ of   condition \refhypo{hypo:safety-angle} in 
Theorem \ref{theorem:homeomorphism-from-sampling-conditions}.
For $\frac{\varepsilon}{\reach}$   small enough then Lemma  \ref{lemma:PerturbationLemma} applies.
In particular $\theta_m > 4 \frac{\rho}{R} $  and $\theta_m \geq \Theta(P,3\rho) \geq \Theta(P,\rho) $.
Then, since $\theta_m > 4 \frac{\rho}{R} \Rightarrow  \theta_m > \arcsin  \frac{2\rho}{R} >  \arcsin  \frac{\rho+ \delta}{R}$,
we get that:
\[
 \Theta(P,\rho)  \leq 3 \theta_m - 2 \arcsin  \frac{\rho+ \delta}{R}
 \]
and  condition \refhypo{hypo:safety-angle} of  Theorem \ref{theorem:homeomorphism-from-sampling-conditions}  is satisfied with $\theta = 3 \theta_m$.
Conditions \refhypo{hypo:safety-protection} is satisfied as well, and, for   $\frac{\varepsilon}{\mathcal{R}}$  small enough,
we see that  condition \refhypo{hypo:safety-separation} is satisfied as well.
%
%

\end{document}